\documentclass[12pt,reqno]{amsart}

\usepackage{xcolor}
\definecolor{brickred}{rgb}{0.8, 0.25, 0.33}

\usepackage[normalem]{ulem}

\usepackage[a4paper,top=3cm,bottom=3cm,
           left=3cm,right=3cm,marginparwidth=60pt]{geometry}
\usepackage[english]{babel} 
\usepackage[utf8]{inputenc}
\usepackage[T1]{fontenc} 
\usepackage{color}
\usepackage{braket} 
\usepackage{physics}
\usepackage{tabularx} 
\usepackage{booktabs} 
\usepackage{rotating}
\usepackage{multirow} 
\usepackage{comment} 
\usepackage{tikz}
\usetikzlibrary{math,shapes.callouts} 
\usetikzlibrary{3d} 
\usepackage{tikz-3dplot}
\usepackage{tikz-cd} 
\usepackage{microtype}
\usepackage{verbatimbox}
\usepackage[sort,compress]{cite}
\usepackage{float}

\tikzcdset{
  cells={font=\everymath\expandafter{\the\everymath\displaystyle}},
}

\DeclareRobustCommand{\gobblefive}[5]{}
\newcommand*{\SkipTocEntry}{\addtocontents{toc}{\gobblefive}}

\usepackage{amssymb} 
\usepackage{mathtools}
\usepackage{subcaption} 
\DeclareCaptionFormat{custom}

\usepackage[pdftex,colorlinks=true]{hyperref} 
\hypersetup{urlcolor=blue, citecolor=red, linkcolor=blue}

\usepackage{float} 
\usepackage{todonotes}

\usepackage[widespace]{fourier} 
\usepackage[capitalise]{cleveref}

\DeclareMathOperator{\Proj}{Proj} 
\renewcommand{\vec}[1]{\mathbf{#1}}

\DeclareMathOperator{\Bl}{Bl}

\DeclareMathOperator{\D}{D}
\DeclareMathOperator{\Eg}{E}

\newcommand{\C}{\mathbb{C}} 

\newcommand{\Z}{\mathbb{Z}}

\newcommand{\BF}{\mathbb{F}}
\newcommand{\Pj}{\mathbb{P}}

\DeclareMathOperator{\Id}{Id}

\DeclareMathOperator{\Div}{Div}

\newtheorem{maintheorem}{Theorem}

\newtheorem{theorem}{Theorem}[section]
\newtheorem{proposition}[theorem]{Proposition}
\newtheorem{corollary}[theorem]{Corollary}
\newtheorem{lemma}[theorem]{Lemma}
\theoremstyle{definition}
\newtheorem{definition}[theorem]{Definition}
\newtheorem{notation}[theorem]{Notation}

\theoremstyle{remark} 
\newtheorem{remark}[theorem]{Remark}

\newtheorem{conjecture}{Conjecture}

\numberwithin{equation}{section}

\theoremstyle{plain}

\newcommand{\p}{\mathbb{P}}

\newcommand{\pain}[1]{\operatorname{P}_{\mathrm{#1}}}
\newcommand{\defeq}{\vcentcolon=}

\definecolor{dark green}{rgb}{0.0, 0.5, 0.0}
\allowdisplaybreaks

\usepackage[foot]{amsaddr}
 
\title{The Painlev\'e equivalence problem for a constrained 3D system}

\author[Galina Filipuk]{Galina Filipuk$^{1}$}

\address{$^1$Institute of Mathematics, University of Warsaw, ul. Banacha 2, 02-097, Warsaw, Poland. ORCID: 0000-0003-2623-5361.  (Corresponding author)}

\email{filipuk@mimuw.edu.pl}

\author[Michele Graffeo]{Michele Graffeo$^{2}$}

\address{$^{2}$Scuola Internazionale Studi Superiori Avanzati, Via Bonomea, 265, 34136 Trieste. ORCID: 000-0002-7973-6023.}

\email{mgraffeo@sissa.it}

\author[Giorgio Gubbiotti]{Giorgio Gubbiotti$^{3}$} 

\address{$^{3}$Dipartimento di Matematica ``Federigo Enriques'' Universit\`a degli 
Studi di Milano, Via Cesare Saldini 50, 20133, Milano and INFN Sezione di Milano,
Via Giovanni Celoria 16, 20133, Milano. ORCID: 0000-0003-0423-3736.}

\email{giorgio.gubbiotti@unimi.it}

\author[Alexander Stokes]{Alexander Stokes$^{1,4}$}

\address{$^{4}$ Waseda Institute for Advanced Study, Waseda University, Nishiwaseda Building,
Nishiwaseda 1-21-1, Shinjuku-ku, Tokyo 169--0051 Japan. ORCID: 0000-0001-6874-7141.
}

\email{stokes@aoni.waseda.jp}



\begin{document}

\begin{abstract}
    In this paper we propose a geometric approach to study Painlev\'e equations appearing as constrained systems of three first-order ordinary differential equations.   We illustrate this approach on a system of three first-order differential
    equations arising in the theory of semi-classical orthogonal polynomials. We show that it 
    can be restricted to a system of two first-order differential equations
    in two different ways on an invariant hypersurface.
    We build the space of initial conditions for each of these restricted systems and 
    verify that they exhibit the Painlev\'e property from a geometric perspective.
    Utilising the Painlev\'e identification algorithm 
    we also relate this system to the Painlev\'e VI equation and we build its
    global Hamiltonian structure.
    Finally, we prove that the autonomous limit of the original system is Liouville
    integrable, and the level curves of its first integrals are elliptic curves, which 
    leads us to conjecture that the 3D system itself also possesses the Painlev\'e property
 without the need to restrict it to the invariant hypersurface.
\end{abstract}

\maketitle

\setcounter{tocdepth}{1}

{\hypersetup{linkcolor=black}\tableofcontents}
\section{Introduction}

At the turn of the 20\textsuperscript{th} century Paul Painlev\'e and his school
carried out a program aiming to classify all second-order differential equations 
in the complex domain admitting at most movable 
poles of solutions~\cite{Painleve1900,Painleve1902,Fuchs1906,Gambier1910}\footnote{To be more 
precise, Painlev\'e and his school considered equations of the form 
$x''(t)=F(x(t),x'(t);t)$ where $F$ is polynomial in $x'(t)$, rational in $x(t)$, 
and analytic in $t$.}. The main outputs of this program were six equations, nowadays known 
as the \emph{Painlev\'e equations} and denoted $\pain{I}$, $\pain{II}$, \ldots, $\pain{VI}$,
whose solutions are transcendental functions not yet defined at that time. 
We refer to the book of Edward L. Ince~\cite{InceBook} for an extended explanation of the 
original work of Painlev\'e and his followers.

During the 20\textsuperscript{th} century Painlev\'e equations and their discrete counterparts appeared 
in many problems, both theoretical and applied. In particular, they play a prominent 
r\^ole in the theory of semi-classical orthogonal polynomials, see \cite{walterbook} 
and the numerous references therein. These reasons keep fueling the interest in Painlev\'e 
equations  nowadays.

One of the main advancements in the theory of Painlev\'e equations was made in the late
70's by Kazuo Okamoto who proposed a completely new approach by means of the construction of the \emph{space of initial conditions} for each of $\pain{I},\ldots,\pain{VI}$, which allows Painlev\'e's
results to be understood in algebro-geometric terms~\cite{Okamoto french}. 
After the discrete analogues of Painlev\'e equations were introduced, see~\cite{fokasitskitaev1991,Grammaticos1991,RGH91}, 
this approach was extended to a unified framework for Painlev\'e equations and their discrete analogues by Hidetaka Sakai \cite{Sakai}. We call this approach to Painlev\'e equation
the (Okamoto-Sakai) \emph{geometric theory of Painlev\'e equations}. 
It allows one to partially bypass the classical
complex analytic setting, by associating to a second-order differential equation (or equivalently a 
system of two first-order differential equations) a complex rational 
surface~\cite{BeauvilleBook} obtained from a suitable compactification of the affine plane via a sequence of blow ups.  The main idea behind this association is that  Painlev\'e-type
equations are related to the so-called generalised Halphen
surfaces, see \cite{Sakai}, with zero-dimensional anti-canonical linear system whose unique member consists
of a configuration of $(-2)$-curves. This configuration identifies the equation-type uniquely. Moreover,  the dynamics 
is \emph{vertical} on the $(-2)$-curves. For instance, the ``most degenerate'' 
Painlev\'e equation, i.e. $\pain{I}$, is associated to a rational surface 
with a configuration of $(-2)$-curves intersecting according to the $\Eg_8^{(1)}$ Dynkin diagram,
while the ``least degenerate''
Painlev\'e equation, i.e. $\pain{VI}$, is associated to the Dynkin diagram $\D_4^{(1)}$, see \Cref{fig:Dynkin D4}, through its configuration of $(-2)$-curves. 

Since its introduction the geometric theory of Painlev\'e equations has been a valuable
tool to describe properties of Painlev\'e equations
such as Lax pairs, B\"acklund transformation symmetries, particular rational and/or hypergeometric solutions, and asymptotics,
see the survey \cite{KNY} and references therein for a complete account of these
developments. 
It is remarkable that the geometric theory of Painlev\'e 
equations solves (to a certain extent) a version of the so-called Painlev\'e equivalence problem \cite{clarksonopenproblems}.
Given a system either known or suspected to be equivalent via
birational transformation to one of $\pain{I},\pain{II},\ldots,\pain{VI}$, this problem consists in determining which one, and finding the transformation to the standard form of the relevant Painlev\'e equation explicitly.
The identification procedure was first formulated in the discrete 
case in \cite{hypergeometric} and then adapted to the differential case in \cite{Stud}.
It is based on constructing the minimal space of initial conditions for the system under consideration and matching it with that of a standard form of a Painlev\'e equation, which consists in building an identification between two generalised Halphen surfaces and their
configurations of $(-2)$-curves. 
This algorithm was applied successfully
to identify different Hamiltonian forms of Painlev\'e equations \cite{Hams,P3hams}, and 
recently extended to the case of Hamiltonian systems with the quasi-Painlev\'e property 
in~\cite{tommarta}, using the results of \cite{tomgalina}.

In this paper we are interested in a system of three first-order differential equations coming from the 
work of Chao Min and Yang Chen \cite{MinChen}. Therein the polynomials orthogonal with
respect to a measure, called the \textit{degenerate Jacobi unitary measure}, were studied. In particular,
in \cite{MinChen} several equations satisfied by their recurrence coefficients, and certain auxiliary functions, were derived, see \Cref{sec:backgroundmaterial} for more details.
It is worth mentioning that the results of \cite{MinChen} extend the results of \cite{ChenZhang, DaiZhang} 
where special cases of the same weight were considered, and connections with particular cases of  $\pain{VI}$ were presented.

We will show in \Cref{sec:backgroundmaterial} that the results in \cite{MinChen}
relate the recurrence  coefficients and the auxiliary functions with the solutions of the following system of three non-linear ordinary differential equations 
in three dependent variables $\vec{x}(t)=(x,y,z)(t)$:

\begin{equation}
    \begin{aligned}
    (t-1) t x' &=x (\kappa+\beta  -\alpha  t-\beta  t+2 t y-2 (t-1) z-t)+(t-1) x^2-\kappa  (\beta +2 z),
    \\
    (\kappa -1) t^2 x y' &
    \begin{aligned}[t]
        =& 
    t x \left(n (n +\gamma)(x(t-1)+1) +txy \left( y  +2n+\gamma \right)+ ( \kappa-\beta -1)y(1-x)\right) 
    \\
    &-z\left(\beta   (\kappa - 2) \kappa + x  \left( (1-t)\beta + \left(t(2  y+  
       2  n   + \gamma )+\beta \right)(x(t-1)+1) \right)\right)
       \\
       &- z^2 \left((\kappa -2) \kappa -((t-1)x+1)^2+1\right), 
    \end{aligned}
    \\
    (t-1)z'&=t y',
    \end{aligned}
    \label{syst3D}%
\end{equation} 
where $\alpha$, $\beta$, $\gamma$ are complex constants, and $\kappa=\alpha +\beta +\gamma +2 n+1$,
such that $\vec{x}(t)$ satisfy the following non-linear, time-dependent constraint:
\begin{equation}
    S_t := \Set{h_t:=h(x,y,z;t)=0}\subset\C_t^3, 
    \label{eq:St}
\end{equation} 
where 
\begin{equation}\label{eq:h}
\begin{aligned}
    h_t &=  t x y (2 n (\kappa +(t-1) x)+t x (\gamma +y-2 z)+(x-\kappa ) ( 2 z-\alpha -\gamma)) \\
&\quad +(\kappa +(t-1) x) \left(x
   \left(n t (\gamma +n)-z (\beta +2 n t+\gamma  t)+(t-1) z^2\right)+\kappa  z (\beta +z)\right).
\end{aligned}
\end{equation}
That is, the solutions of the system~\eqref{syst3D} have to satisfy the additional condition:
\begin{equation} \label{hypersurfaceconstraint}
     h_t(\vec{x}(t);t)\equiv 0, \quad \forall t \in B,
\end{equation}
for a domain $B\subset \C$.

Note that it is not evident if the hypersurface $S_t$ admits a rational parametrisation,
so that we can obtain a restriction of the system~\eqref{syst3D} to a system of two 
rational first-order differential equations. In this paper, we will show that
it is possible to obtain such a restriction in at least two different ways:
one deduced from results in \cite{MinChen}, and another through a purely 
algebro-geometric analysis of the hypersurface. To be more precise, we will prove the 
following statement about the system~\eqref{syst3D}.

\begin{maintheorem}[\Cref{Th1,prop:darboux,cor:secondsystemtosystfg,prop:hamsys,prop:sys3d}]
    The system~\eqref{syst3D}, complemented with the invariant hypersurface 
    condition~\eqref{eq:St} restricts to a 2D system, admitting 
    a space of initial conditions whose $(-2)$-curves intersect according to the Dynkin diagram $\D_4^{(1)}$. In addition, we show:
    \begin{itemize}
        \item the restricted system can be birationally mapped to Painlev\'e-VI in 
            a Hamiltonian form, and inherits a Hamiltonian structure on its space of initial conditions; 
        \item the polynomial defining the hypersurface $S_t$~\eqref{eq:h} is a 
            time-dependent Darboux polynomial for the system~\eqref{syst3D}, 
            i.e. the hypersurface $S_t$ is a \emph{Darboux surface}, see \Cref{ssec:darb};
        \item the system~\eqref{syst3D} admits a Liouville--Poisson integrable      
            autonomous limit such that the fibration defined by the integrals is given by elliptic curves.
    \end{itemize}
    \label{thm:main}
\end{maintheorem}

\begin{figure}
    \centering
    \begin{tikzpicture}[thick,scale=2]
        \def\pt{2};
        \node[draw,circle] (a0) at ($({-sqrt(\pt)/2},{sqrt(\pt)/2})$) {$0$};
        \node[draw,circle] (a1) at ($({-sqrt(\pt)/2},{-sqrt(\pt)/2})$) {$1$};
        \node[draw,circle] (a2) at (0,0) {$2$}; 
        \node[draw,circle] (a4) at ($({sqrt(\pt)/2},{sqrt(\pt)/2})$) {$3$};
        \node[draw,circle] (a5) at ($({sqrt(\pt)/2},{-sqrt(\pt)/2})$) {$4$};
        \draw[-] (a0) -- (a2) -- (a4);
        \draw[-] (a1) -- (a2);
        \draw[-] (a2) -- (a5); 
    \end{tikzpicture}
    \caption{The Dynkin diagram of $\D_4^{(1)}$.}
    \label{fig:Dynkin D4}
\end{figure}
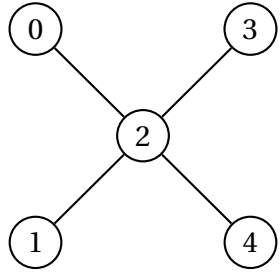

Before presenting the plan of the paper, we give some comments on this main result.
First, we observe that as stated in~\Cref{thm:main} (and previously),  we will reduce the 3D system~\eqref{syst3D} 
to two 2D systems. In previous papers, see \cite{hypergeometric, Appl, ITSF, Stud, prepr}, the 
Painlev\'e identification procedure was applied to 2D systems coming from orthogonal polynomials 
theory with respect to various semi-classical weights\footnote{The two components of the system usually represent
coefficients from the three-term recurrence for the orthogonal polynomials, or   some auxiliary functions 
defined using  the method of ladder operators, see, for instance, \cite[Chap. 4]{walterbook} and the 
references therein.}. So, to the best of our knowledge, this is the first time the geometric theory of Painlev\'e 
equations is applied to recognize a higher-dimensional system subject to a constraint as
a Painlev\'e equation.
However, it is worth mentioning that the relation between the degenerate Jacobi weight and $\pain{VI}$ is not new \cite{MinChen}. 
The idea of this paper is to develop the tools to identify Painlev\'e equations appearing as constrained higher-dimensional systems.
We also remark that asking that the solutions of the system~\eqref{syst3D} lie on the hypersurface $S_t$ 
is \emph{stronger} than the existence of the Darboux polynomial.

Finally, we underline that many questions about the full 3D system~\eqref{syst3D} remain open. 
From our viewpoint, the main open question is whether the 
system in its full generality  possesses the Painlev\'e property. As stated 
in~\Cref{thm:main}, we prove that there exists an autonomous limit which is 
multi-Hamiltonian and it defines an elliptic fibration. So, from the known relationship 
between elliptic fibrations and Painlev\'e equations, see 
e.g.~\cite{tsuda,duistermaatbook,carsteatakenawaelliptic,sakaiODEsonrationalellipticsurfaces}, this gives a strong 
indication for the following conjecture.

\begin{conjecture}
    The 3D system~\eqref{syst3D} possesses the Painlev\'e property, and it is
    possible to construct a space of initial conditions for it whose fibres are 
    rational threefolds with zero-dimensional anti-canonical linear system.
    \label{conj:3d}
\end{conjecture}

Note that, a crucial difference in the surfaces providing spaces of initial conditions between the autonomous and the non-autonomous
cases is the following. 
In the autonomous case (differential or difference equations solved by elliptic functions) the linear system associated to the anti-canonical
divisor of the surface constructed after blow up has \emph{positive dimension}, while 
in the non-autonomous case (Painlev\'e or discrete Painlev\'e equations) it has \emph{dimension zero}. See 
again~\cite{CarsteaTakenawa2019,GraffeoGubbiotti2023} for examples in the discrete setting.

We defer the study of the 3D system to future research, because the geometry of 
higher  dimensional differential equations seems to be much harder in comparison with the well-known two-dimensional one.
Though some examples of differential equations in higher dimensions have been studied geometrically \cite{kimurauniformfoliation,kimuragarnier,kimuradegengarnier,Suzukigarnier,taharaA4,sasanocoupledP5,sasanocoupledP2}, the construction of a space of initial conditions relies strongly on knowing that the system under consideration has the Painlevé property. 
We mention that also in the corresponding discrete 
case going beyond dimension two greatly increases the difficulty of problem and the involved computations, see e.g.~\cite{CarsteaTakenawa2019,GraffeoGubbiotti2023,takenawaFST,takenawagarnier}.

\subsection*{Outline of the paper}

In \Cref{sec:backgroundmaterial} we present a derivation of the system~\eqref{syst3D} 
and the foundations of the Okamoto-Sakai geometric theory of Painlev\'e equations. Moreover, we introduce the notion of a space of initial conditions in algebro-geometric language and we recall the notion of symplectic and Hamiltonian atlases for it. 
In \Cref{sec:firstpar} we consider a system of two first-order differential equations 
giving the flow of \eqref{syst3D} restricted to the hypersurface \eqref{eq:St} using a 
parametrisation derived from the results of Min-Chen \cite{MinChen}, construct a 
space of initial conditions for it, and build an  identification with the standard 
$\pain{VI}$ equation.  In \Cref{sec:hypersurf} we show that the hypersurface \eqref{eq:St} 
is in fact the zero locus of a (time dependent) Darboux polynomial~\cite{Goriely2001} 
for the system~\eqref{syst3D}. Then, using ideas from 
resolution of singularities we obtain another 
parametrisation of the hypersurface,  and we construct a space of initial conditions for 
the restriction of the system~\eqref{syst3D} in 
this parametrisation. Also, we  play again the identification game   compactifying the rational parametrisation to the Hirzebruch  surface 
$\mathbb{F}_1=\Bl_p\Pj^2$.  In the subsequent  
\Cref{sec:hamiltonianhyper} we end our study of the restricted system 
by discussing its Hamiltonian structure in both parametrisations. In \cref{sec:autlim} we 
prove, in the  same spirit as~\cite{alves} does for the $\pain{VI}$ equation, that  
system~\eqref{syst3D}  admits a Liouville--Poisson integrable  autonomous  limit, with 
the particular properties  described above. We give some 
conclusions and an outlook as complete as possible  on the prospective future 
developments in  \Cref{sec:concl}. In \Cref{appendix:standardP6}, we review some results on the geometry of $\pain{VI}$. Then, in \Cref{app:sigma} we highlight some analogy between  system \eqref{syst3D} and the $\sigma$-form of $\pain{VI}$. In \Cref{app:apparentsingularity}, we explain issues which can arise in the construction of a space of initial conditions from a non-optimal choice of compactification, namely the appearance of apparent singularities which should not be blown up. Finally, in \Cref{app:clim} we present two autonomous exponential limits of system \eqref{syst3D}.

\section{Background material} \label{sec:backgroundmaterial}

In this section, we give the basic tools we need in the paper. First, in \Cref{sss:minchen} we show how to derive the system~\eqref{syst3D}
using the results in~\cite{MinChen}, and in \Cref{subsec:Painleveq} we provide a general introduction to
the geometric theory of Painlev\'e equations. 
Later, in \Cref{appendix:eltransf} we recall the main properties of the intersection pairing on smooth projective surfaces and how Hirzebruch surfaces are related by elementary transformations. Then, in \Cref{subsec:spacesofinitialconditions,subsec:haplhen} we introduce the notion of space of initial conditions, and we explain Okamoto's models in terms of Sakai surfaces. Finally, in \Cref{subsec:HAMILTONIANSTRUC} we recall the notion of symplectic and Hamiltonian atlases necessary to prove the existence of global Hamiltonian structures of Painlev\'e equations on their spaces of initial conditions.

We will keep our discussion
to a minimum, and we will introduce additional tools as needed to complete
specific steps of the proof of~\Cref{thm:main} later.

\begin{notation}
    From now on, throughout the paper for the sake of simplicity, we will refer 
    to systems of first-order differential equations simply as ``systems''. Moreover, we sometimes
    omit the explicit dependence on the independent variable $t$, e.g.\ $x=x(t)$,
    when no possibility of confusion arises.
\end{notation}

\subsection{The 3D model}
\label{sss:minchen}

In this subsection we shall present the system of three first-order differential equations which will be the main object of study of this paper, and explain how it can be derived from various relations in \cite{MinChen}.

Consider the space $L^2([0,1],d\mu)$ where:
\begin{equation}
    d\mu = w(x,t)dt,
\end{equation}
with the following weight function:
\begin{equation}\label{eq:weight}
    w(x,t)=x^{\alpha} (1-x)^{\beta}|x-t|^{\gamma}(A+B \theta(x-t)), \quad x,\,t\in [0,\,1],\quad \alpha, \,\beta,\,\gamma>0,
\end{equation} 
where $\theta(x)$ is the \textit{Heaviside step function}, and $A\geq 0,\,A+B\geq 0$,
and $t\in\mathbb{C}$ are parameters. The orthogonal polynomials with respect to this weight
are called the \emph{degenerate Jacobi unitary polynomials}, and are
denoted by $P_n(z)$, where the non-negative integer $n\ge 0$ is the degree of $P_n(z)$, see~\cite{MinChen}.

Using the theory of ladder operators, see again~\cite[Chap. 4]{walterbook}, one defines 
two functions  $A_n(z)$ and $B_n(z)$ such that they satisfy the so-called \emph{lowering operator 
equation}:
\begin{equation}
\label{lowering}
\left(\frac{d}{dz}+B_{n}(z)\right)P_{n}(z)=\beta_{n}A_{n}(z)P_{n-1}(z),    
\end{equation}
and the \emph{raising operator equation}:
\begin{equation}
\left(\frac{d}{dz}-B_{n}(z)-\mathrm{v}_{0}'(z)\right)P_{n-1}(z)=-A_{n-1}(z)P_{n}(z),    
\end{equation}
where
\begin{equation}
    \mathrm{v}_{0}(x)=-\alpha\ln x-\beta\ln(1-x),\;\;\alpha,\;\beta>0.
\end{equation}

The expansions as $z\to \infty$ of the functions $A_n$ and $B_n$ are governed by six auxiliary 
functions $R_n(t),$ $r_n(t),$ $x_n(t)$, $y_n(t)$, $\alpha_n(t),$ $\beta_n(t)$, where 
$\alpha_n(t)$ and $\beta_n(t)$  are recurrence coefficients for the sequence of the associated 
monic orthogonal polynomials. This is the content of the following theorem.

\begin{theorem}[ {\cite[Theorem 2.4]{MinChen}} ]
    As $z\rightarrow\infty$, $A_{n}(z)$ and $B_{n}(z)$ have the following series expansions:
    \begin{equation}
        \begin{aligned}
            \label{anz1}
            A_{n}(z)&
            \begin{aligned}[t]
            &=\frac{\gamma+(t-1)R_{n}-t x_{n}}{z^2}+\frac{\gamma(t+\alpha_{n})+(t^2-1)R_n-t^2x_{n}}{z^3}
            \\
            &+\frac{\gamma(t^2+\alpha_{n}^2+t\alpha_{n}+\beta_{n}+\beta_{n+1})+(t^3-1)R_n-t^3x_{n}}{z^4}+O\left(\frac{1}{z^5}\right),    
            \end{aligned}
            \\ 
            B_{n}(z)&
            \begin{aligned}[t]
            &=-\frac{n}{z}+\frac{(t-1) r_{n}-ty_{n}-nt}{z^2}+\frac{\gamma\beta_{n}-nt^2+(t^2-1) r_n-t^2 y_{n}}{z^3}
            \\
            &+\frac{\gamma\beta_{n}(t+\alpha_{n}+\alpha_{n-1})+t^3(r_{n}-y_{n}-n)-r_n}{z^4}+O\left(\frac{1}{z^5}\right).                
            \end{aligned}
        \end{aligned}
    \end{equation}
    \label{thm:minchen0}
\end{theorem}
 
There are  numerous difference (with respect to $n$) and differential-difference relations 
among the auxiliary functions $R_n(t)$, $r_n(t)$, $y_n(t)$ and the recurrence coefficients 
$\alpha_n(t)$, $\beta_n(t)$. The relations relevant to our discussion are 
summarised in the following theorem which collects some of the results in \cite{MinChen}.

\begin{theorem}
    The variables $\alpha_n$ and $\beta_n$ satisfy the following
    Toda-like equations:
    \begin{equation}
        \begin{aligned}
            t\alpha_{n}' &=\alpha_{n}+r_n-r_{n+1},
            \\
            t\beta_{n}' &=\beta_{n}(2-R_{n}+R_{n-1}).
        \end{aligned}
        \label{eq:todalike}
    \end{equation}
    Furthermore, the recurrence coefficients $\alpha_{n}$ and $\beta_{n}$ have the 
    following expressions in terms of $r_{n}, y_{n}$ and $R_{n}$:
    \begin{equation}
        \label{eq:prop31minchen}
        \begin{aligned} 
        (\kappa+1)\alpha_{n} &=2(t-1)r_{n}-2ty_{n}-(t-1)R_{n}+t\alpha+(t-1)\beta+t,
        \\
        \kappa(\kappa-2)\beta_{n}
            &
        \begin{aligned}[t] 
            &=\left[ty_{n}-(t-1)r_{n}\right]^{2}
            -(t-1)(2nt+\gamma t+\beta)r_{n}
            \\
            &+t\left[(t-1)(2n+\gamma)-\alpha\right]y_{n}
            +n(n+\gamma)(t^2-t),
        \end{aligned}
        \end{aligned}
    \end{equation}
    and the following relation between the variables $r_n$ and $y_n$ holds:
    \begin{equation}
        (t-1)r_n'(t)=ty_{n}'(t).
        \label{eq:rnyn}
    \end{equation}
    Finally,  the coefficient $y_n$ satisfies:
    \begin{equation}
        \begin{aligned}
        y_{n}=&\frac{1}{2 t R_n}\Big\{t(t-1) R_n'-(t-1) R_n^2
        \\
        &+\left[2 (t-1) r_n+(\alpha +\beta +1)t-\kappa -\beta\right]R_n + \kappa (2r_n+\beta)\Big\},
        \end{aligned}
        \label{eq:ynsol}
    \end{equation}
    and there is a polynomial equation relating $r_n$ with
    $R_n$ and $R_n'$ of the form:
    \begin{equation}
        \mathcal{N}(R_n,R_n')r_n - \mathcal{D}(R_n,R_n')=0.
        \label{eq:rnRnlin}
    \end{equation}
    \label{thm:minchen}
\end{theorem}

\Cref{thm:minchen} summarises the results of~\cite[Sec. 3]{MinChen}.
In particular, equations \eqref{eq:todalike} are~\cite[Eqs. (3.5) and (3.2)]{MinChen},
equations~\eqref{eq:prop31minchen} are the content of~\cite[Prop. 3.1]{MinChen},
relation \eqref{eq:rnyn} is~\cite[Eq. (3.6)]{MinChen}, while 
equation~\eqref{eq:ynsol} is~\cite[Eq. (3.21)]{MinChen}, and the final relation~\eqref{eq:rnRnlin}
is~\cite[Eq. (3.22)]{MinChen}. We omit the explicit expressions of the two
polynomials $\mathcal{N}$, $\mathcal{D}$ because they are too cumbersome,
see~\cite[p. 9180]{MinChen}.

From this we get the following result.

\begin{proposition}
    Relabel the variables in \Cref{thm:minchen} as:
    \begin{equation}
        \label{vars xyz}
        x(t)=R_n(t), \;z(t)=r_n(t), \; y(t)=y_n(t).
    \end{equation}
    Then, the relations in~\Cref{thm:minchen} imply that these three variables satisfy the 
    system~\eqref{syst3D} subject to the condition~\eqref{eq:St}.
\end{proposition}

\begin{proof}
    Solve~\eqref{eq:prop31minchen} with respect to  $\alpha_n(t),$ $\beta_{n}(t)$ 
    and substitute them and their derivatives into the Toda-type 
    equations~\eqref{eq:todalike}. Solve the resulting equations with respect to $r_n'(t)$ and 
    $y_n'(t)$ in terms of $R_n(t),$ $r_n(t),$ $y_n(t)$ and $R_n'(t)$. Substituting $R_n'(t)$ from~\eqref{eq:ynsol} and using equation~\eqref{eq:rnyn} gives  the desired equations. Finally, the constraint  can be derived from \eqref{eq:ynsol} and \eqref{eq:rnRnlin} by eliminating 
    $R_n'$ and substituting \eqref{vars xyz}. 
\end{proof}    

\begin{remark} 
    We remark that in \cite[Prop. 3.2, Thm. 3.4, and Thm. 3.5]{MinChen} it was 
    proved by a direct computation that 
    the functions $r_n(t)$ and $y_n(t)$ are related to the $\sigma$-form of 
    $\pain{VI}$, see \cite{JimboMiwaII} and \Cref{app:sigma}, and the function $R_n(t)$ is related to 
    a particular solution of $\pain{VI}$.
    From our perspective, it is more  natural 
    to consider the system \eqref{syst3D}, together with the invariant hypersuface 
    condition~\eqref{eq:St}. As anticipated in the 
    Introduction we consider the restriction of system \eqref{syst3D} to a two-dimensional one and relate it to $\pain{VI}$ via a geometric argument. Moreover, we will comment in \Cref{app:sigma}, on how to relate system~\eqref{syst3D} to the $\sigma$-form 
    of $\pain{VI}$.
\end{remark}

\subsection{Painlev\'e equations}
\label{subsec:Painleveq}

The celebrated six Painlev\'e equations $\pain{I}$, $\pain{II}$, \ldots, $\pain{VI}$ 
naturally arise in numerous problems in mathematics and mathematical physics. 
They have many remarkable properties, including Hamiltonian structures, symmetries forming affine Weyl groups, and relation with isomonodromic deformation of linear differential equations \cite{painlevehandbook,riemannhilbert,GLS,noumibook}.
Recall that the only singularities of solutions of the Painlev\'e equations which are movable, i.e. their locations in the complex plane depend on initial conditions,
are poles. 
This is usually referred to as the \emph{Painlev\'e property}, and $\pain{I}$, \ldots, $\pain{VI}$ 
are characterised by the fact that they have this property but are not solvable in general in terms of elementary functions or classical special functions which satisfy linear differential equations.  

Recall that the sixth Painlev\'e equation in its standard scalar form is given by:
\begin{equation}
    \begin{aligned}
        f'' &= \frac{1}{2}\left(\frac{1}{f}+\frac{1}{f-1}+\frac{1}{f-t}\right)\left(f'\right)^2
-\left(\frac{1}{t}+\frac{1}{t-1}+\frac{1}{f-t}\right)f'   
        \\
        &+\frac{f(f-1)(f-t)}{t^2 (t-1)^2 }\left(A+B \frac{t}{f^2}+C\frac{t-1}{(f-1)^2}+D \frac{t(t-1)}{(f-t)^2}\right),        
    \end{aligned}
    \label{P6}
\end{equation}
where $A, \, B,\,C,\, D$ are arbitrary (complex) parameters.
The Painlev\'e property of equation \eqref{P6} means  that any (locally defined) solution $f(t)$ can be continued to a single-valued meromorphic function on the universal covering space of $B = \C \setminus  \Set{0,1 }$, where the need to take the universal cover arises because solutions may be branched about fixed singularities.

One of the standard forms of $\pain{VI}$ as a system of two first-order equations appears in the survey \cite{KNY} of Kajiwara, Noumi and Yamada (KNY):
\begin{equation} \label{systfg}
    \begin{aligned}
f'&= \frac{(f-1) f (f-t)}{t(t-1)} \left(\frac{2 g}{f} - \frac{a_0-1}{f-t}-\frac{a_3}{f-1}-\frac{a_4}{f}\right), \\
g'& = - \frac{ f^2(g+a_2)(g+a_1+a_2) - t g(g-a_4)}{t(t-1)f},
    \end{aligned}
\end{equation}
where $a_0,\dots,a_4$ are complex parameters subject to the single constraint 
\begin{equation}
    a_0 + a_1 + 2 a_2 + a_3 + a_4 =1.
\end{equation}

The two forms \eqref{P6} and \eqref{systfg} are equivalent, via elimination of $g$, up to the following identification of parameters: 
\begin{equation}
    A = \frac{a_1^2}{2}, \quad B = - \frac{a_4^2}{2}, \quad C = \frac{a_3^2}{2}, \quad D = \frac{1- a_0^2}{2}.
\end{equation}
The system \eqref{systfg} is of Hamiltonian form 
\begin{equation} \label{systfgham}
    \begin{gathered}
        \frac{ f'}{f} = \frac{\partial H^{\operatorname{KNY}}_{\rm{VI}}}{\partial g}, \quad \frac{ g'}{f} = - \frac{\partial H^{\operatorname{KNY}}_{\rm{VI}}}{\partial f}, 
    \end{gathered}
\end{equation}
with Hamiltonian given by 
\begin{equation} \label{systf H - intro}
        H_{\rm{VI}}^{\operatorname{KNY}}= \frac{ (f-1)(f-t)g}{t(t-1)} \left( \frac{g}{f} - \frac{a_0-1}{f-t} - \frac{a_3}{f-1} -\frac{a_4}{f} \right) + \frac{a_2 (a_1+a_2)(f-t)}{t(t-1)}.
\end{equation} 

\subsection{Intersection form and elementary transformations}
\label{appendix:eltransf}

Given a smooth projective surface $S$ denote by $\Div(S)$ its divisor group, i.e. the free abelian group generated by its irreducible codimension-one subvarieties. In this subsection we recall the main properties of the intersection product on $\Div(S)$, see \cite{BeauvilleBook} for more details. Moreover, we will briefly explain how elementary transformations relate different Hirzebruch surfaces, see \cite[Example V-5.7.1]{Hartshorne2013}.

Recall \cite[Theorem V-1.1]{Hartshorne2013} that there is a unique symmetric bilinear form
\[
\begin{tikzcd}[row sep =tiny]
    \Div(S)\times \Div(S)\arrow[r] &\Z\\
    (C,D)\arrow[r,mapsto] & C.D,
\end{tikzcd}
\]
named \textit{intersection pairing}, such that:
\begin{itemize}
    \item if $C,D\subset S$ are non-singular curves meeting transversally then $C.D=|C\cap D|$;
    \item it only depends on linear equivalence classes, i.e. if $C,D,D'\in\Div(S)$ are divisors then $C.D=C.D'$ whenever $D\sim D'$. 
\end{itemize} 

The following lemma is a direct consequence of the defining properties of the intersection pairing.
\begin{lemma}\label{lemma:eltransf}
Let $\varepsilon:\Bl_p S\rightarrow S$ be the blow up of a smooth projective surface $S$  centered at a point $p\in S$. Then, the exceptional curve $L=\varepsilon^{-1}(p)$ is a smooth rational curve of self-intersection $L^2=-1$.

Let also $p\in C\subset S$ be an irreducible curve passing through $p$, and smooth at $p$.
Consider its proper transform $D=\overline{\varepsilon^{-1}(C\setminus \Set{p})}$. Then, we have $C^2=D^2-1$. 
\end{lemma}
Recall the definition of Hirzebruch surface.
\begin{definition}[Hirzebruch surface]
   A \textit{Hirzebruch surface} is the datum of a smooth projective surface $S$ and a morphism $S\to \Pj^1$ with all fibres isomorphic to $\Pj^1$.
\end{definition}
\begin{remark} 
Hirzebruch surfaces are naturally indexed by non-negative integers. Under this identification, the $k$-th  Hirzebruch  surface is the projectivisation of the rank-two vector bundle $\mathcal{O}_{\Pj^1}\oplus\mathcal{O}_{\Pj^1}(k)$, i.e.
\[
\BF_k\cong \Proj_{\Pj^1}(\mathcal{O}_{\Pj^1}\oplus\mathcal{O}_{\Pj^1}(k)).
\]
For instance, we have $\BF_0=\Pj^1\times\Pj^1$ and $\BF_1=\Bl_p\Pj^2$.
    Moreover, the fibration $\BF_k\to\Pj^1$ has two sections $C_+,C_-$ such that $C_+^2=-C_-^2=k$. More details on Hirzebruch surfaces can be found in  \cite[Chapter V]{Hartshorne2013}.
    \label{rem:sechirz}
\end{remark}
Any two Hirzebruch surfaces are birational to each other. The birational tranformations relating them are sequences of  the so-called \textit{elementary transformations}.

\Cref{fig:eltransf}   describes the elementary transformation relating $\BF_k$ and $\BF_{k+1}$. Each curve is labeled by its self-intersection. Notice that at each step self-intersections are computed via  \Cref{lemma:eltransf}.
\begin{figure}
    \centering
    \begin{tikzpicture}[basept/.style={circle, draw=black!100, fill=black!100, thick, inner sep=0pt,minimum size=1.2mm},scale=0.8]    
     \begin{scope}[xshift = -4cm, yshift=+1cm]
        \node at (0,2.5) {$\mathbb{F}_k$};
        \draw[thick] (-2,-2)--(-2,2) node [midway, left] {$0$}
            (2,-2)--(2,2) node [midway, right] {$0$}
            (-2.5,1.5)--(2.5,1.5) node [midway, above] {$k$}
            (-2.5,-1.5)--(2.5,-1.5) node [midway, below] {$-k$}
      ;
      
	   \node at (-2,1.5) [basept,label={[xshift=-16pt, yshift = -3 pt] \small{$(\infty,0)$}}] {};
    \node at (2,1.5) [basept,label={[xshift=18pt, yshift = -3 pt] \small{$(\infty,\infty)$}}] {};
    \node at (2,-1.5) [basept,label={[xshift=18pt, yshift = -16 pt] \small{$(0,\infty)$}}] {};
    \node at (-2,-1.5) [basept,label={[xshift=-14pt, yshift = -16 pt] \small{$(0,0)$}}] {};
    \end{scope}

    \begin{scope}[xshift = +4cm, yshift=+1cm]
    \node at (0,2.5) {$\mathbb{F}_{k+1}$};
    \draw[thick] (-2,-2)--(-2,2) node [midway, left] {$0$}
            (2,-2)--(2,2) node [midway, right] {$0$}
            (-2.5,1.5)--(2.5,1.5) node [midway, above] {$k+1$}
            (-2.5,-1.5)--(2.5,-1.5) node [midway, below] {$-k-1$}
      ;
      
	   \node at (-2,1.5) [basept,label={[xshift=-16pt, yshift = -3 pt] \small{$(\infty,0)$}}] {};
    \node at (2,1.5) [basept,label={[xshift=18pt, yshift = -3 pt] \small{$(\infty,\infty)$}}] {};
    \node at (2,-1.5) [basept,label={[xshift=18pt, yshift = -16 pt] \small{$(0,\infty)$}}] {};
    \node at (-2,-1.5) [basept,label={[xshift=-14pt, yshift = -16 pt] \small{$(0,0)$}}] {};
    \end{scope}
    
    \draw [->] (1.5,-2.75)--(2.5,-1.75) node[pos=0.5, below right] {$\text{Bl}_{(\infty,\infty)}\BF_{k+1}$};

    \begin{scope}[xshift = 0cm, yshift=-5.5cm]
    \draw[thick] (-2,-2)--(-2,2) node [midway, left] {$0$}
            (-2.5,1.5)--(2.5,1.5) node [midway, above] {$k$}
            (-2.5,-1.5)--(2.5,-1.5) node [midway, below] {$-k-1$}
      ;
    \draw[thick,red] (2,2)  .. controls (2,0.8) .. (1.3,-0.2) node [midway, right] {$-1$};
    \draw[thick,red] (2,-2)  .. controls (2,-0.8) .. (1.3,0.2) node [midway, right] {$-1$};
      
    \end{scope}
 \draw [->] (-1.5,-2.75)--(-2.5,-1.75) node[pos=0.5, below left] {$\text{Bl}_{(0,\infty)}\BF_{k}$};
\end{tikzpicture}
    \caption{The elementary transformation relating $\BF_k$ and $\BF_{k+1}$.}
    \label{fig:eltransf}
\end{figure}
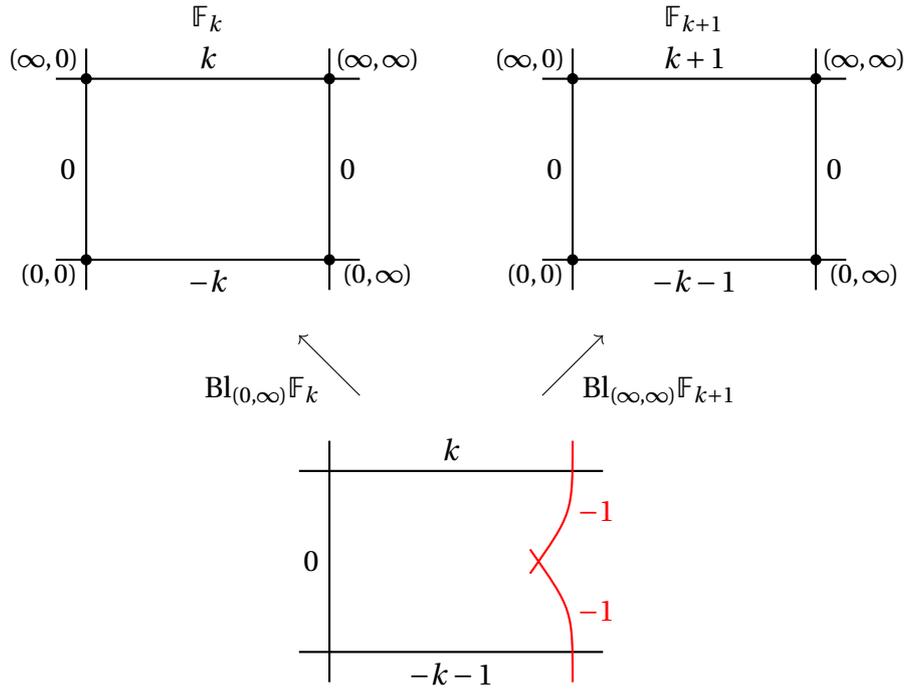

\subsection{Spaces of initial conditions and the Painlev\'e property}
\label{subsec:spacesofinitialconditions}

The fact that the sixth Painlev\'e equation \eqref{P6} has the Painlev\'e property is 
closely related to the existence of a space of initial conditions, as defined and constructed 
for $\pain{I}$, \ldots, $\pain{VI}$ by Okamoto \cite{Okamoto french}. His construction consists 
in considering an equivalent system of two first-order differential equations and then, through a 
combination of compactification and birational transformations, constructing an augmented phase space of which 
the flow of the system defines a uniform foliation.
Okamoto worked with polynomial first-order systems with meromorphic $t$-dependence of coefficients. Since we will be working with rational systems we give a formal definition of a space of initial conditions in algebro-geometric language. 
Before that, recall the notion of uniform foliation as used in the context of Painlev\'e equations \cite{Okamoto french, takano1, kimurauniformfoliation}. 

\begin{definition}[Uniform foliation]
    Consider a triple $(E,\pi,B)$ consisting of a complex manifold $E$ with a surjective holomorphic map $\pi : E\rightarrow B$ onto some domain $B\subset \C$.
    A uniform foliation of $(E,\pi,B)$ is a nonsingular foliation $\mathcal{F}$ of $E$ into complex one-dimensional analytic subsets called leaves such that
    \begin{itemize}
        \item each leaf of $\mathcal{F}$ intersects every fibre $E_t=\pi^{-1}(t)$, for $t\in B$, transversally;
        \item any path $\ell$ in $B$ with starting point $t_*  \in B$ and any point $p\in E_{t_* }$, can be lifted to the leaf passing through $p$.
    \end{itemize} 
\end{definition}
Consider a non-autonomous system of two first-order differential equations 
\begin{equation} 
    \label{systqpdef}
    q' = F(q,p;t), 
    \quad 
    p' = G(q,p;t),
\end{equation} 
where $F,G$ are rational in $q,p$ with coefficients being rational functions of $t$ 
regular on some domain $B\subset \C$. Regarding $q,p$ as coordinates on $\C^2$, the 
system~\eqref{systqpdef} defines a rational vector field on $\C^2 \times B$. We will sometimes refer to $B$ as the \textit{independent variable space} for the system \eqref{systqpdef}.

\begin{definition}[Indeterminacy point]
    A point $b=(q_*,p_*)$ in $\C^2 \times \Set{t_*}$ such that both the numerator and denominator of $F$ or $G$
    in equation \eqref{systqpdef} vanish is said to be an \emph{indeterminacy point}.
\end{definition}

If equation \eqref{systqpdef} admits a one-parameter family of local holomorphic solutions $(q_{\mu}(t),p_{\mu}(t))$, parametrised by $\mu\in \C$, all passing through the same point at $t=t_*$, i.e. $(q_{\mu}(t_*),p_{\mu}(t_*))=(q_*,p_*)$, then $(q_*,p_*)$ is an indeterminacy point.

\begin{remark}\label{rem:blowupseparate}
    Since we are working with smooth projective surfaces, the intersections of the zero loci of  the numerator and the denominator of $F$ or $G$
    in~equation \eqref{systqpdef} consist of a finite number of points. In order to regularise the system the main strategy consists in resolving these indeterminacy points and separating the solution curves passing through them to achieve a nonsingular foliation.
    The correct tool for this is the blow up of the surface at the intersection points. Indeed, roughly speaking this procedure decreases the tangency multiplicity of the two curves at meeting points. For instance, transversal curves are separated by blow ups.  
\end{remark}

\begin{definition}[Space of initial conditions] \label{def:spaceofinitialconditions}
    A \textit{space of initial conditions}, or \textit{space of initial values}, for the system \eqref{systqpdef} is a
    commutative diagram
    \begin{equation}\label{eq:spaceofinitialconditions}
        \begin{tikzcd}
     D\arrow[r,hook]\arrow[drr," \pi_D "']&   X \arrow[rr,dashed," \varphi_X"]\arrow[dr," \pi_X "] && \C^2\times B\arrow[dl,"\pi_B"]\\
        &&B
    \end{tikzcd}
    \end{equation}
    where
    \begin{itemize}
    \item $\pi_X$ is a proper morphism, $\varphi_X$ is a birational map and $\pi_B$ is the natural projection,
    \item $X$ is a smooth variety and the fibres $X_t$ of $\pi_X$ over each $t\in B$ are complex smooth rational projective surfaces,
        \item the restrictions $(\varphi_X)|_{X_t}$ are birational for all $t\in B$,
        \item $D\hookrightarrow X$ is a closed immersion, and the restriction $D_t \subset X_t$ defines 
            a non-empty codimension 1 subvariety\footnote{Actually, this is a Cartier divisor, see \Cref{subsec:parametrisation2}.},
        \item the restriction to $E=X\setminus D$ of the pullback via $\varphi_X$ of the 
            system \eqref{systqpdef} defines a uniform foliation of $(E,\pi_X|_{E},B)$, regarding
            $E$ as a complex manifold.
    \end{itemize} 
\end{definition}

\begin{notation} \label{notationSOIC}
   Whenever no confusion arises we will encode the data in \eqref{eq:spaceofinitialconditions} in a pair $(X,D)$ omitting the morphisms. Sometimes, with abuse of notation we will also refer to the variety $X$ as a space of initial conditions. 
   In what follows $D=\sum_{i} m_i D_i$ will be a divisor, but by abuse of notation we will sometimes use the same symbol to indicate the underlying subvariety $D=\cup_i D_i$.
\end{notation}

\begin{definition}\label{def:minimalSOIC}
   We say that a space of initial conditions is minimal if  any birational projective morphism to another space of initial conditions is an isomorphism.
\end{definition}

\begin{remark} \label{rem:inaccessible}
        In a space of initial conditions $X$, the pullback of the system is devoid
    of indeterminacy points on $X\setminus D$. 
    However it is important to note that achieving this kind of regularisation of a system does not alone guarantee that it defines a uniform foliation, and it must be shown that solutions do not reach $D$, usually by invoking the Painlev\'e property of the equation, though to some extent this can be done from knowledge of $(X,E)$ \cite{definingmanifolds}. 
\end{remark}

\begin{remark}\label{rem:MINIMAL}
    We also remark, the condition of being minimal for a space of initial conditions says that there are no curves that can be contracted without introducing new singularities.
    In the discrete setting, that is when dealing with iterations of 
    birational maps, one requires that the pullbacks
    to $X$ of the original map and its iterates has no divisorial contractions 
    in the indeterminacy locus. In other terms one asks that the map is algebraically stable, 
    see for instance~\cite{GraffeoGubbiotti2023,CarsteaTakenawa2019}.
\end{remark} 

Let us explain Okamoto's construction of a space of initial conditions for $\pain{I},\ldots,
\pain{VI}$ and how it relates to \Cref{def:spaceofinitialconditions}. For each Painlev\'e equation, 
Okamoto considered an equivalent system of two first-order differential equations, see equation \eqref{systqpdef}, with $F$ and $G$ polynomials in $q,p$ with coefficients analytic in $t$ on a domain $B\subset \C$ obtained removing the locations of fixed singularities of the Painlev\'e equation\footnote{For example $B =\C \setminus\Set{0,1}$ in the case $\pain{VI}$.}.
The systems in question are the polynomial Hamiltonian forms of the Painlev\'e equations which were provided by Okamoto \cite{OkamotoHams1, OkamotoHams2}, though they appeared much earlier in the work of Malmquist \cite{malmquist}.

As for a non-autonomous system of two first-order differential equations the phase space of \eqref{systqpdef} can be initially taken to be $\C^2 \times B$.
The Painlev\'e property means that for any path $\ell$ in $B$ with starting point $t_*  \in B$, the solution of any initial value problem for system \eqref{systqpdef} at $t=t_* $ can be meromorphically continued in $\C^2 \times B$ along $\ell$. 
 
The system \eqref{systqpdef} defines a nonsingular foliation of $\C^2 \times B$. 
However the presence of movable singularities of solutions means that solutions may not stay in $\C^2 \times B$ and the foliation is not uniform. 
This gives rise to the need to compactify the fibres of the phase space.

By compactifying $\C^2$ to some projective rational surface $S$ 
(common choices include $\p^2$, $\p^1 \times \p^1$ or more generally Hirzebruch 
surfaces $\mathbb{F}_{k}$ for any $k\ge 0$) one extends the system \eqref{systqpdef} to the 
trivial bundle $S \times B$. In general, the flow of the system does not define 
a nonsingular foliation since there may now be families of infinitely many solutions 
passing through a common point in the fibre over some $t\in B$, i.e.\ there can be families 
of infinitely many solutions with movable poles at the same location. 
Such families of solutions are parametrised by a free coefficient in the Laurent expansion 
about such a movable pole, classically known as a \emph{resonant parameter}, see for instance~\cite[\S IV]{Weissetal1983}.
To remedy this, one performs a sequence of blow ups of the fibre over $t \in B$ (of possibly $t$-dependent points).
Geometrically, this procedure separates the families of solutions and desingularises the foliation, see \Cref{rem:blowupseparate}. Finally, in order to reach a minimal space of initial condition, as Okamoto did, one might need to contract some rational curves. For instance in \Cref{sec:firstpar}, we have this necessity, while in \Cref{sec:hypersurf} we do not need further contractions.

This leads to a rational surface $X_t$ on which the system extended from \eqref{systqpdef} defines regular initial value 
problems everywhere away from a collection $D_t=\bigcup_i D_{t,i}$ of curves. 
It can be shown using the Painlev\'e property of system \eqref{systqpdef} that analytic continuations of the solutions of these initial value problems will never reach these curves.
Removing from $X_t$ the subvariety $D_t$ yield the fibre $E_t$ of the complex 
analytic bundle $\pi \colon E \rightarrow B$. Moreover, the flow of the system defines a 
uniform foliation of $E$.
In particular, any solution to an initial value problem for the 
system extended from \eqref{systqpdef} at $t_* \in B$ can be holomorphically continued in $E$ 
along any path $\ell$ with starting point $t_* $. 
In this sense each fibre $E_t$ parametrises the set of 
solutions and is called a space of initial conditions. 
\begin{remark}
With regards to classical terminology, note that the terms ``space of initial conditions'', 
``initial value space'' or ``space of initial values'' are traditionally used, following Okamoto, to refer to a fibre $E_t = X_t\setminus D_t$ of $\pi_X|_E$ over $t\in B$, i.e. the surface $X_t$ with divisor $D_t$ removed.  
The divisor $D_t$ or its components in the fibre over $t$ are usually called \emph{vertical leaves} of the foliation of $X$ with respect to the projection of $X$ to $B$, or \emph{inaccessible divisors} since leaves passing through points in $E_t$ will not reach them.  
   We will sometimes refer to $X_t$ as the \textit{compact surfaces} and $E_t=X_t\setminus D_t$ as the \textit{open surfaces}.
In the case of Painlev\'e equations, the total space $E$ is called a \emph{defining manifold} and has the structure of a complex analytic fibre bundle over $B$, but this is not algebraic and the isomorphisms between different fibres comes from the flow of the Painlev\'e equation, which is in general transcendental.

\end{remark} 
\subsection{Generalised Halphen surfaces}\label{subsec:haplhen}

For each Painlev\'e equation, a minimal space of initial conditions is a family of generalised Halphen surfaces, as defined by Sakai \cite{Sakai}.
A generalised Halphen surface is a smooth complex projective rational surface $X$ with an effective anti-canonical divisor of canonical type, i.e. $D \in |- \mathcal{K}_{X}|$ whose decomposition into irreducible components $D = \sum_{i} m_i D_i$ is such that the intersection pairing $\mathcal{K}_X . [D_i] = 0$ vanishes for all $i$, see \Cref{appendix:eltransf}. 
A generalised Halphen surface $X$ has anti-canonical linear system $|- \mathcal{K}_{X}|$ of dimension equal to either zero or one.
In the latter case $X$ is a rational elliptic surface, with the anti-canonical linear system providing its elliptic fibration, see \cite{SchuttShioda}.
In the former case there is a unique effective anti-canonical divisor $D$, and we call such $X$ a \emph{Sakai surface}.

As anticipated in the introduction there is a strong relation between affine root systems and Sakai surfaces.  
For a Sakai surface, the intersection graph of the irreducible components of the anti-canonical divisor $D$ is an affine Dynkin diagram of type $\mathrm{ADE}$.
In the cases of Sakai surfaces associated with Painlev\'e differential equations, the irreducible components are the curves removed from the fibre in the last step of Okamoto's construction, see \Cref{subsec:spacesofinitialconditions} and \Cref{notationSOIC}.
We list the corresponding types of Dynkin diagrams in Table \ref{table:surfacetypes}. 

\renewcommand{\arraystretch}{1.5}
\begin{table}[htb]
\centering
    \begin{tabular}{  | c | c | c c c  | c | c | c |}
    $\pain{I}$ & $\pain{II}$ & ~& $\pain{III}$ &~& $\pain{IV}$ & $\pain{V}$ & $\pain{VI}$ \\ \hline   
   $E_8^{(1)}$ & $E_7^{(1)}$ & $D_8^{(1)}$ & $D_7^{(1)}$ & $D_6^{(1)}$ & $E_6^{(1)}$ & $D_5^{(1)}$ & $D_4^{(1)}$  
    \end{tabular}
    \caption{Dynkin diagrams from intersection graphs of irreducible components of $D$ on Sakai surfaces for differential Painlev\'e equations}
   \label{table:surfacetypes}
   \end{table}
   \renewcommand{\arraystretch}{1}

\begin{remark}\label{rem:SAKAIMINIMAL} 
    Sakai surfaces are minimal as spaces of initial conditions for Painlev\'e equations, see \Cref{def:minimalSOIC}. 
Indeed, in order to contract a curve $C$ on a space of initial conditions without introducing new indeterminacies the curve $C$  has to be a vertical leaf of the foliation, i.e. an inaccessible divisor. On Sakai surfaces associated with differential Painlev\'e equations the inaccessible divisors are all $(-2)$-curves, which cannot be contracted on smooth points, see \cite[Theorem 5.7 and Remark V-5.7.2]{Hartshorne2013}.
\end{remark}

\subsection{Hamiltonian structures of Painlev\'e equations}\label{subsec:HAMILTONIANSTRUC}

We start this subsection by fixing our notation.

\begin{notation}
    Throughout this subsection we denote by $E$ a complex analytic fibre bundle 
    $\pi \colon E \rightarrow B$ over a domain $B \subset \C$, with fibre $E_t$ over 
    $t \in B$ equipped with a holomorphic symplectic form $\omega_t$. Moreover, we will 
    denote by $d$ the exterior derivative on the total space of $E$, and by $d_t$ the exterior 
    derivative on the fibre $E_t$, so that $d_t t = 0$.    
\end{notation}

For each Painlev\'e equation, the equivalent system, in the form    \eqref{systqpdef}, 
considered by Okamoto is of Hamiltonian form, i.e.
\begin{equation} \label{hamsystem}
    q' = \frac{\partial H}{\partial p}, \quad p' = - \frac{\partial H}{\partial q},
\end{equation}
for some Hamiltonian function $H=H(q,p;t)$ which is polynomial in $q,p$ with coefficients analytic in $t\in B$.
The Hamiltonian form of the system \eqref{hamsystem} extends to a global Hamiltonian structure 
of the system on $E$. Roughly speaking, this means that  $E$ admits an atlas such that 
the system is of Hamiltonian form in all charts.
However, the system \eqref{hamsystem} is non-autonomous and the gluing of $E$ is $t$-dependent 
 since the construction involves blow ups of $t$-dependent points. As a consequence,  a global Hamiltonian structure 
does not automatically follow from the Hamiltonian form \eqref{hamsystem} in the chart $(q,p;t)$ and, further, it only 
exists in certain atlases, as we will explain below.

Suppose that in some local coordinates $(q,p)$ for the fibre $E_t$ the symplectic form is written as $\omega_t = d_t q \wedge d_t p$.
Then, further suppose that  in this chart we have a non-autonomous system of differential equations of Hamiltonian form \eqref{hamsystem}.
We can extend this system to $E$, and its Hamiltonian structure is preserved if we have an atlas of canonical coordinates in the sense of the following definition. 
\begin{definition}[Symplectic atlas for $E$] \label{def:symplecticatlas} 
 A symplectic atlas for $E$ is the datum of an atlas $ \mathcal{U}$ for $E$, and a choice of coordinates $(x_U,y_U;t)$, for  each $U\in\mathcal{U}$, such that the symplectic form $\omega_t$ is written on each $U\in\mathcal{U}$ as
    \begin{equation}
        \omega_t = d_t x_U \wedge d_t y_U.
    \end{equation}   
\end{definition}

\begin{remark}
    Notice that, as a consequence of \Cref{def:symplecticatlas},  
    the local coordinates of a symplectic atlas are canonical coordinates 
    for $\omega_t$. They are  also called Darboux coordinates.
\end{remark}

Consider a chart $U_1 \in\mathcal{U}$ of a  symplectic atlas $\mathcal{U}$. A system of differential equations of Hamiltonian form
\begin{equation}
    x_1' = \frac{\partial H_1}{\partial y_1}, 
    \quad 
    y_1' = - \frac{\partial H_1}{\partial x_1}   
\end{equation}
in the chart $U_1$ transforms under the (possibly $t$-dependent) gluing
\begin{equation} \label{gluing}
  (x_1,y_1;t)\mapsto (x_2(x_1,y_1;t), y_2(x_1,y_1;t);t)  
\end{equation}
to be of Hamiltonian form 
\begin{equation}
    x_2' = \frac{\partial H_2}{\partial y_2},
    \quad
    y_2' = - \frac{\partial H_2}{\partial x_2},
\end{equation}
in any other chart $U_2\in \mathcal{U}$, where the Hamiltonians $H_1(x_1,y_1;t)$ and   $H_2(x_2,y_2;t)$ are related under \eqref{gluing} by  
\begin{equation} \label{eq:hamiltonianscorrection}
    \frac{\partial H_2}{\partial x_1} = \frac{\partial H_1}{\partial x_1} - \left(\frac{\partial x_2}{\partial x_1}\frac{\partial y_2}{\partial t}-\frac{\partial y_2}{\partial x_1}\frac{\partial x_2}{\partial t}\right), \quad 
    \frac{\partial H_2}{\partial y_1} = \frac{\partial H_1}{\partial y_1} - \left(\frac{\partial x_2}{\partial y_1}\frac{\partial y_2}{\partial t}-\frac{\partial y_2}{\partial y_1}\frac{\partial x_2}{\partial t}\right).
\end{equation} 
The key difference from the autonomous case is that the Hamiltonians 
$H_1$ and $H_2$ will not necessarily coincide under the gluing if it has non-trivial $t$-dependence, as can be seen in equation \eqref{eq:hamiltonianscorrection}.
The global Hamiltonian structure of a system of differential
equations on $E$ is not provided by a single Hamiltonian function. 
Instead, it is provided 
by a collection of Hamiltonians, one in each chart $U$ of a symplectic atlas 
$\mathcal{U}$, which define a common two-form on $E$ written as
\begin{equation}
    \Omega = d x_U \wedge d y_U + d H_U \wedge d t.
\end{equation}
 
It is possible to relax the requirement that the atlas consists of canonical coordinates 
such that the system is still of Hamiltonian form with respect to $\omega_t$ in each chart.
To explain this we make the following definition.

\begin{definition}[Local Hamiltonian structure of a system of ODEs on $E$]
    Let $(x,y;t)$ be a local chart of $E$  in which the symplectic form $\omega_t$ on $E_t$ 
    is written as 
\begin{equation}
\omega_t = F(x,y;t) d_t x \wedge d_t y,
\end{equation} 
where $F$ is a rational function in $(x,y;t)$.
    We say that a system of ODEs on $E$ has a local Hamiltonian structure with 
    respect to $\omega_t$ in the chart $(x,y;t)$ if there exists a Hamiltonian 
    function $H(x,y;t)$ such that  in this chart  the system  reads as follows:
    \begin{equation}
        F(x,y;t) x' = \frac{\partial H}{\partial y}, 
        \quad 
        F(x,y;t) y' = - \frac{\partial H}{\partial x}.
    \end{equation}
\end{definition}
 
When the coefficient function $F(x,y;t)$ does not have $t$-dependence, 
the following lemma ensures that a local Hamiltonian structure 
for a system on $E$ survives under changes of coordinates. 
The proof is by calculation, and we remark that the case when $F_1(x_1,y_1;t)=1$, $F_2(x_2,y_2;t)=1$  appears in  \cite{takano1}.
\begin{lemma} \label{symplecticlemma} 
    Consider the non-autonomous Hamiltonian system 
    \begin{equation}  \label{H1system}
        F_1(x_1,y_1) x_1' = \frac{\partial H_1}{\partial y_1}, 
        \quad  
        F_1(x_1,y_1) y_1' = - \frac{\partial H_1}{\partial x_1},
    \end{equation} 
    with Hamiltonian function $H_1(x_1,y_1;t)$ rational in $x_1,y_1$, with coefficients rational 
    in $t$ and regular on $B$.
    Let $\varphi \colon \C^3 \rightarrow \C^3$ be a transformation between copies of $\C^3$, 
    with coordinates $(x_1,y_1;t)$ and $(x_2,y_2;t)$ respectively, given by 
    \begin{equation}
        \varphi \colon  (x_1,y_1;t) \mapsto (x_2(x_1,y_1;t), y_2(x_1,y_1; t); t),
    \end{equation}
    and denote its restriction to $\C^2 \times \Set{t}$ by $\varphi_t$. 
    Suppose that it satisfies the condition:
    \begin{equation} \label{symplecticF}
        F_1(x_1,y_1) \, d_t x_1 \wedge d_t y_1 = 
        \varphi_t^*  \left(F_2(x_2,y_2) \, d_t x_2\wedge d_t y_2 \right),
    \end{equation}
    for rational functions $F_1$, $F_2$ whose coefficients are independent of $t$.
    Then, there exists $H_2(x_2,y_2;t)$, unique up to functions 
    of only $t$, such that
    \begin{equation}  \label{twoformsF}
        F_1(x_1,y_1) \,dx_1\wedge dy_1 + d H_1 \wedge dt 
        = 
        \varphi^* \left(F_2(x_2,y_2)\, dx_2\wedge dy_2 + dH_2\wedge dt\right).
    \end{equation}
    Further, the system \eqref{H1system}
    is transformed under $\varphi$ to 
    \begin{equation} 
        F_2(x_2,y_2) x_2' = \frac{\partial H_2}{\partial y_2},
        \quad  
        F_2(x_2,y_2) y_2' = - \frac{\partial H_2}{\partial x_2}.
    \end{equation}
\end{lemma}

Therefore for a local Hamiltonian structure to extend to the whole of $E$ we 
require an atlas such that the symplectic form is independent of $t$ in all charts.

\begin{definition}[Hamiltonian atlas for $E$]  A Hamiltonian atlas for $E$ is the datum of 
    an atlas $ \mathcal{U} $, and a choice of coordinates $(x_U,y_U;t)$, 
    for each $U\in\mathcal{U}$, such that on each $U\in\mathcal{U}$ the 
    symplectic form $\omega_t$ is written  as
    \begin{equation}
        \omega_t = F_U(x_U,y_U) d_t x_U \wedge d_t y_U,
    \end{equation}   
    with $F_U$ independent of $t$.    
\end{definition}

Having a Hamiltonian atlas guarantees that a local Hamiltonian structure of a system of ODEs 
in a single chart extends to all charts of the atlas by \Cref{symplecticlemma}.

\begin{definition}[Global Hamiltonian structure of a system of ODEs on $E$]
    Given a Hamiltonian atlas $\mathcal{U}$ on $E$, a global Hamiltonian 
    structure of a system on $E$ is a collection of Hamiltonian functions $H_U$ 
    on each $U\in\mathcal{U}$ defining a two-form $\Omega$ on $E$ via
    \begin{equation}
        \Omega = F_U(x_U,y_U) dx_U \wedge dy_U + dH_U \wedge dt.
    \end{equation} 
\end{definition}
 
In the case of the Painlev\'e equations, the divisor associated to the symplectic form $\omega_t$, i.e. the form with respect to which the Hamiltonian 
structure is defined,   is  the unique effective 
anti-canonical divisor $D_t\in |-\mathcal K_{X_t}|$, see~\Cref{subsec:haplhen}. After Okamoto's construction for $\pain{II}$, 
\ldots, $\pain{VI}$, Kyoichi Takano and his collaborators constructed symplectic 
atlases for $E$ in the sense of \Cref{def:symplecticatlas} \cite{takano1,takano2,takano3}.
The system extended from Okamoto's Hamiltonian form of the corresponding Painlev\'e equation 
then has a global Hamiltonian structure and all Hamiltonians $H=H_U(x_U,y_U;t)$, for 
$U\in \mathcal{U}$, are polynomial functions of $x_U,y_U$. Further, Takano's school showed 
that this is the unique holomorphic global Hamiltonian structure extending meromorphically 
to $X$, so in this sense the manifold $E$ determines the Painlev\'e equation uniquely.

We will present the space of initial conditions, symplectic atlas and global Hamiltonian 
structure for the sixth Painlev\'e equation in \Cref{appendix:standardP6}.

\section{The system on the hypersurface in the first parametrisation}
\label{sec:firstpar}

In this section we provide our first example of a system of two first-order differential equations governing the restriction of system \eqref{syst3D} to the hypersurface \eqref{eq:St}. The main result of this section is \Cref{Th1} where we show that the restriction considered possesses the Painlevé property, which is part of \Cref{thm:main}.

\subsection{The first parametrisation and the associated 2D system}

In this subsection we show how to derive a system of two first-order equations from \eqref{syst3D}. We adopt the notation $\tilde{x}(t)=x(t)=R_n(t) $ and $\tilde{y}(t)=R_n'(t)$. 

In order to eliminate all the derivatives but $\tilde y'(t)$, we start by differentiating the first equation of system \eqref{syst3D} with respect to $t$ and use other equations of the system to eliminate the derivatives. Now, $\tilde{y}'(t)$ is expressed as a function of $r_n(t)$, $y_n(t)$ 
and $\tilde{x}(t) $. The first equation of the system is also  linear 
in  $y_n(t)$, and we can use it to solve with respect to $y_n(t)$. Substituting it into the expression for 
$\tilde{y}'(t)$ writes $\tilde{y}'(t)$ only in terms of $r_n(t)$ and $\tilde{x}(t)$. Finally, from \eqref{eq:rnRnlin}, we get the following system of  two first-order differential  equations:

\begin{equation}\label{syst1}
\tilde{x}'=\tilde{y},\quad \tilde{y}' =\frac{\phi(\tilde{x},\tilde{y};t)}{2 (t-1)^2 t^2 \tilde{x} (\kappa -\tilde{x}) (\kappa +(t-1) \tilde{x}) },
\end{equation}
where, $\phi$ is the polynomial: 
\begin{equation}
\begin{gathered}
\phi=-(t-1) \tilde{x}^4 \left( 6\kappa^2- \beta ^2 +4 n \left(t^2-6
   t\right) (\kappa-n)+t^2 (\alpha +\beta +1) (\alpha +\beta +2 \gamma +1)
\right. \\ \left.
-t
   \left(2 \alpha  (\kappa-2n)+2 \beta  (\gamma +1)+2 \gamma +5(\kappa-2n)^2\right)\right)
  \\  
-\tilde{x}^2 \left(\kappa ^2 \left(-\kappa^2  +6 \beta ^2   
 +4 n t (\kappa-n)+t^2 \left( \beta ^2-\alpha ^2+1\right)+t \left( 2 \alpha  (\kappa-2n)
\right. \right. \right. \\ \left.  \left.  \left.
-5 \beta ^2+2 \beta  (\gamma +1)+2 \gamma
   \right)\right)+3 t^2 (t-1)^3 \tilde{y}^2+2 \kappa  t \left(t^3-5 t^2+6 t-2\right) \tilde{y}\right)
 \\ 
+2 (t-1)
   \tilde{x}^3 \left((t-1)^2 t \tilde{y}-2 \kappa  \left(\beta^2-\kappa^2 
  +2 n t ( \beta +\gamma+n )+\beta  \gamma  t+(\kappa-1 )(\alpha+1)t\right)\right)
  \\  
-2 \kappa  \tilde{x} \left(-\left((t-2)
   (t-1)^2 t^2 \tilde{y}^2\right)+\kappa  t \left(2 t^2-3 t+1\right) \tilde{y}+\beta ^2 \kappa ^2
   (t-2)\right)
\\
-\kappa ^2 \left(\beta ^2 \kappa ^2-(t-1)^2 t^2 \tilde{y}^2\right)-(t-1)^3 \tilde{x}^6+2 \kappa 
   (t-2) (t-1)^2 \tilde{x}^5 .
\end{gathered}    
\end{equation}
Note that we also have the parametrisation of the invariant hypersurface in terms of $\tilde{x}$  and $\tilde{y}$. 
 
\subsection{Space of initial conditions} \label{subsec:spaceofinitialconditionssys3.2}
We now construct a space of initial conditions for the system \eqref{syst1}. Moreover, by proceeding as in \cite{Stud} we identify it with the KNY Hamiltonian form \eqref{systfg} for the sixth Painlev\'e equation.  
As outlined in \Cref{subsec:spacesofinitialconditions}, we consider the system as a rational vector field on $\C^2 \times B$, where $B = \C \setminus \Set{0,1}$ has coordinate $t$. 

Now, to construct the space of initial conditions, see \Cref{subsec:spacesofinitialconditions}, we compactify the fibres of the canonical projection $\pi_B:\C^2\times B\to B$ to the product $\Pj^1_{[x_0:x_1]}\times \Pj^1_{[y_0:y_1]}$ by considering the following identification:
\begin{equation}
    \tilde{x} = \frac{x_0}{x_1},\quad \tilde{y} = \frac{y_0}{y_1},
\end{equation}
where  as usual we have omitted the $t$-dependence from $\tilde x,\tilde y$ and $x_i,y_i$ for $i=0,1$.

 Precisely, each fibre $(\mathbb{P}^1\times\mathbb{P}^1)_t$, over $t\in B$, is covered by the four coordinate charts:
\begin{equation}
    (\tilde{x},\tilde{y}), 
    \quad
    \left(\tilde{x},\frac{1}{\tilde{y}}\right) = \left(\frac{x_0}{x_1},\frac{y_1}{y_0}\right), 
    \quad
    \left(\frac{1}{\tilde{x}},\tilde{y}\right)= \left(\frac{x_1}{x_0},\frac{y_0}{y_1}\right),
    \quad 
    \left(\frac{1}{\tilde{x}},\frac{1}{\tilde{y}}\right)= \left(\frac{x_1}{x_0},\frac{y_1}{y_0}\right).    
\end{equation} 
After compactification the phase space for system \eqref{syst1} is then the trivial bundle $(\p^1 \times \p^1) \times B$ over $B$. 
We perform a sequence of blow ups of the fibres over each $t\in B$, of points where the rational vector field defining the system has indeterminacies.  
We perform the blow up at a point $(x,y)=(a, b)$ as prescribed in \cite[Propositions IV-21 \& IV-25]{EisenbudHarrisBook}. This procedure introduces two charts whose coordinates we denote by $(u,v)$ and $(U,V)$. In this setting the blow up map is given by
\begin{equation}\label{eq:coordblowp}
\begin{tikzcd}[row sep =tiny]
    (u,v)\arrow[r,mapsto]& (a+uv,b+v) =(x,y),\\
    (U,V)\arrow[r,mapsto]& (a+V,b+UV)=(x,y).
\end{tikzcd}
\end{equation}
In these coordinates, the exceptional divisor of the blow up has local equations $v=V=0$.

We initially find the following points of indeterminacy on $(\p^1 \times \p^1)_t$
for the system \eqref{syst1}:
\begin{equation}
\begin{gathered}
p_1=\left(\left[0:1\right],\left[\beta\kappa:t(1-t)\right]\right) ,\quad 
p_2=\left(\left[0:1\right],\left[ \beta\kappa:t(t-1)\right]\right) ,\\
p_3=\left(\left[\kappa:1-t\right],\left[ (1-\alpha )\kappa:(t-1)^2\right]\right) ,\quad 
p_4=\left(\left[\kappa:1-t\right],\left[ (1+\alpha )\kappa:(t-1)^2\right]\right) ,\\
p_5=\left(\left[\kappa:1\right],\left[\gamma\kappa:1-t \right]\right),\quad
p_6=\left(\left[\kappa:1\right],\left[\gamma\kappa:t-1 \right]\right), \\
p_7=\left(\left[1:0\right],\left[1:0 \right]\right).
\end{gathered}
\label{eq:MIp1p7}    
\end{equation} 

A direct computation shows that the indeterminacies at $p_1,\ldots,p_6$ are resolved after a single blow up. 
After blowing up $p_7$,  the system still has an indeterminacy point on the corresponding exceptional curve. Precisely, adopting as usual the coordinates coming from \cite[Propositions IV-21 \& IV-25]{EisenbudHarrisBook} the indetermincy point is 
\begin{equation}
  p_8:\, (U_7,\,V_7)=(0,0).  
  \label{eq:MIp8}
\end{equation}
Then, if we blow up also $p_8$ we find two further indeterminacies: 
\begin{equation}
     p_9:\, (u_8,\,v_8)=(-t,0), \quad  p_{10}:\, (u_8,\,v_8)=(t,0).
\end{equation}
Finally, the blow up of these two points regularises the system.  

Denote by $S_{t}$ the  blow up of the fibre $(\Pj^1\times\Pj^1)_t$ with center at the points $p_i$, $i=1,\ldots,10$. The inaccessible divisors on $S_t$ are the proper transforms of the lines 
\[
\Set{\tilde x = 0},\ \Set{\tilde x =\frac{\kappa}{1-t}},\ \Set{\tilde x =\kappa},\ \Set{\tilde x =\infty},\ \Set{\tilde y =\infty}, 
\]
and of the exceptional divisors over the points $p_8,p_7$, see \Cref{fig:surface:syst1}. Denote by $D_t$ their union. Then, system  \eqref{syst1} defines a uniform foliation of $S_t\setminus D_t$, assuming $D_t$ is inaccessible.
Rather than establishing this directly by analysis of equation \eqref{syst1} on $S_t$, we will deduce this via the transformation to the Hamiltonian form \eqref{systfg} of $\pain{VI}$.
 Therefore, we have constructed a space of initial conditions for system  \eqref{syst1}. 

\begin{remark} \label{rem:stricttransforms}
    We remark that, with abuse of notation, in \Cref{fig:surface:syst1} we have denoted by
    the same symbol $L_i$ the exceptional divisor over the point $p_i$, for $i=1,\ldots,10$, and its proper transforms under blow up.
\end{remark}

If we want to run the identification procedure, we need now to reach a minimal space of initial conditions, i.e. a Sakai surface (see \Cref{rem:SAKAIMINIMAL}). A direct computation shows that $S_t$ is not minimal as we can contract two curves without introducing new indeterminacies, see \Cref{rem:MINIMAL}.

The first is the proper transform in $S_{t}$ of the coordinate line $\Set{\tilde{x}=\infty}$, which has self-intersection $-1$ and is inaccessible. The second is the image under the first contraction of the exceptional line coming from the blow up of $p_7$, which similarly is inaccessible and it has self-intersection $-1$, see \Cref{lemma:eltransf}.
Denote by $D_t$ the union of the $(-2)$-curves on the resulting surface $X_t$. 
Again a direct computation shows that the induced system on $X_t$ has no indeterminacies and any contraction onto a smooth surface does.
Then we get a family $X$ over $B$, and after removal of $D_t$ from each fibre we have $\pi_X|_{E} : E \rightarrow B$ with fibre $E_t = X_t \setminus D_t$ such that the system \eqref{syst1} defines regular initial value problems everywhere on $E$ and a uniform foliation assuming $D_t$ is inaccessible as will be confirmed below (see \Cref{rem:inaccessible2}).
Summing up, this procedure leads to a minimal space of initial conditions $(X,D)$ in the sense of \Cref{def:spaceofinitialconditions} and \Cref{def:minimalSOIC}. 

We depict the configuration of $(-2)$-curves of $X_t$ in the bottom-right corner of
   \Cref{fig:surface:syst1}. Precisely, curves of self-intersection $-1,\,-2$ are coloured in red, blue respectively, see \Cref{appendix:eltransf}.

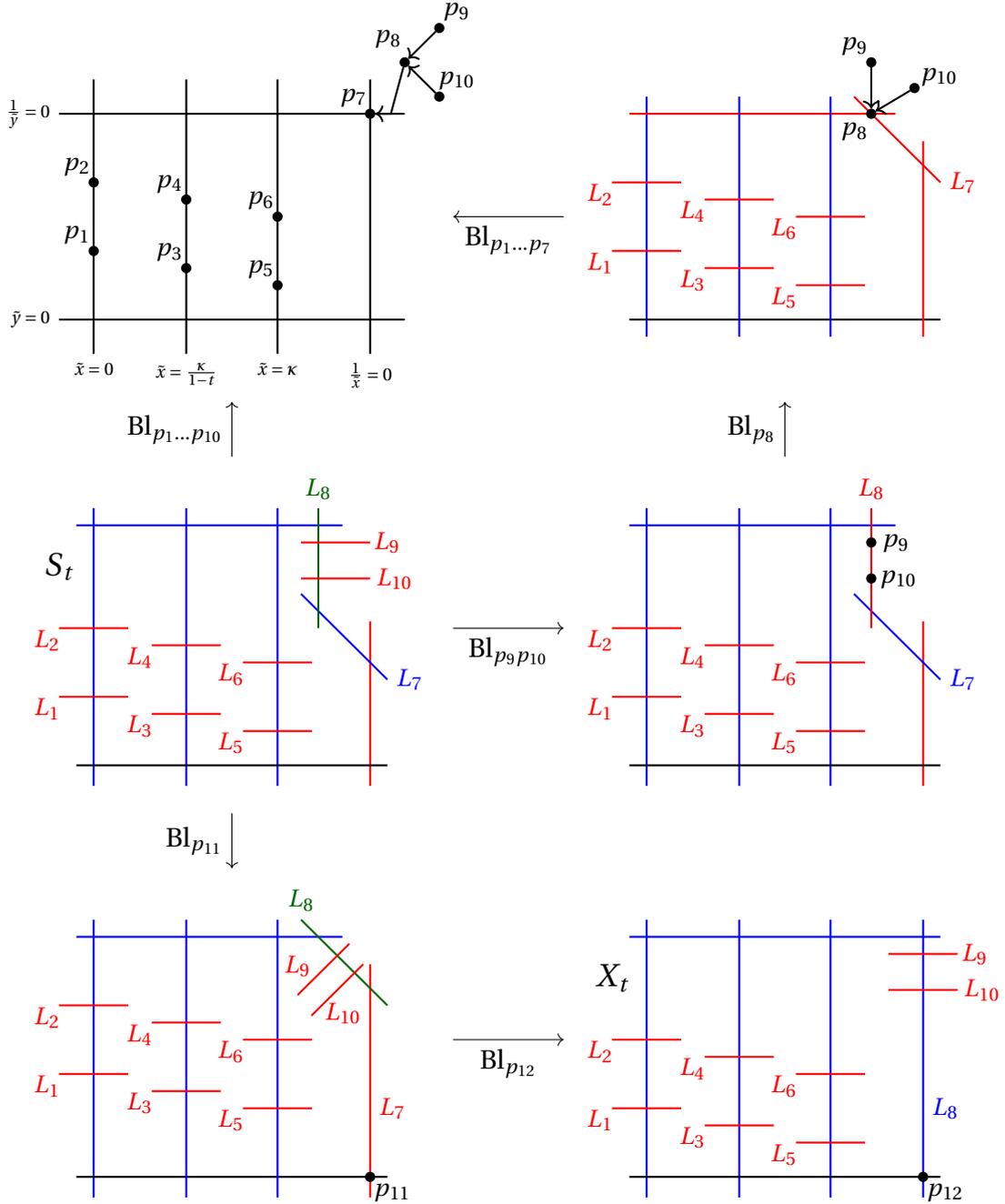
\begin{figure}[htb]
\centering
    \begin{tikzpicture}[basept/.style={circle, draw=black!100, fill=black!100, thick, inner sep=0pt,minimum size=1.2mm}]
    \begin{scope}[xshift = -4cm]
    \draw[thick] 
    (-2,-2)--(-2,2)
    (-2.5,-1.5)--(2.5,-1.5) 
    (2,-2)--(2,2) 
    (-2.5,1.5)--(2.5,1.5)
      (-.66,-2)--(-.66,2)  
      (.66,-2)--(.66,2)  
      ;     
      \node[left] at (-2.5,-1.5) {\tiny $\tilde{y}=0$};
      \node[left] at (-2.5,1.5) {\tiny $\frac{1}{\tilde{y}}=0$};     
      \node[below] at (-2,-2) {\tiny $\tilde{x}=0$};
      \node[below] at (2,-2) {\tiny $\frac{1}{\tilde{x}}=0$};
      \node[below] at (-.66,-2) {\tiny $\tilde{x}=\frac{\kappa}{1-t}$};
      \node[below] at (.66,-2) {\tiny $\tilde{x}=\kappa$};
      
	 \node (p1) at (-2,-.5) [basept,label={[xshift=-7pt, yshift = -3 pt] \small $p_{1}$}] {};
	 \node (p2) at (-2,+.5) [basept,label={[xshift=-7pt, yshift = -3 pt] \small $p_{2}$}] {};
	 \node (p3) at (-.66,-.75) [basept,label={[xshift=-7pt, yshift = -3 pt] \small $p_{3}$}] {};
	 \node (p4) at (-.66,+.25) [basept,label={[xshift=-7pt, yshift = -3 pt] \small $p_{4}$}] {};	 
      \node (p5) at (+.66,-1) [basept,label={[xshift=-7pt, yshift = -3 pt] \small $p_{5}$}] {};
	 \node (p6) at (+.66,+0) [basept,label={[xshift=-7pt, yshift = -3 pt] \small $p_{6}$}] {};
    \node (p7) at (2,1.5) [basept,label={[xshift=-7pt, yshift = -3 pt] \small $p_{7}$}] {};
      
	 \node (p8) at (2.5,2.25) [basept,label={[xshift=-7pt, yshift = 0 pt] \small $p_{8}$}] {};

    \node (p9) at (3,2.75) [basept,label={[xshift=7pt, yshift = -3 pt] \small $p_{9}$}] {};
    \node (p10) at (3,1.75) [basept,label={[xshift=7pt, yshift = -3 pt] \small $p_{10}$}] {};
    \draw[thick, ->] (p8) --(2.3,1.5)-- (p7);
    \draw[thick, ->] (p10) -- (p8);
    \draw[thick, ->] (p9) -- (p8);

    \end{scope}

 \draw [->] (.8,0)--(-.8,0) node[pos=0.5, below] {$\text{Bl}_{p_1\dots p_{7}}$};
  
    	\draw [->] (4,-3.5)--(4,-2.7) node[pos=0.5, left] {$\text{Bl}_{p_8}$};

    \begin{scope}[xshift = +4cm, yshift=-6cm]
    \draw[thick,blue] 
    (-2,-2.3)--(-2,1.75); 
        \draw[thick,red] 
    (2,-2.3)--(2,0.1) ;
    \draw[thick,blue] 
    (-2.25,1.5)--(1.6,1.5);
    \draw[thick,blue] 
      (-.66,-2.3)--(-.66,1.75);  
    \draw[thick,blue] 
      (.66,-2.3)--(.66,1.75); 
    \draw[thick]  (-2.25,-2)--(2.25,-2)  ; 
    \draw[red, thick] (-2.5,-1) -- (-1.5,-1) node[pos=0,xshift=-5pt,yshift=-5pt] {\small $L_1$};
    \draw[red, thick] (-2.5,0) -- (-1.5,0) node[pos=0,xshift=-5pt,yshift=-5pt] {\small $L_2$};
    \draw[red, thick] (-1.16,-1.25) -- (-.16,-1.25) node[pos=0,xshift=-5pt,yshift=-5pt] {\small $L_3$};
    \draw[red, thick] (-1.16,-.25) -- (-.16,-.25) node[pos=0,xshift=-5pt,yshift=-5pt] {\small $L_4$};
    \draw[red, thick] (+.16,-1.5) -- (+1.16,-1.5) node[pos=0,xshift=-5pt,yshift=-5pt] {\small $L_5$};
    \draw[red, thick] (+.16,-0.5) -- (+1.16,-0.5) node[pos=0,xshift=-5pt,yshift=-5pt] {\small $L_6$};
    
    \draw[blue, thick] (1.0,0.5) --(2.25,-0.75);      
    \draw[red, thick] (1.25, 0) -- (1.25,1.75);
      \node (p9) at (1.25,1.25) [basept,label={[xshift=10pt, yshift = -10 pt] \small $p_{9}$}] {};
      \node (p10) at (1.25,0.725) [basept,label={[xshift=11pt, yshift = -10 pt] \small $p_{10}$}] {}; 
    \node[above,red] at (1.25,1.75) {\small $L_8$};
    \node[right,blue] at (2.25,-0.75) {\small $L_7$};  
    \end{scope}

 \draw [<-] (.8,-12)--(-.8,-12) node[pos=0.5, below] {$\text{Bl}_{ p_{12}}$};
 \draw [<-] (.8,-6)--(-.8,-6) node[pos=0.5, below] {$\text{Bl}_{p_{9} p_{10}}$};

\draw [->] (-4,-3.5)--(-4,-2.7) node[pos=0.5, left] {$\text{Bl}_{p_1\dots p_{10}}$};
  \begin{scope}[xshift = +4cm]
    \draw[thick,blue] 
    (-2,-1.75)--(-2,1.75);
     \draw[thick,black]    
    (-2.25,-1.5)--(2.25,-1.5);
        \draw[thick,red] 
    (2,-1.75)--(2,1.1) ;
    \draw[thick,red] 
    (-2.25,1.5)--(1.6,1.5);
    \draw[thick,blue] 
      (-.66,-1.75)--(-.66,1.75);  
    \draw[thick,blue] 
      (.66,-1.75)--(.66,1.75)  ; 
      
    \draw[red, thick] (-2.5,-.5) -- (-1.5,-.5) node[pos=0,xshift=-5pt,yshift=-5pt] {\small $L_1$};
    \draw[red, thick] (-2.5,+.5) -- (-1.5,+.5) node[pos=0,xshift=-5pt,yshift=-5pt] {\small $L_2$};
    \draw[red, thick] (-1.16,-.75) -- (-.16,-.75) node[pos=0,xshift=-5pt,yshift=-5pt] {\small $L_3$};
    \draw[red, thick] (-1.16,+.25) -- (-.16,+.25) node[pos=0,xshift=-5pt,yshift=-5pt] {\small $L_4$};
    \draw[red, thick] (+.16,-1) -- (+1.16,-1) node[pos=0,xshift=-5pt,yshift=-5pt] {\small $L_5$};
    \draw[red, thick] (+.16,0) -- (+1.16,0) node[pos=0,xshift=-5pt,yshift=-5pt] {\small $L_6$};
    \draw[red, thick] (1.0,1.75) --(2.25,0.5); 
    \node[right,red] at (2.25,0.5) {\small $L_7$};      

      \node (p8) at (1.25,1.5) [basept,label={[xshift=-7pt, yshift = -17 pt] \small $p_{8}$}] {};
      \node (p9) at (1.25,2.25) [basept,label={[xshift=-7pt, yshift = -3 pt] \small $p_{9}$}] {};
      \node (p10) at (1.875,1.875) [basept,label={[xshift=10pt, yshift = -5 pt] \small $p_{10}$}] {};
    \draw[thick, ->] (p10) -- (p8);
    \draw[thick, ->] (p9) -- (p8);
    \end{scope}
    \begin{scope}[xshift =  4cm, yshift=-12cm]
    \node at (-2.5,0.9) {\large $X_t$};
    \draw[thick,blue] 
    (-2,-2.3)--(-2,1.75); 
        \draw[thick,blue] 
    (2,-2.3)--(2,1.75) ;
    \draw[thick,blue] 
    (-2.25,1.5)--(2.25,1.5);
    \draw[thick,blue] 
      (-.66,-2.3)--(-.66,1.75);  
    \draw[thick,blue]  (.66,-2.3)--(.66,1.75)  ;  
    \draw[thick]  (-2.25,-2)--(2.25,-2)  ; 
    \draw[red, thick] (-2.5,-1) -- (-1.5,-1) node[pos=0,xshift=-5pt,yshift=-5pt] {\small $L_1$};
    \draw[red, thick] (-2.5,0) -- (-1.5,0) node[pos=0,xshift=-5pt,yshift=-5pt] {\small $L_2$};
    \draw[red, thick] (-1.16,-1.25) -- (-.16,-1.25) node[pos=0,xshift=-5pt,yshift=-5pt] {\small $L_3$};
    \draw[red, thick] (-1.16,-.25) -- (-.16,-.25) node[pos=0,xshift=-5pt,yshift=-5pt] {\small $L_4$};
    \draw[red, thick] (+.16,-1.5) -- (+1.16,-1.5) node[pos=0,xshift=-5pt,yshift=-5pt] {\small $L_5$};
    \draw[red, thick] (+.16,-0.5) -- (+1.16,-0.5) node[pos=0,xshift=-5pt,yshift=-5pt] {\small $L_6$};

    \draw[red, thick] (1.5,1.25) -- (2.5,1.25) node[pos=1,xshift=7pt,yshift=0pt] {\small $L_9$};

    \draw[red, thick] (1.5,.725) -- (2.5,.725) node[pos=1,xshift=10pt,yshift=0pt] {\small $L_{10}$};
	 \node (p12) at (2,-2) [basept,label={[xshift=9pt, yshift = -16 pt] \small $p_{12}$}] {};  
    \node[blue,right] at (2,-1) {\small $L_8$}; 
    \end{scope}
    \begin{scope}[xshift = -4cm, yshift=-6cm]
    \draw[thick,blue] 
    (-2,-2.3)--(-2,1.75); 
        \draw[thick,red] 
    (2,-2.3)--(2,0.1) ;
    \draw[thick,blue] 
    (-2.25,1.5)--(1.6,1.5);
    \draw[thick,blue] 
      (-.66,-2.3)--(-.66,1.75);  
    \draw[thick,blue]  (.66,-2.3)--(.66,1.75)  ;  
    \draw[thick]  (-2.25,-2)--(2.25,-2)  ; 
    
    \draw[red, thick] (-2.5,-1) -- (-1.5,-1) node[pos=0,xshift=-5pt,yshift=-5pt] {\small $L_1$};
    \draw[red, thick] (-2.5,0) -- (-1.5,0) node[pos=0,xshift=-5pt,yshift=-5pt] {\small $L_2$};
    \draw[red, thick] (-1.16,-1.25) -- (-.16,-1.25) node[pos=0,xshift=-5pt,yshift=-5pt] {\small $L_3$};
    \draw[red, thick] (-1.16,-.25) -- (-.16,-.25) node[pos=0,xshift=-5pt,yshift=-5pt] {\small $L_4$};
    \draw[red, thick] (+.16,-1.5) -- (+1.16,-1.5) node[pos=0,xshift=-5pt,yshift=-5pt] {\small $L_5$};
    \draw[red, thick] (+.16,-0.5) -- (+1.16,-0.5) node[pos=0,xshift=-5pt,yshift=-5pt] {\small $L_6$};
    
    \draw[blue, thick] (1.0,0.5) --(2.25,-0.75);      
    \draw[black!60!green, thick] (1.25, 0) -- (1.25,1.75);
    \node[above,black!60!green] at (1.25,1.75) {\small $L_8$};

    \draw[red, thick] (1.0,1.25) -- (2,1.25) node[pos=1,xshift=7pt,yshift=0pt] {\small $L_9$};

    \draw[red, thick] (1.0,.725) -- (2,.725) node[pos=1,xshift=10pt,yshift=0pt] {\small $L_{10}$}; 
    \node[right,blue] at (2.25,-0.75) {\small $L_7$}; 

    \node at (-2.5,0.9) {\large $S_t$};
    \end{scope}

    \begin{scope}[xshift = -4cm,yshift = -12cm]
    \draw[thick,blue] 
    (-2,-2.3)--(-2,1.75);
     \draw[thick,black]    
    (-2.25,-2)--(2.25,-2);
        \draw[thick,red] 
    (2,-2.3)--(2,1.1) ;
    \draw[thick,blue] 
    (-2.25,1.5)--(1.6,1.5);
    \draw[thick,blue] 
      (-.66,-2.3)--(-.66,1.75);  
    \draw[thick,blue] 
      (.66,-2.3)--(.66,1.75)  ; 
    \draw[red, thick] (-2.5,-.5) -- (-1.5,-.5) node[pos=0,xshift=-5pt,yshift=-5pt] {\small $L_1$};
    \draw[red, thick] (-2.5,+.5) -- (-1.5,+.5) node[pos=0,xshift=-5pt,yshift=-5pt] {\small $L_2$};
    \draw[red, thick] (-1.16,-.75) -- (-.16,-.75) node[pos=0,xshift=-5pt,yshift=-5pt] {\small $L_3$};
    \draw[red, thick] (-1.16,+.25) -- (-.16,+.25) node[pos=0,xshift=-5pt,yshift=-5pt] {\small $L_4$};
    \draw[red, thick] (+.16,-1) -- (+1.16,-1) node[pos=0,xshift=-5pt,yshift=-5pt] {\small $L_5$};
    \draw[red, thick] (+.16,0) -- (+1.16,0) node[pos=0,xshift=-5pt,yshift=-5pt] {\small $L_6$};

    \node[black!60!green,above] at (1.0,1.75) {\small $L_8$};
    \draw[black!60!green, thick] (1.0,1.75) --(2.25,0.5);  
    \draw[red, thick] (0.95,0.65) --(1.7,1.4) node[pos=0,xshift=0pt,yshift=12pt] {\small $L_{9}$};   
    \draw[red, thick] (1.15,0.35 )--(1.9,1.1) node[pos=0,xshift=13pt,yshift=1pt] {\small $L_{10}$};  
	 \node (p11) at (2,-2) [basept,label={[xshift=9pt, yshift = -16 pt] \small $p_{11}$}] {}; 
    \node[red,right] at (2,-1) {\small $L_7$}; 
    \end{scope}
\draw [<-] (-4,-9.5)--(-4,-8.7) node[pos=0.5, left] {$\text{Bl}_{p_{11}}$}; 
\end{tikzpicture}
\caption{Sequence of blow ups and contraction from $(\Pj^1\times \Pj^1)_t$ to $X_t$. Curves of self-intersection $-1,-2,-3$ are coloured in red, blue, green respectively. See \Cref{rem:stricttransforms} for notation.
}
\label{fig:surface:syst1}
\end{figure}

\subsection{Relation between the first 2D system and $\pain{VI}$}

Following the procedure in \cite{Stud}, we can obtain a change of variables between the 
system~\eqref{syst1} and the Hamiltonian form \eqref{systfg} of $\pain{VI}$. This is the content of the following theorem and it is a part of \Cref{thm:main}. 
  
\begin{theorem}\label{Th1}
System \eqref{syst1} transforms to the Hamiltonian form \eqref{systfg} of the sixth Painlev\'e equation with parameters
\begin{equation} \label{rootvars}
a_0= -\gamma, \;a_1 =\kappa, \;  
a_2= -n - \alpha,\; a_3= -\beta, \; 
a_4=  
  \alpha,
  \end{equation}
via the transformation 
\begin{gather} 
    f = \frac{ \kappa + (t-1)\tilde{x}}{\kappa} ,\label{identsyst11}\\
    g =  \frac{(t-1)\left(\alpha- \beta - \gamma - 1\right)\tilde{x}^2 -  \kappa \left( 1- \alpha + 2 \beta + \gamma + (\alpha - \beta -1)t \right) \tilde{x}  + \kappa \left( \beta \kappa + (t-1) t \tilde{y} \right) }{2(t-1)(\tilde{x} - \kappa)\tilde{x}}.\label{identsyst12}
\end{gather}
The inverse transformation is given by
\begin{gather}
\tilde{x} = \frac{\kappa  (f-1)}{t-1},\\
\tilde{y} = \frac{\kappa \left[ (1-\alpha)t +  \left(  \alpha -\gamma - 1 + (\alpha - \beta  - 1 )t \right) f + \left(1- \alpha + \beta + \gamma \right) f^2 + 2 (f-1)(f-t) g\right]}{(t-1)^2 t} .
\end{gather}
\end{theorem}
This result comes from an isomorphism between $X$ constructed in subsection \ref{subsec:spaceofinitialconditionssys3.2} and the space of initial conditions $X^{\operatorname{KNY}}$ for the KNY form of $\pain{VI}$ in \Cref{appendix:standardP6}.
Interpreting the transformation (\ref{identsyst11}-\ref{identsyst12}) as a birational map $\C^2_{\tilde{x},\tilde{y}}\times B \dashrightarrow \C^2_{f,g}\times B$, this pulls back to a birational map $X\rightarrow X^{\operatorname{KNY}}$, which is verified to be an isomorphism by checking in charts.
Then $D$ and $D^{\operatorname{KNY}}$ coincide under the pullback, as do the rational vector fields defined by systems \eqref{syst1} and \eqref{systfg}.

\begin{remark}   \label{rem:inaccessible2}
    The identification in \Cref{Th1} also confirms that $D_t\subset X_t$ from the space of initial conditions constructed above for system \eqref{syst1} are inaccessible, since $D_t^{\operatorname{KNY}}$ are for system \eqref{systfg} on $X^{\operatorname{KNY}}$.
\end{remark}
In particular, the function $f$ defined in terms of the solution $(\tilde{x},\tilde{y})$ of the system \eqref{syst1} solves the sixth Painlev\'e equation \eqref{P6} with parameters 
\begin{equation}\label{par MinChen}
A = \frac{ \kappa^2}{2}, \;B =- \frac{ \alpha^2}{2},\; C = \frac{\beta^2}{2},\;  D= \frac{1}{2} (1 - \gamma^2).
\end{equation}
Thus, we recover \cite[Th. 3.5]{MinChen}. 

\begin{remark}
    We remark that the identification of the system \eqref{syst1} with the Hamiltonian form \eqref{systfg} of $\pain{VI}$ presented in \Cref{Th1} is not unique. Indeed, the 
    identification is up to the action of the extended affine Weyl group of B\"acklund 
    transformations leaving $\pain{VI}$ invariant. In \Cref{Th1} as a target representative
    we choose the same copy of $\pain{VI}$ as in \cite{MinChen}, i.e. with 
    parameters~\eqref{par MinChen}. We observe, that this is not the only possible choice.
    In particular, in the theory of orthogonal polynomials there might be additional considerations to choose a representative properly. For instance, 
    matching the evolution in $n$ of the recurrence coefficients with some standard example of a discrete Painlev\'e equation, or relating initial conditions for the recurrence coefficients with the solutions of the Riccati equation which solves the sixth Painlev\'e equation for a particular choice of parameters.
\end{remark}

\section{The system on the hypersurface in the second parametrisation} 
\label{sec:hypersurf}
In this section we analyse the system \eqref{syst3D} subject to the constraint \eqref{hypersurfaceconstraint} as its restriction to a Darboux surface. First, in \Cref{prop:darboux} we prove the existence of a Darboux surface and then, in \Cref{cor:secondsystemtosystfg}, we identify the restriction with $\pain{VI}$, so proving part of \Cref{thm:main}.

\subsection{The algebraic Darboux hypersurface \texorpdfstring{$S_t$}{}}
\label{ssec:darb}

In \Cref{sss:minchen} we recalled how the system \eqref{syst3D} and the invariant
algebraic surface $S_t$ \eqref{eq:St} arise from the theory of degenerate Jacobi unitary 
polynomials. In this subsection we show that the hypersurface $S_t$ \eqref{eq:St} is
defined by a \emph{Darboux polynomial} for the system \eqref{syst3D}. 

We recall that, given a system of (non-autonomous) first-order 
differential equations:
\begin{equation}
    \vec{x'} = \vec{F}_t(\vec{x}),
    \label{eq:firstord}
\end{equation}
a Darboux function $P_t=P_t(\vec{x})$ is a scalar function such that, on each solution  of \eqref{eq:firstord}, we have:
\begin{equation}
    P_t'(\vec{x}) = C_t(\vec{x})P_t(\vec{x}),
\end{equation}
for some function $C_t$,  
see~\cite{Darboux1878,PrelleSinger1983} and the book~\cite[Section 2.5]{Goriely2001} 
for a modern account of this theory.
The function $C_t$ is usually called the \emph{cofactor} of $P_t$. Moreover, when $P_t$ is a polynomial, 
it is usually called a \emph{Darboux polynomial}, and its zero locus
is said to be a \textit{Darboux hypersurface}.

\begin{remark}
    We observe that looking for Darboux polynomials is a non-trivial
    task. In the autonomous setting the problem of finding Darboux 
    polynomials of an assigned degree can be reduced to a problem of solving 
    some polynomial equations, which can be addressed 
    using tools from computational algebra, see e.g.~\cite{antonovetal2019}. 
    On the other hand, in the non-autonomous setting finding Darboux 
    polynomials of an assigned degree is a much harder problem involving 
    the solution of a system of algebraic differential equations.
\end{remark}

The following proposition guarantees the existence of a Darboux polynomial
for the system~\eqref{syst3D}, and it is part of \Cref{thm:main}.

\begin{proposition}
    \label{prop:darboux}
    The polynomial $h_t$ in \eqref{eq:h} is a Darboux polynomial for
    the system \eqref{syst3D}, i.e.:
    \begin{equation} \label{eq:diffeqforhyper}
        h_t' = C_t h_t, 
    \end{equation}
    where the cofactor $C_t$ is a rational function of $t$ and $\vec{x}$.
\end{proposition}

\begin{proof}
    The proof consists of a direct computation which can be 
    carried out with a computer algebra system, e.g.\ Maple~\cite{Maple} or 
    Mathematica~\cite{Mathematica}. For the sake of readability we omit the 
    explicit expression of $C_t$.
\end{proof}
 
We observe that the existence of a Darboux surface gives information about the solutions 
of the system \eqref{syst3D}.

First of all, let us notice that given  a solution $\vec{x}(t)=(x,y,z)(t)$ of the system \eqref{syst3D}, we can write the general solution of \eqref{eq:diffeqforhyper} as
\begin{equation}
    h_t = h_*  \eta_t, \quad \eta_t= \exp \left(\int^t C_\xi(\vec{x}(\xi)) d\xi\right), 
    \label{eq:htsol}
\end{equation}
where $h_* $ is a constant of integration. Clearly, if $h_* \equiv 0$ then the solution 
$\vec{x}(t)$ belongs to {the surface for all $t\in\C\setminus\Set{0,1}$}.

Second, the asymptotic behaviour near a movable singularity can be used to determine 
whether or not a solution stays in the surface. Indeed, let us assume that we have a 
solution $\vec{x}_*(t)$ of the system \eqref{syst3D}, admitting a pole-like 
movable singularity at $t=t_* $. Then, since $h_t$ is a polynomial we  have 
$h_t\sim (t-t_* )^\lambda$ for $\lambda\in\Z$ in a neighbourhood of $t_*$. So, if a 
solution of the system \eqref{syst3D} is devoid of essential singularities then 
the functions $h_t$ and $\eta_t$ are the same up to a potentially vanishing constant 
factor, i.e. $h_t\sim \eta_t$ or $h_t=0$. On the other hand, this forces 
$C_t = \mu (t-t_*)^{-1}+\mathrm{O}(1)$ with $\mu\in\Z$. 
If, near a movable singularity $t=t_*$, the cofactor $C_t$ has a different asymptotic behaviour, then $h_t\equiv 0$ along $\vec{x}_*(t)$.
 
For instance, the system \eqref{syst3D} admits a class of solutions $\vec{x}_*(t)$ with the following behaviour in a neighbourhood of a movable singularity at $t=t_* $:
\begin{equation} \label{eq:galinaseries}
    \begin{aligned}
        x_* (t) &= t_* (t-t_* )^{-1} + a_0 + A(a_0) (t-t_* ) + \mathrm{O}\left((t-t_* )^2\right),  \\
        y_* (t) &= \frac{t_* (t_* -1)^2}{\kappa-1} (t-t_* )^{-2} + \frac{2t_* (t_* -1)}{\kappa-1} (t-t_* )^{-1} + \mathrm{O}\left(1\right), \\
        z_* (t) &= \frac{t_* ^2(t_* -1)}{\kappa-1} (t-t_* )^{-2} + \frac{2t_* (t_* -1)}{\kappa-1} (t-t_* )^{-1} + \mathrm{O}\left(1\right).
    \end{aligned} 
\end{equation}
Here $a_0$ is an arbitrary constant and $A(a_0)$ is a known rational function of $a_0$ that can 
be computed with a computer algebra system, e.g. Maple~\cite{Maple} or 
Mathematica~\cite{Mathematica}, but whose 
expression we omit for the sake of readability. Then, in a neighbourhood of $t=t_*$ 
the cofactor $C_t$ with respect to the solution  \eqref{eq:galinaseries} behaves as:
\begin{equation}
    C_t = - \frac{\kappa-3}{\kappa-1} \frac{1}{t-t_* } + \mathrm{O}(1),
\end{equation}
implying $\eta_t\sim (t-t_* )^{-\frac{\kappa-3}{\kappa-1}}$. Since for generic values of $\kappa$ the behaviour of $\eta_t$  gives raise to a branch point, we conclude that, for the class of solutions in \eqref{eq:galinaseries} we have $h_* \equiv 0$. Hence, the family is entirely contained in the surface $S_t$.

As stated in the Introduction, in what follows we will prove that the
system~\eqref{syst3D} restricts to $S_t$ in such a way that the resulting
two-dimensional system possesses the Painlev\'e property.  Whether or not
the ``full'' system \eqref{syst3D} might possess the Painlev\'e property
remains an open problem, see \Cref{conj:3d}.  We postpone its study to
subsequent works. However, in \Cref{sec:autlim} we will present some strong
evidence for it. That is, we will prove that there exists an algebraically
integrable autonomous limit of the system \eqref{syst3D} whose dynamics is
associated to an elliptic fibration. This is a strong indication that the
system \eqref{syst3D} possesses the Painlev\'e property since the
two-dimensional Painlev\'e equations can be constructed as
de-autonomisations of two-dimensional systems admitting an invariant
elliptic fibration~\cite{sakaiODEsonrationalellipticsurfaces}.

\subsection{Parametrisation of the hypersurface }\label{subsec:parametrisation2} 
In this section, we produce a parametrisation for the surface $S_t$. The study of the singularities of $S_t$ is crucial in our construction. The issue in finding the parametrisation is that this surface is  singular making it hard to find a rational parametrisation. Our strategy to overcome this problem is to blow up the surface $S_t$ along a non-Cartier divisor in order to make it less singular and then to look for a parametrisation, see~\cite[Example IV-27]{EisenbudHarrisBook} 
and~\cite[App. A]{DeMarcoetal2024}. 

\begin{definition}
    Let $X$ be a quasi-projective variety. Recall that a Weil divisor is a formal $\mathbb{Z}$-linear combination of irreducible codimension-one subvarieties of $X$. 
\end{definition}

\begin{remark}
    Recall that, when a codimension-one subvariety $Y\subset X$ is Zariski locally defined by one equation it is
     a Cartier divisor, see \cite[Section II.6]{Hartshorne2013}. Similarly given a Cartier divisor $G=\Set{f_1^{m_1}\cdots f_s^{m_s}=0}$, with $f_i$ prime for all $i=1,\ldots,s$, we can cook up the Weil divisor
     \[
     D_G=\sum_{i=1}^sm_i\cdot \Set{f_i=0}.
     \]
     We will often implicitly make use of this identification. 
\end{remark}

\begin{remark}
    Whenever $X$ is smooth, the notions of Weil and Cartier divisors agree. For an example of a Weil divisor which is non-Cartier, consider the line $\Set{x=z=0}$ on the quadric cone $\Set{xy-z^2=0}\subset \mathbb{C}^3$.
\end{remark}

\begin{definition}\label{def:blowup weil}
    Let $Y\subset X$ be a variety and a Weil divisor. The blow up of $X $ centered at $Y$ is a birational projective morphism 
    of varieties
    \begin{equation}
    \begin{tikzcd} \Bl_YX\arrow[r,"\varepsilon" ]& X
    \end{tikzcd}    
    \end{equation}
    such that the preimage\footnote{Here the notion of preimage is purely algebraic and the precise definition can be found in \cite[Caution II-7.12.2]{Hartshorne2013}. However, we omit it as it is  unnecessary for our purpose.} $\varepsilon^{-1} Y $ is a Cartier divisor, called the \textit{exceptional divisor}, and such that any other morphism having the same property factors through $\varepsilon$.
\end{definition}  

\begin{remark}
    The blow up of a Weil divisor always exists and it is unique up to unique isomorphism, see \cite{EisenbudHarrisBook}. As for the basic case of the blow up of a singular point on a surface, the blow up of a non-Cartier divisor is a powerful tool in the theory of resolution of singularities.

    All the instances we consider will be blow ups of $X$ with center $Y$, where $Y$ is a  complete intersection in the ambient space where $X$ is located. In this case the blow up can be computed as explained in \cite[Propositions IV-21 \& IV-25]{EisenbudHarrisBook}.
\end{remark}

We start by considering a suitable compactification $\overline{S}_t$ of $S_t$. Namely, first we consider the inclusion
\begin{equation}
\begin{tikzcd}[row sep = tiny]
    S_t\arrow[r,hook]& \Pj^2_{[y_0:y_1:y_2]}\times \Pj^1_{[x_0:x_1]}\\
    (x,y,z)\arrow[r,mapsto] & ({[y:z:1]},{[x:1]})
\end{tikzcd}    
\end{equation}
and then we consider the closure $\overline{S}_t\subset \Pj^2_{[y_0:y_1:y_2]}\times \Pj^1_{[x_0:x_1]}$. 

\begin{remark}
    We stress that, as in previous sections, we omit the $t$-dependence implicitly assuming we are working with the fibres of a trivial $B$-bundle, e.g. in the case $(\mathbb C^2\times\Pj^1)\times B$, we write $\mathbb C^2\times\Pj^1 $ in place of $(\mathbb C^2\times\Pj^1)_t$ to denote the fibre over $t\in B$.
\end{remark}

In technical terms,  the equation of the surface $\overline{S}_t$ can be obtained by replacing 
\begin{equation}
    y=\frac{y_0}{y_2},\;z=\frac{y_1}{y_2},\;x=\frac{x_0}{x_1} 
\end{equation}
in $h_t$, see \eqref{eq:h}.  A direct check shows that the  hypersurface $\overline{S}_t$ has only isolated singularities.

We focus now on the affine chart
\begin{equation}
S_{1,1,t}=\Set{([y_0:y_1:y_2], [x_0:x_1])\in \overline{S}_t| x_1,y_1\not=0 }\subset \overline{S}_t.    
\end{equation}
The surface $S_{1,1,t}$  is singular at
\begin{equation}
    P_1=\Set{x_0+\frac{\kappa}{t-1}=\,y_0=y_2=0},\;P_2=\Set{x_0-\kappa=y_0-1=y_2=0}.
\end{equation}

We blow up now the non-Cartier divisor 
\begin{equation}
    Y=\Set{ \;x_0(1-t+t y_0)-\kappa=y_2=0}\subset S_{1,1,t}    
\end{equation}
passing through both $P_1$ and $P_2$. A direct check shows that this blow up resolves the 
two singularities $P_1,P_2\in S_{1,1,t}$. We also remark that the curve $Y$ is an 
indeterminacy of system  \eqref{syst3D}. The surface $S_{1,1,t}$ lives in the affine space 
$\mathbb C^3$ with coordinates $y_0,y_2,x_0$. Hence, the blow up $\Bl_Y S_{1,1,t}$ lives in $
\Bl_Y \mathbb C^3$ \cite[Propositions IV-21]{EisenbudHarrisBook}. We have 
\begin{equation}
 \Bl_Y \mathbb C^3 =\Set{ ((y_0,y_2,x_0)(t),[u :v](t))\in\mathbb \C^3 \times \Pj^1 | \rank\begin{pmatrix}
     u & v\\
     y_2 & x_0(1-t+t y_0)-\kappa
 \end{pmatrix}\le 1 }.    
\end{equation}

Since we look for a rational parametrisation, it is worth restricting to the chart $\tilde S _t\subset \Bl_Y S_{1,1,t}$ given by
\begin{equation}
\tilde S_t =\Set{(y_0,y_2,x_0,[u:v])\in \Bl_Y S_{1,1,t}| u\not=0}.    
\end{equation}
The equation of $\tilde S_t$ is derived  by replacing 
\begin{equation}
    y_2=v(x_0(1-t+ty_0)-\kappa)
\end{equation}
in the equation of $S_{1,1,t}$ and dividing by the equation of the exceptional divisor in $ \Bl_Y\mathbb C^3$.

Finally, we notice that  $y_0$ is a rational function of $x_0$ and $v$.
This gives us the surface parametrisation. Explicitly:
\begin{equation}
    y_0=\frac{(\kappa +(t-1) x_0 ) \left(1+n t v^2 x_0 (\gamma +n) \left(\kappa +(t-1) x_0\right){} -v   \left(\beta
    \kappa -x_0 (\beta +2 n t+\gamma  t)\right)\right)}{t x_0 \left(n t v^2 x_0 (\gamma +n) \left(\kappa +(t-1)
   x_0\right)+v \left(\kappa  (\kappa-\beta-1)+x_0 (-\alpha +(2 n +\gamma)(t-1))\right)+1\right)}.
    \label{eq:y0par}
\end{equation}

\subsection{The 2D system in the second parametrisation}

Now, we study the restriction, to the hypersurface $\tilde S_t$, of the lift to the blow up $\Bl_Y\mathbb
C^3 $ of the  system  \eqref{syst3D}.
Via the parametrisation obtained above, it consists of two equations in the variables $x_0$ and $v$. Namely, by
direct computation from equation \eqref{eq:y0par} the system reads:
\begin{equation} 
    \begin{aligned}
        x_0' &= \frac{v \left(  ((t-1) x_0 +\kappa)(x_0-\beta) -tx_0 (\alpha +1)
   \right)+2}{(t-1) t v},
       \\
        v' &=-\frac{\eta(x_0,v,t)}{(t-1) t x_0 \left(\kappa -x_0\right) \left(\kappa +(t-1)
   x_0\right)},
    \end{aligned}
    \label{second_system}%
\end{equation}
where
\begin{equation}
    \begin{aligned}
        \eta(x_0,v;t) &=\kappa ^2-n t v^2 x_0^2 (\gamma +n) \left(\kappa +(t-1) x_0\right){}^2
        \\
        &+v \left(x_0 \left(x_0 \left((t-1) x_0
   \left(-2 (\kappa +\beta)+t (2 \alpha +2 \beta +1)-(t-1) x_0\right) \right.\right.\right. 
        \\ 
        &\left. \left. \left.-\kappa
       (5 \beta +\kappa +t (\alpha  (t-2)+\beta  (t-6)-1) )\right) \right.\right. 
       \\
       &\qquad   \left. \left. -2 \beta  \kappa ^2 (t-2)\right)-\beta\kappa ^3\right)+2 \kappa  (t-2) x_0-3 (t-1) x_0^2.
    \end{aligned}
    \label{second_system_eta} 
\end{equation} 
\begin{remark}
    The birational transformation between $\tilde  S_t$ and $S_t$ is
    \begin{equation}
    \begin{tikzcd}[row sep=tiny]
        \tilde S_t \arrow[r,dashed] & S_t\\
        (x_0,v)\arrow[r,mapsto]& \left(x_0,\frac{-y_0}{v \left(\kappa +x_0 \left(t-1-t y_0\right)\right)},\frac{-1}{v
   \left(\kappa +x_0 \left(t-1-t y_0\right)\right)}\right)
    \end{tikzcd}
    \end{equation}
    and it has inverse
    \begin{equation}
    \begin{tikzcd}[row sep=tiny]
         S_t \arrow[r,dashed] &\tilde S_t\\
        (x,y,z)\arrow[r,mapsto]& \left(x_0,\frac{1}{x  (t y -t z +z )-\kappa    z }\right).
    \end{tikzcd}
    \end{equation}
\end{remark}
 
Now,   \Cref{Th1} implies that given a solution $(f,\,g)$ of system \eqref{systfg} with parameters chosen according to \eqref{rootvars}, the triple $x,y,z$ defined by 
    \begin{equation} \label{injectionfg}
        \begin{aligned}
            x &= \frac{\kappa(f-1)}{(t-1)}, \\
            y &= \frac{(f^2-g(1+t))\left(\alpha + n - g\right)^2 + t\left(\alpha - g \right)\left(\kappa - 1 - g\right) - g (\alpha +n-g)\left(  \gamma+\beta    \right) }{(\kappa-1)t}, \\
            z &= \frac{(f-1)\left(f^2\left(\alpha+n-g\right)^2 + t g \left(g-\alpha \right) - f\left(g-\alpha-n\right)\left( (n+\gamma )-(1-t)-t\alpha + (t+1) g\right)\right) }{(t-1)(\kappa-1)f},
        \end{aligned}
    \end{equation} 
    is a solution of the system  \eqref{syst3D} lying on $S_t$, and we have 
\begin{equation}\label{invinj}
f=\frac{\kappa+(t-1)x}{\kappa},\;\;g=n+\alpha+\frac{\kappa((\kappa+(t-1)x)z-t x (n+y))}{(t-1)(\kappa-x)x}.
\end{equation}

Note that the birational transformations 
\begin{equation}
    \begin{tikzcd}[row sep = tiny]
        \mathbb C^2_t\arrow[r,dashed,"\Psi"] & S_{t}\\
        (f,g)\arrow[r,mapsto] & (x,y,z)
    \end{tikzcd} \quad 
    \begin{tikzcd}[row sep = tiny]
       S_t \arrow[r,dashed,"\Phi"] &  \mathbb C^2_t\\
        (x,y,z)\arrow[r,mapsto] & (f,g)
    \end{tikzcd}  
\end{equation}
defined in \eqref{injectionfg} and \eqref{invinj}  are inverse to each
other, i.e. they satisfy
\begin{equation}
    \Phi\circ\Psi=\Id_{\C_t^2}\mbox{ and }
     \Psi\circ\Phi=\Id_{ S_t}.
\end{equation}
Summing up, we obtain the following equivalence between system
\eqref{second_system} and the Hamiltonian form of $\pain{VI}$,
which is part of \Cref{thm:main}.

\begin{theorem} 
    \label{cor:secondsystemtosystfg} 
    The systems \eqref{second_system} and \eqref{systfg} are related by the change of variables
    \begin{equation}\label{eq:from3.8to2.2}
        f=\frac{(t-1)x_0 + \kappa}{\kappa},\quad g= \frac{\kappa + x_0  v \left( (\alpha+n)(x_0(t-1) +\kappa)-\alpha t\kappa \right) }{(t-1) x_0 (x_0 - \kappa)v },
    \end{equation}
    with parameters identified in the same way as in   \Cref{Th1}, i.e. according to 
\begin{equation}
a_0= -\gamma, \quad a_1 =\kappa, \quad  
a_2= -n - \alpha,\quad a_3= -\beta, \quad 
a_4=  
  \alpha.
  \end{equation}
\end{theorem}

\begin{remark}
    The change of variables inverse to \eqref{eq:from3.8to2.2} is 
\begin{equation}
    x_0=\frac{\kappa (f-1)}{(t-1)},\quad v=-\frac{1-t}{\kappa(f-1)\left( f(\alpha+n-g) + t(g-\alpha) \right)}.
\end{equation}

\end{remark}

\begin{remark}
    The transformations to the KNY Hamiltonian form \eqref{systfg} from system \eqref{syst1} in \Cref{Th1} and from system \eqref{second_system} in \Cref{cor:secondsystemtosystfg} lead to the following relation between the two parametrisations:
    \begin{equation}
        \tilde{x} = x_0, \quad \tilde{y} = \frac{2 + v\left( (t-1) x_0^2 + \left( \kappa +\beta - t(\alpha+\beta+1) \right)x_0 - \beta \kappa \right)}{t(t-1)v}.
    \end{equation}
\end{remark}

\subsection{Space of initial conditions}\label{subsec:soichyper}

We now turn to the problem of constructing the space of initial conditions
for the system \eqref{second_system}. For the sake of readability, first we
rescale the variable $x_0$ as $x_0\to \kappa x_0$ in system
\eqref{second_system}, and abuse notation  rather than introduce a new
symbol.  Then, for convenience we compactify the parametrisation to the
first \emph{Hirzebruch surface} $\BF_1$ \cite{Hirzebruch1951} instead of
the more common compactification to $\Pj^{1}\cross\Pj^{1}=\BF_{0}$, see
\Cref{appendix:eltransf}.

Following, for instance, \cite{Takenawa2001CommMathPhys} the surface $\BF_{1}$ can be
defined as the gluing of four affine charts with coordinates:
\begin{equation}
    \left\{
    \begin{aligned}
        (\xi_0,\upsilon_0) &=\left( \frac{1}{\xi_{1}}, \xi_{1} \upsilon_{1} \right)  = 
    \left( \frac{1}{\xi_{2}}, \frac{\xi_{2}}{\upsilon_{2}} \right)  = 
    \left(  \xi_3, \frac{1}{\upsilon_3} \right)
    \\
        (\xi_1,\upsilon_1)&= \left( \frac{1}{\xi_0}, \xi _0\upsilon _0\right)\\
        (\xi_2,\upsilon_2)&=\left( \frac{1}{\xi_0}, \frac{1}{\xi_0\upsilon_0} \right)\\
        (\xi_3,\upsilon_3)&=\left( \xi_0,\frac{1}{\upsilon_0} \right) .
    \end{aligned}
    \right.
    \label{eq:F1charts}
\end{equation} 
Notice that, in \eqref{eq:F1charts}, we have  given some   transition functions of the given atlas.  Note also that the lines $C_+=\Set{\upsilon_2=0}\cup\Set{\upsilon_3=0}$ and $C_-=\Set{\upsilon_0=0}\cup\Set{\upsilon_1=0}$  have self intersection 1 and -1 respectively, see \Cref{appendix:eltransf}. This can be seen by noting that the union of the first and the
second chart (resp. the third and the fourth chart) is isomorphic to the total space
of the line bundle $\mathcal{O}_{C_+}(1)$ (resp. $\mathcal{O}_{C_-}(-1)$), see \Cref{appendix:eltransf}.

Let us rewrite the system \eqref{second_system} in terms of the coordinates of the zeroth chart
of $\BF_{1}$ via $(x_0,v)=(\xi_0,\upsilon_0)$:
\begin{equation}
    \begin{aligned}
        \xi_0' &= \frac{\upsilon_{0} \left(\xi_0 \left(\beta+\kappa-t (\alpha +\beta +1)+(t-1) \xi_0\right)-\beta  \kappa
        \right)+2}{(t-1) t \upsilon_{0}},
       \\
       \upsilon_{0}' &=-\frac{\eta(\xi_0,\upsilon_{0},t)}{(t-1) t \xi_0
           \left(\kappa -\xi_0\right) \left(\kappa +(t-1)
   \xi_0\right)},
    \end{aligned}
    \label{second_systemF1}%
\end{equation}
where $\eta(\xi_0,\upsilon_{0},t)$ is the function in
\eqref{second_system_eta} evaluated at
$(x_{0},v,t)=(\xi_{0},\upsilon_{0},t)$.

Then, the system \eqref{second_systemF1} has the following indeterminacy points:
\begin{equation}
    \begin{gathered}
        \pi_{1}:\ (\xi_3,\upsilon_3)=(0,0), 
        \quad
        \pi_{2} :\ (\xi_3,\upsilon_3)= (0,\kappa \beta), 
        \\[4pt]
        \pi_{3} :\ (\xi_3,\upsilon_3)= (1,-tn),
        \quad
        \pi_{4} :\ (\xi_3,\upsilon_3)=
        (1, -(\gamma+n)t), 
        \\[4pt]
        \displaystyle
        \pi_{5}:\ (\xi_3,\upsilon_3)= \left(  -\frac{1}{t-1},0\right),
        \quad
        \displaystyle
        \pi_{6} :\ (\xi_3,\upsilon_3)=
        \left(  -\frac{1}{t-1},   -\frac{t\alpha}{t-1}\right)\\
    \pi_{7}:\ (\xi_1,\upsilon_1)=( 0, 0).
    \end{gathered}
    \label{eq:indsys2}
\end{equation}
The first six are contained in the third chart and resolved after a single blow up, while
indeterminacy point $\pi_7$ is located at the origin of the first chart and it requires one further blow up to be resolved.

Adopting the same notation as in \Cref{sec:firstpar}, we denote by  $(u_i,v_i)$ and 
$(U_i,V_i)$ the coordinate charts covering the exceptional divisor over the point $\pi_i$, for $i=1,\ldots,7$, see \eqref{eq:coordblowp}.
In these coordinates, the eighth indeterminacy point is given by:
\begin{equation}
    \pi_8:\  (u_7,v_7)=\left(\kappa(1-t), 0\right).
\end{equation}
After blowing up the point $\pi_8$ we get a surface $X _{t }$ 
whose configuration of $(-2)$-curves  $D_{t,i}$, for $i=0,\ldots,4$ is depicted in the bottom picture in  \Cref{fig:surface:syst2}. Precisely, they intersect according
to the $\D_4^{(1)}$ configuration associated with $\pain{VI}$ and their union $D_t=\bigcup_{i=0}^4D_{t,i}$ is the inaccessible divisor. 
After removing  $D_t $ from $X_t$, we get the space $E$ over $B$, with fibre $E_t = X_t \setminus D_t$, of which the system defines a uniform foliation.

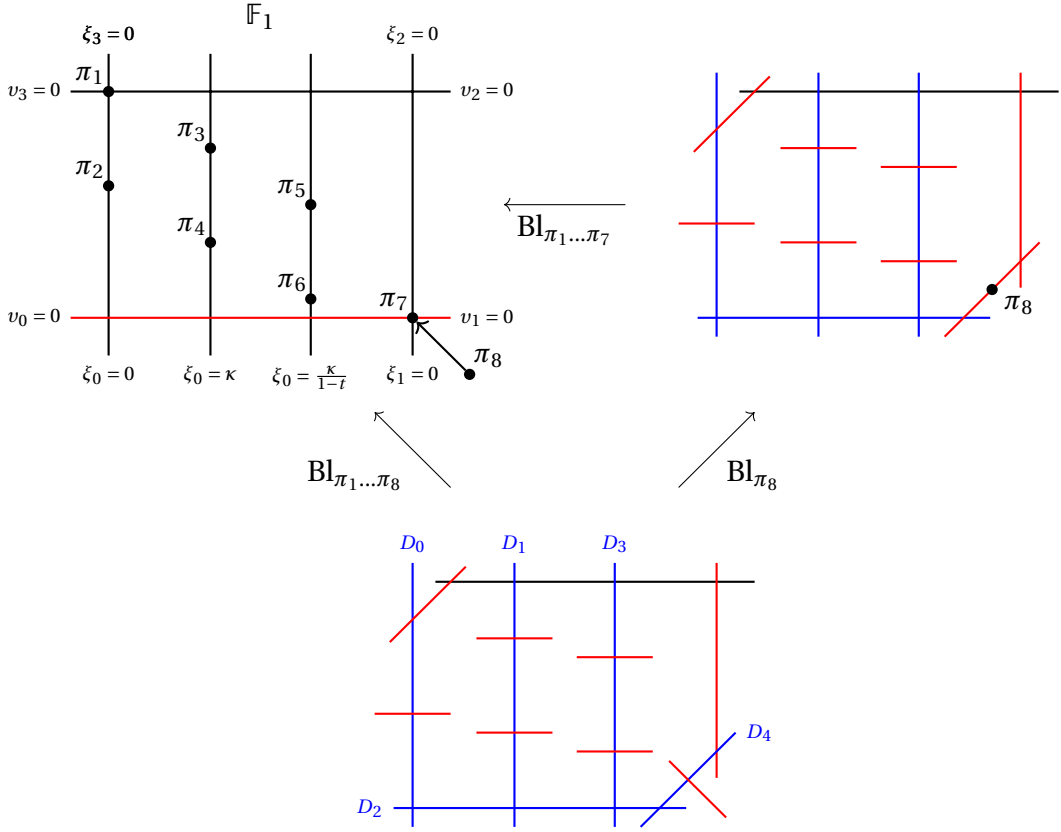
\begin{figure}[htb]
\centering
    \begin{tikzpicture}[basept/.style={circle, draw=black!100, fill=black!100, thick, inner sep=0pt,minimum size=1.2mm}]
    
     \begin{scope}[xshift = -4cm, yshift=+1cm]
    \node at (0,2.5) {$\mathbb{F}_1$};
    \draw[thick] 
    (-2,-2)--(-2,2)
    (2,-2)--(2,2) 
    (-2.5,1.5)--(2.5,1.5)
      (-.66,-2)--(-.66,2)  
      (.66,-2)--(.66,2)  
      ;
      \draw[thick, red]
          (-2.5,-1.5)--(2.5,-1.5) ;
      \node[left] at (-2.5,-1.5) {\tiny $\upsilon_0=0$};
      \node[right] at (2.5,-1.5) {\tiny $\upsilon_1=0$};
      \node[left] at (-2.5,1.5) {\tiny $\upsilon_3=0$};     
      \node[right] at (2.5,1.5) {\tiny $\upsilon_2=0$};     
      \node[below] at (-2,-2) {\tiny $\xi_0=0$};
      \node[above] at (-2,2) {\tiny $\xi_3=0$};
      \node[above] at (2,2) {\tiny $\xi_2=0$};
      \node[above] at (-2,2) {\tiny $\xi_3=0$};
      \node[below] at (2,-2) {\tiny $\xi_1=0$};
      \node[below] at (-.66,-2) {\tiny $\xi_0=\kappa$};
      \node[below] at (.66,-2) {\tiny $\xi_0=\frac{\kappa}{1-t}$};
      
	 \node (p1) at (-2,+1.5) [basept,label={[xshift=-7pt, yshift = -3 pt] \small $\pi_{1}$}] {};
	 \node (p2) at (-2,+.25) [basept,label={[xshift=-7pt, yshift = -3 pt] \small $\pi_{2}$}] {};
	 \node (p3) at (-.66,+.75) [basept,label={[xshift=-7pt, yshift = -3 pt] \small $\pi_{3}$}] {};	 
	 \node (p4) at (-.66,-.5) [basept,label={[xshift=-7pt, yshift = -3 pt] \small $\pi_{4}$}] {};
	 \node (p5) at (+.66,+0) [basept,label={[xshift=-7pt, yshift = -3 pt] \small $\pi_{5}$}] {};
      \node (p6) at (+.66,-1.25) [basept,label={[xshift=-7pt, yshift = -3 pt] \small $\pi_{6}$}] {};
    \node (p7) at (2,-1.5) [basept,label={[xshift=-7pt, yshift = -3 pt] \small $\pi_{7}$}] {};
      
	 \node (p8) at (2.75,-2.25) [basept,label={[xshift=7pt, yshift = -3 pt] \small $\pi_{8}$}] {};

    \draw[thick, ->] (p8) -- (p7);

    \end{scope}

 \draw [->] (.8,1)--(-.8,1) node[pos=0.5, below] {$\text{Bl}_{\pi_1\dots \pi_{7}}$};
 
    \begin{scope}[xshift = +4cm, yshift=+1cm]
   \draw[thick,blue] 
    (-2,-1.75)--(-2,1.75);
        \draw[thick] 
    (-1.7,1.5)--(2.5,1.5);
        \draw[thick,red] 
    (2,-1.1)--(2,1.75) ;
    \draw[thick,blue] 
    (-2.25,-1.5)--(1.6,-1.5);
    \draw[thick,blue] 
      (-.66,-1.75)--(-.66,1.75);  
    \draw[thick,blue] 
      (.66,-1.75)--(.66,1.75)  ;
    \draw[red, thick] (1.0,-1.75) --(2.25,-0.5);      

    \draw[red, thick] (-2.3,0.7) -- (-1.3,1.7) 
    ;
    \draw[red, thick] (-2.5,-.25) -- (-1.5,-.25) 
    ;
    \draw[red, thick] (-1.16,.75) -- (-.16,.75) 
    ;
    \draw[red, thick] (-1.16,-.5) -- (-.16,-.5) 
    ;
    \draw[red, thick] (+.16,.5) -- (+1.16,.5) 
    ;
    \draw[red, thick] (+.16,-.75) -- (+1.16,-.75) 
    ;
    \node (p8) at (1.625,-1.125) [basept,label={[xshift=10pt, yshift = -15 pt] \small $\pi_{8}$}] {};

    \end{scope}
    
 \draw [->] (1.5,-2.75)--(2.5,-1.75) node[pos=0.5, below right] {$\text{Bl}_{\pi_{8}}$};

    \begin{scope}[xshift = 0cm, yshift=-5.5cm]
   \draw[thick,blue] 
    (-2,-1.75)--(-2,1.75)node[pos=1,above] {\tiny $D_0$} ;
        \draw[thick] 
    (-1.7,1.5)--(2.5,1.5);
        \draw[thick,red] 
    (2,-1.1)--(2,1.75) ;
    \draw[thick,blue] 
    (-2.25,-1.5)--(1.6,-1.5)node[pos=0,left] {\tiny $D_2$} ;
    \draw[thick,blue] 
      (-.66,-1.75)--(-.66,1.75)node[pos=1,above] {\tiny $D_1$} ;  
    \draw[thick,blue] 
      (.66,-1.75)--(.66,1.75) node[pos=1,above] {\tiny $D_3$} ;
    \draw[blue, thick] (1.0,-1.75) --(2.25,-0.5) node[pos=1,right] {\tiny $D_4$} ;      

    \draw[red, thick] (-2.3,0.7) -- (-1.3,1.7) 
    ;
    \draw[red, thick] (-2.5,-.25) -- (-1.5,-.25) 
    ;
    \draw[red, thick] (-1.16,.75) -- (-.16,.75) 
    ;
    \draw[red, thick] (-1.16,-.5) -- (-.16,-.5) 
    ;
    \draw[red, thick] (+.16,.5) -- (+1.16,.5) 
    ;
    \draw[red, thick] (+.16,-.75) -- (+1.16,-.75) 
    ;
    \draw[red, thick]
    (1.375,-.875) -- (2.125,-1.625);

    \end{scope}
 \draw [->] (-1.5,-2.75)--(-2.5,-1.75) node[pos=0.5, below left] {$\text{Bl}_{\pi_1\dots\pi_{8}}$};
\end{tikzpicture}
\caption{Sequence of blow ups and surface for system \eqref{second_system} with $\mathbb{F}_1$ compactification, with $(-1)$-curves in red and $(-2)$-curves in blue.  
}
\label{fig:surface:syst2}
\end{figure} 
Note that the surface obtained after blowing up the eight points
    $\pi_i$, $i=1,\ldots,8$ is minimal, in the sense of \Cref{def:minimalSOIC}.
\Cref{cor:secondsystemtosystfg} provides an isomorphism between $X$ constructed from system \eqref{second_system} and $X^{\operatorname{KNY}}$ in \Cref{appendix:standardP6}.
\begin{remark}
     We observe that we could resolve the indeterminacy points of the system
    \eqref{second_system} considering the more common compactification $\BF_0 = \Pj^1\times\Pj^1$. However, in this particular case
    we have to underline that compactifying the system \eqref{second_system} inside $\BF_0$
    present an additional difficulty: \emph{a seemingly infinite cascade} of indeterminacy points.
    This fact underlines again the importance
    of the choice of compactification: albeit the final results will
    be equivalent upon minimisation, the number of steps to reach it can be
    different.
    See   \Cref{app:apparentsingularity} for an extended discussion of this phenomenon 
    and its solution. 
    \label{rem:whyf1}
\end{remark} 
 
\section{Hamiltonian structure of the 3D system on the hypersurface}
\label{sec:hamiltonianhyper}

In this section we construct the Hamiltonian structures of the system \eqref{syst3D} of differential equations restricted to the hypersurface \eqref{eq:St}, as considered in \Cref{sec:firstpar,sec:hypersurf}.  
In   \Cref{Th1} and   \Cref{cor:secondsystemtosystfg} we identified the systems \eqref{syst1} and \eqref{second_system} with a Hamiltonian form of the sixth Painlev\'e equation.
The systems having as spaces of initial conditions the ones constructed in \Cref{subsec:spaceofinitialconditionssys3.2,subsec:soichyper} and system \eqref{systfg} from \cite{KNY}  are related by a birational equivalences of the associated surfaces with $X^{\operatorname{KNY}}_{t}$ after  appropriate matching of parameters, see \Cref{appendix:standardP6}. 
Moreover, the birational transformations relating them allow us to pull back the Hamiltonian atlas in   \Cref{appendix:standardP6} from $E^{\operatorname{KNY}}$ to the respective manifolds.
Then, both systems have a global Hamiltonian structure, which is exactly that of the sixth Painlev\'e equation. We will present this structure in   \Cref{prop:hamsys} which is part of \Cref{thm:main} of the Introduction.

\begin{remark}
    The systems in the coordinates $(\tilde{x},\tilde{y})$ and $(x_0,v)$, in \eqref{syst1} and \eqref{second_system} respectively, do not have a local Hamiltonian structure, because the symplectic form associated with the anti-canonical divisor is $t$-dependent when written in these charts. We expand on this in the next section.
\end{remark}

\subsection{Local Hamiltonian structures}

In this subsection we provide a local Hamiltonian structure for the second system \eqref{second_system} by considering the compactification $\mathbb{F}_1$.
A similar approach also works for the first system \eqref{syst1}, as we will explain at the end of this subsection.

In the chart $(x_0,v)$ the rational two-form $\omega_t$ providing the anti-canonical divisor $D_t$ as its divisor of poles is given by
\begin{equation} \label{omegatx0v}
    \omega_t = \frac{k d_t x_0 \wedge d_tv}{x_0(x_0 - \kappa)( (1-t)x_0 - \kappa) v^2},
\end{equation}
uniquely up to the non-zero constant  $k\in\mathbb{C}^*$. 
Indeed, in these coordinates we cannot find a local Hamiltonian structure for the system \eqref{second_system} with respect to $\omega_t$, because of the $t$-dependence in \eqref{omegatx0v}.
That is, there is no function $H(x_0,v,t)$ rational in $x_0$ and $v$ with coefficients analytic in $t$ such that system \eqref{second_system} is written as 
\begin{equation} \label{hamformeqs:secondsystem}
    \frac{ \kappa x_0 ' }{x_0 (x_0 - \kappa)\left( (1-t)x_0 - \kappa\right) v^2} = \frac{\partial H}{\partial v}, \quad \frac{ \kappa v ' }{x_0 (x_0 - \kappa)\left( (1-t)x_0 - \kappa\right) v^2} = - \frac{\partial H}{\partial x_0}.
\end{equation}
Indeed, equations \eqref{hamformeqs:secondsystem} are incompatible when regarded as a pair of partial differential equations for $H$,  i.e. $\frac{\partial^2 H}{\partial x_0 \partial v} \neq \frac{\partial^2 H}{\partial v \partial x_0}$.
 
Therefore, for the system \eqref{second_system}  we aim to find coordinates in which the symplectic form is $t$-independent.

Recall the atlas for the Hirzebruch surface $\mathbb{F}_1\cong\Bl_p\Pj^2$ as in the previous section, i.e. the four charts $(\xi_i,\upsilon_i)$, for $i=0,1,2,3$, with gluing defined by  \eqref{eq:F1charts}.

Again regarding the variables $x_0$ and $v$ from the system \eqref{second_system} as the coordinates for the chart $(\xi_0,\upsilon_0)$ via $\xi_0=x_0$, $\upsilon_0=v$, the rational two-form \eqref{omegatx0v} is written in these coordinates as 
\begin{equation} \label{omegat}
    \omega_t = \frac{k d_t \xi_0 \wedge d_t \upsilon_0}{\xi_0(\xi_0 - \kappa)( (1-t)\xi_0 - \kappa) \upsilon_0^2},
\end{equation}
where $k$ is possibly $t$-dependent and will be chosen appropriately later. 
We need to remedy the $t$-dependence in the factor $( (1-t)\xi_0 - \kappa)$, and we do this by obtaining a change of our atlas for $\mathbb{F}_1$ such that the problematic $t$-dependent divisor is not visible in one of the new charts, and in this chart the system will have a local Hamiltonian structure.
We represent this procedure to find the coordinate change in \Cref{fig:coordchange}.

First contract the $(-1)$-curve  $C_-$  on $\mathbb{F}_1$, with local equations $\upsilon_0=0$, $\upsilon_1=0$, to get $\p^2$, denoting by $\pi : \mathbb{F}_1 \rightarrow \p^2$ the contraction, see \Cref{rem:sechirz}. 
For $\p^2$ take the following homogeneous coordinates 
\begin{equation}
    [Z_0:Z_1:Z_2] = [\xi_0 \upsilon_0 : \upsilon_0 : 1] = [\upsilon_1 : \xi_1 \upsilon_1 : 1]
\end{equation} 
so that the point $p = \pi(C_-)$ is at the origin in the coordinates 
\begin{equation}
    (r,s)=\left(\frac{Z_0}{Z_2},\frac{Z_1}{Z_2}\right) = \left( \xi_0 \upsilon_0, \upsilon_0 \right) = \left( \upsilon_1, \xi_1 \upsilon_1 \right).
\end{equation}
The two-form \eqref{omegat} pushed forward under $\pi$ is then given by
\begin{equation} \label{omegatP2}
    \pi_*(\omega_t) =  \frac{k d_t r \wedge d_t s}{r(r - \kappa s)( (1-t)r - \kappa s)}. 
\end{equation}
Consider the linear system of lines in $\p^2$ passing through $p$, which is written in coordinates
\begin{equation}
  \mathcal{L}=  \Set{ L_{\lambda} : \lambda_0 Z_0 + \lambda_1 Z_1 = 0 ~|~ \lambda = [\lambda_0 : \lambda_1] \in \p^1 }\cong\Pj^1.  
\end{equation} 
The poles of $\pi_*(\omega_t)$ then correspond to the lines 
\begin{equation}
    L_{[1:0]} =\Set{ Z_0 = 0}, \quad L_{[1:-\kappa]} =\Set{ Z_0 = \kappa Z_1}, \quad L_{[t-1:\kappa]}  =\Set{ (1-t)Z_0 =\kappa Z_1}.
\end{equation}
We look for a change of homogeneous coordinates for $\p^2$, say
\begin{equation}
    [Z_0:Z_1:Z_2] \leftrightarrow [\hat{Z}_0:\hat{Z}_1:\hat{Z}_2],
\end{equation}
corresponding to a M\"obius transformation of $\mathcal{L}$, that fixes the origin in the $(r,s)$ chart as well as the line at infinity $\Set{Z_2=0}$. 
We want this to be such that in the linear system
\begin{equation}
  \Set{ \hat{L}_{\hat{\lambda}} : \hat{\lambda}_0 \hat{Z}_0 + \hat{\lambda}_1 \hat{Z}_1 = 0 ~|~ \hat{\lambda} = [\hat{\lambda}_0 : \hat{\lambda}_1] \in \p^1 }
\end{equation}
the lines on which the two-form has poles are given by
\begin{equation}
      \hat{L}_{[1:0]} =\Set{ \hat{Z}_0 = 0}, \quad  \hat{L}_{[1:-1]}=\Set{ \hat{Z}_0 = \hat{Z}_1}, \quad  \hat{L}_{[0:1]} =\Set{ \hat{Z}_1 = 0}.
 \end{equation}
A direct calculation shows that this is achieved by a change of coordinates 
\begin{equation}
    \begin{aligned}
      Z_0 &= A \hat{Z}_0,\\
      Z_1 &= \tfrac{A}{\kappa} \left( (1-t) \hat{Z}_0 + t \hat{Z}_1 \right),\\
      Z_2 &= B \hat{Z}_2,
   \end{aligned} 
   \qquad 
    \begin{aligned}
      \hat{Z}_0 &= \tfrac{1}{A} Z_0 ,\\
      \hat{Z}_1 &= \tfrac{1}{A t} \left( (t-1) Z_0 + \kappa Z_1\right),\\
      \hat{Z}_2 &= \tfrac{1}{B} Z_2,
   \end{aligned} 
\end{equation}
where $A, B$ are nonzero and possibly $t$-dependent, and will be chosen appropriately below.
Now, taking the coordinates $(\hat{r},\hat{s})= \left(\frac{\hat{Z}_0}{\hat{Z}_2},\frac{\hat{Z}_1}{\hat{Z}_2}\right)$, we blow up the origin $(\hat{r},\hat{s})=(0,0)$ reaching a new copy of $\mathbb{F}_1$. 
The projection $\hat{\pi} : \mathbb{F}_1 \rightarrow \p^2$ is given in coordinates by $\hat{\pi}: (\hat{\xi}_0,\hat{\upsilon}_0) \mapsto (\hat{r},\hat{s}) = (\hat{\xi}_0 \hat{\upsilon}_0,\hat{\upsilon}_0)$, and we have the new coordinates $(\hat{\xi}_i,\hat{\upsilon}_i)$, $i=0,1,2,3$, defined by
\begin{equation} 
    \hat{\xi}_0 = \frac{t \xi_0}{\kappa + (t-1) \xi_0}, \quad \hat{\upsilon}_0 = \frac{B}{A t} \upsilon_0 \left(\kappa + (t-1)\xi_0\right),
\end{equation}
and the remaining $\hat{\xi}_i$, $\hat{\upsilon}_i$ with the same gluing as before, i.e. \eqref{eq:F1charts} with hats.

By design, the problematic $t$-dependent pole of $\omega_t$ is not visible in the chart $(\hat{\xi}_0,\hat{\upsilon}_0)$.
Indeed, we have 
\begin{equation}\label{omegathat}
    \omega_t = \frac{-B}{A \kappa t}\frac{ k d_t \hat{\xi}_0 \wedge d_t \hat{\upsilon}_0}{\hat{\xi}_0 (\hat{\xi}_0 - 1) \hat{\upsilon}_0^2},
\end{equation}
so we choose $A=1$, $B=\kappa t$. We also choose $k=\kappa$ for neater matching with the Hamiltonian system for $(f,g)$ through \Cref{cor:secondsystemtosystfg}, but this is without loss of generality.
Then the system becomes Hamiltonian in these coordinates as follows.
\begin{proposition} \label{prop:localhamsyst2}
    In the coordinates 
    \begin{equation} 
    \hat{\xi}_0 = \frac{t \xi_0}{\kappa + (t-1) \xi_0}, \quad \hat{\upsilon}_0 = \kappa \upsilon_0 \left(\kappa + (t-1)\xi_0\right),
    \end{equation}
    the system \eqref{second_system} has a local Hamiltonian structure with respect to $\omega_t$.
    Explicitly,
    \begin{equation}
        \begin{gathered}
            \frac{\kappa \hat{\xi}_0'}{\hat{\xi}_0 (1-\hat{\xi}_0) \hat{\upsilon}_0^2} = \frac{ \partial H}{\partial \hat{\upsilon}_0}, \quad \frac{\kappa \hat{\upsilon}_0'}{\hat{\xi}_0 (1-\hat{\xi}_0) \hat{\upsilon}_0^2} = - \frac{ \partial H}{\partial \hat{\xi}_0}, \\ 
            H = 
            \frac{\kappa( \beta \hat{\upsilon}_0 - \kappa)}{(t-1)\hat{\xi}_0 \hat{\upsilon}_0^2}
            +
            \frac{(\kappa+n \hat{\upsilon}_0)(\kappa + (n+\gamma)\hat{\upsilon}_0)}{t(t-1)(1-\hat{\xi}_0)\hat{\upsilon}_0^2}
            +\frac{\alpha \kappa}{t \hat{\upsilon}_0}. 
        \end{gathered}
    \end{equation}
\end{proposition}
This Hamiltonian structure coincides with that of the KNY form \eqref{systfg} of $\pain{VI}$ under the identification in \Cref{cor:secondsystemtosystfg}.
The identification induces a birational transformation $\varphi : \C^2 \times B \dashrightarrow \C^2 \times B$, for $B=\C\setminus \Set{0,1}$, defined in coordinates $(\hat{\xi}_0,\hat{\upsilon}_0;t) \mapsto (f,g;t)$ explicitly as
\begin{equation}
\hat{\xi}_0 = \frac{t (f-1)}{(t-1)f}, \quad  
\hat{\upsilon}_0 = \frac{ \kappa(t-1) f}{(f-1)\left( f (\alpha+ n - g) + t(g-\alpha)\right)},
\end{equation}
or conversely
\begin{equation}
    f= \frac{t}{t - (t-1)\hat{\xi}_0},
    \quad 
    g = \frac{t \kappa + \hat{\xi}_0 \left( \kappa(1-t) + (n + \alpha (t-1)(\hat{\xi}_0 -1) )\hat{\upsilon}_0\right)}{(t-1)(\hat{\xi}_0-1)\hat{\xi}_0\hat{\upsilon}_0}. 
\end{equation}
Under $\varphi$ we have the equality of rational two-forms
\begin{equation} \label{twoformequality}
    \varphi^* \left(\omega^{\operatorname{KNY}} + d H^{\operatorname{KNY}}_{\mathrm{VI}} \wedge dt \right) =  \omega + d H \wedge dt,
\end{equation}
where 
\begin{equation}
\omega^{\operatorname{KNY}} = \frac{df \wedge dg}{f},
\quad
\omega = \frac{ \kappa d \hat{\xi}_0 \wedge d \hat{\upsilon}_0}{\hat{\xi}_0 (1-\hat{\xi}_0) \hat{\upsilon}_0^2},
\end{equation}
and $H^{\operatorname{KNY}}_{\mathrm{VI}}$ is given by \eqref{systf H - intro}.

\begin{figure}[htb]
\centering
    \begin{tikzpicture}[basept/.style={circle, draw=black!100, fill=black!100, thick, inner sep=0pt,minimum size=1.2mm},scale=0.9]
    \begin{scope}[xshift=-4cm]
    \node at (0,3.5) {$\mathbb{F}_1$};
    \draw[thick,blue] (-2.5,-2)  -- (+2.5,-2) node[pos=0.5,above,black] {$C_-$} node[pos=0,left] {\small$y_0=0$} node[pos=1,right] {\small$y_1=0$};
    \draw[thick,black] (-2.5,+2) -- (+2.5,+2) node[pos=0,left] {\small$y_3=0$} node[pos=1,right] {\small$y_2=0$};
    \draw[thick,blue] (-2,-2.5) -- (-2,+2.5) node[pos=0,below] {\small$x_0=0$} node[pos=1,above,black] {\small$x_3=0$};
    \draw[thick,black] (+2,-2.5) -- (+2,+2.5)node[pos=0,below] {\small$x_1=0$} node[pos=1,above] {\small$x_2=0$};

    \draw[thick,blue] (-.66,-2.5) -- (-.66,+2.5) node[pos=0,below] {\small$x_0=\kappa$};
    \draw[thick,blue] (+.66,-2.5) -- (+.66,+2.5) node[pos=0,below] {\small$x_0=\tfrac{\kappa}{1-t}$};

    \end{scope}
    
    \draw[thick,black,<->] (-1,0)--(+1,0) node[midway,above] {$\sim$};
    
    \begin{scope}[xshift=+4cm]
    \node at (0,3.5) {$\mathbb{F}_1$};
    \draw[thick,blue] (-2.5,-2)  -- (+2.5,-2)  node[pos=0,left] {\small$ \hat{y}_0=0$} node[pos=1,right] {\small$ \hat{y}_1=0$};
    \draw[thick,black] (-2.5,+2) -- (+2.5,+2) node[pos=0,left] {\small$ \hat{y}_3=0$} node[pos=1,right] {\small$ \hat{y}_2=0$};
    \draw[thick,blue] (-2,-2.5) -- (-2,+2.5) node[pos=0,below] {\small$\hat{x}_0=0$} node[pos=1,above,black] {\small$\hat{x}_3=0$};
    \draw[thick,blue] (+2,-2.5) -- (+2,+2.5)node[pos=0,below] {\small$\hat{x}_1=0$} node[pos=1,above,black] {\small$\hat{x}_2=0$};
    \draw[thick,blue] (0,-2.5) -- (0,+2.5) node[pos=0,below] {\small$\hat{x}_0=1$};

    \end{scope}

    \draw[thick,black,->] (-4,-3.5)--(-4,-4.5) node[midway,left] {$\pi$};

    \begin{scope}[xshift=-4cm, yshift=-8cm]
    \node at (0,-3) {$\mathbb{P}^2$};
    \draw[thick,black] (-2.5,+2) -- (+2.5,+2)  node [pos=0,left] {\small$Z_2=0$} ;
    \draw[thick,blue] (-2.25,+2.5) -- (+.3,-2) node [pos=1,right,black] {\small$Z_0=0$} node[pos=0,above] {\small${L}_{[1:0]}$};
    \draw[thick,black] (+2.25,+2.5) -- (-.3,-2) node [pos=1,left] {\small$Z_1=0$}   node [pos=0,above] {\small${L}_{[0:1]}$};
    
    \draw[thick,blue] (+.09,-2) -- (-.66,+2.5) node [pos=1,above] {\small${L}_{[1:-\kappa]}$} ;
    \draw[thick,blue] (-.09,-2) -- (+.66,+2.5) node [pos=1,above,xshift=+5pt] {\small${L}_{[t-1:-\kappa]}$} ;
    
    \draw[very thick,red,->] (0,-1.5) -- (.35,-.9) node[pos=1,right] {\small$r$};
    \draw[very thick,red,->] (0,-1.5) -- (-.35,-.9) node[pos=1,left] {\small$s$};
    \end{scope}

    \draw[thick,black,->] (+4,-3.5)--(+4,-4.5) node[midway,right] {$\hat{\pi}$};

    \draw[thick,black,<->] (-1,-8)--(+1,-8) node[midway,above] {$\sim$};

    \begin{scope}[xshift=+4cm, yshift=-8cm]
    \node at (0,-3) {$\mathbb{P}^2$};
    \draw[thick,black] (-2.5,+2) -- (+2.5,+2) node [pos=1,right] {\small$\hat{Z}_2=0$} ;
    \draw[thick,blue] (-2.25,+2.5) -- (+.3,-2) node [pos=1,right,black] {\small$\hat{Z}_0=0$} node [pos=0,above] {\small$\hat{L}_{[1:0]}$};
    \draw[thick,blue] (+2.25,+2.5) -- (-.3,-2) node [pos=1,left,black] {\small$\hat{Z}_1=0$} node [pos=0,above] {\small$\hat{L}_{[0:1]}$};

    \draw[thick,blue] (0,-2) -- (0,+2.5) node [pos=1,above] {\small$\hat{L}_{[1:-1]}$} ;

    \draw[very thick,red,->] (0,-1.5) -- (.35,-.9) node[pos=1,right] {\small$\hat{r}$};
    \draw[very thick,red,->] (0,-1.5) -- (-.35,-.9) node[pos=1,left] {\small$\hat{s}$};
    \end{scope}
\end{tikzpicture}
\caption{Coordinate change for $\mathbb{F}_1$ in order to obtain a local Hamiltonian structure for system \eqref{second_system}, with locations of poles of $\omega_t$ in blue. 
}
\label{fig:coordchange}
\end{figure}
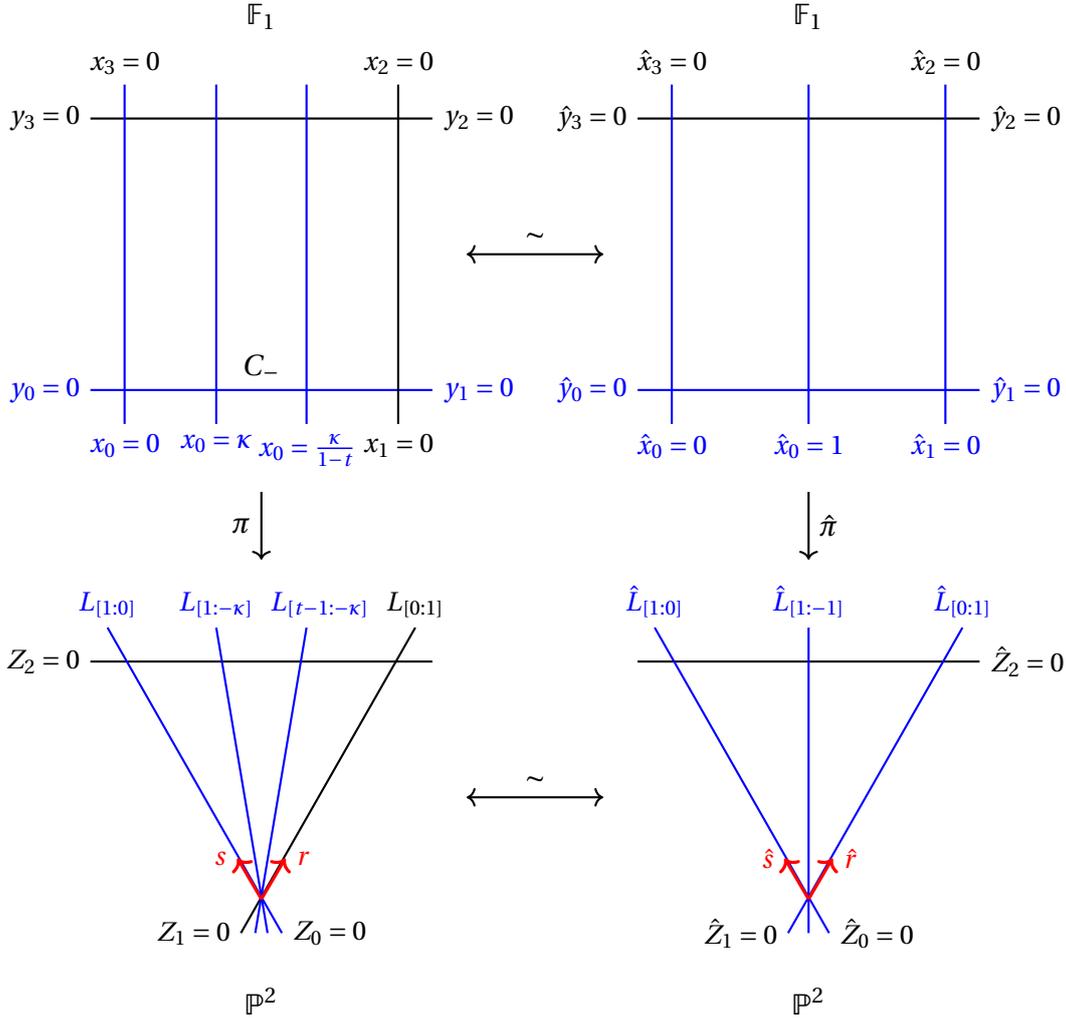

A similar trick can be employed in the case of system \eqref{syst1} in the first parametrisation studied in \Cref{sec:firstpar}, to obtain coordinates, related in a relatively simple way to $\tilde{x},\tilde{y}$, in which the system \eqref{syst1} has a local Hamiltonian structure.

Recall the sequence of blow ups and blow downs performed to construct a minimal space of initial conditions for system \eqref{syst1} from the $\p^1 \times \p^1$ compactification, as shown in \Cref{fig:surface:syst1}. 
In this case the rational two-form with respect to which the Hamiltonian structure should be defined is given in coordinates $\tilde{x},\tilde{y}$ by 
\begin{equation} \label{omegaxytilde}
    \omega_t = \frac{k d_t \tilde{x}\wedge \tilde{y}}{\tilde{x}(\tilde{x}-\kappa)((1-t)\tilde{x}-\kappa)},
\end{equation}
again unique up to the choice of the possibly $t$-dependent constant $k$. 
This is similar to the two-form considered above in the $(x_0,v)$-chart as in \eqref{omegatx0v}, but with the $\p^1 \times \p^1$ compactification this has not only poles along the lines $\{\tilde{x}=0\}$, $\{\tilde{x}=\kappa\}$ and $\{\tilde{x}=\tfrac{\kappa}{1-t}\}$, but also zeroes along $\{\tilde{x}=\infty\}$. 
These zeroes are an artifact of the contractions $\operatorname{Bl}_{p_{12}}\circ \operatorname{Bl}_{p_{11}}$ as in \Cref{fig:surface:syst1} required to arrive at a minimal space of initial conditions on which the two-form has only poles.

The problematic $t$-dependent factor $\left( (1-t)\tilde{x} - \kappa\right)$ in the two-form \eqref{omegaxytilde} can be dealt with by a combination of elementary transformations and the same trick as above.
After the blow up of $\p^1 \times \p^1$ centred at $p_7:(\tilde{x},\tilde{y})=(\infty,\infty)$, if we contract the proper transform of the line $\{\tilde{x}=\infty\}$ we arrive at a copy of $\mathbb{F}_1$, i.e. we have performed an elementary transformation $\p^1 \times\p^1 = \mathbb{F}_0 \dashrightarrow \mathbb{F}_1$ as in \Cref{fig:eltransf}.
Then we are in a similar situation to that considered above, with three lines along which the two-form has poles, one of which is $t$-dependent. We make a coordinate change for $\mathbb{F}_1$ then go back to $\p^1\times\p^1$ via another elementary transformation.
The result is the following counterpart to \Cref{prop:localhamsyst2}.

\begin{proposition} \label{prop:localhamsyst1}
    After making the birational change of variables from $(\tilde{x},\tilde{y})$ to $(\hat{x},\hat{y})$ according to
    \begin{equation}  \label{xytildehat}
    \hat{x} = \frac{t \tilde{x}}{\kappa + (t-1) \tilde{x}}, \quad \hat{y} = \frac{t(t-1)\tilde{y}}{2(\kappa+(t-1)\tilde{x})},
    \end{equation}
    the system \eqref{syst1} has a local Hamiltonian structure with respect to 
    \begin{equation}
        \omega_t = \frac{d_t\hat{x}\wedge d_t \hat{y}}{\hat{x}(\hat{x}-1)}.
    \end{equation}
    Explicitly, the system \eqref{syst1} after the change \eqref{xytildehat} reads
    \begin{equation}
        \frac{\hat{x}'}{\hat{x}(\hat{x}-1)} = \frac{\partial H}{\partial \hat{y}}, \quad  \frac{\hat{y}'}{\hat{x}(\hat{x}-1)} = -\frac{\partial H}{\partial \hat{x}},
    \end{equation}
    where $H=H(\hat{x},\hat{y};t)$ is a Hamiltonian function, rational in its arguments, that can be computed explicitly.
\end{proposition}
Let us give more details of the geometric meaning of the transformation \eqref{xytildehat}, as depicted in \Cref{fig:syst1trick}.
After arriving at the copy of $\mathbb{F}_1$ as explained above, we contract the $(-1)$-curve given in coordinates by $\Set{\tilde{y}=0}\cup\Set{\tfrac{\tilde{x}}{\tilde{y}}=0}$.
We then arrive at $\mathbb{P}^2$ with homogeneous coordinates $[Z_0:Z_1:Z_2]=[1:\tilde{x}:\tilde{y}]$, perform an appropriate change $[Z_0:Z_1:Z_2]\mapsto [\hat{Z}_0:\hat{Z}_1:\hat{Z}_2]$, then blow back up to $\mathbb{F}_1$, similarly to in \Cref{fig:coordchange}. 
This provides a genuine change of atlas for $\mathbb{F}_1$, with $(\hat{x},\hat{y})$ as set of affine coordinates corresponding to $(\xi_3,\upsilon_3)$ in the atlas as shown in the top-left picture from \Cref{fig:surface:syst2}. 
We then perform the same steps in reverse, i.e. blow up the point $(\tfrac{1}{\hat{x}},\tfrac{\tilde{y}}{\tilde{x}})=(0,0)$, then contract the proper transform of the curve $\Set{\tfrac{1}{\hat{x}}=0}$ to a point $\hat{p}_7:(\hat{x},\hat{y})=(\infty,\infty)$.
This leads to $\p^1 \times\p^1$ with $(\hat{x},\hat{y})$ as an affine chart, and we regard the transformation $(\tilde{x},\tilde{y})\mapsto(\hat{x},\hat{y})$ in \eqref{xytildehat} as a birational mapping between two copies of $\p^1 \times\p^1$.
The problematic $t$-dependent line $\Set{\tilde{x}=\tfrac{\kappa}{1-t}}$ is sent under this to $\hat{p}_7$, so it is not visible in the $(\hat{x},\hat{y})$ coordinates. 
This explains the form of the transformation \eqref{xytildehat}, in particular the denominators of the rational functions giving $\hat{x},\hat{y}$.

\begin{remark} The propositions above can also be obtained   analytically. However,  an appropriate Ansatz is needed. 
For the system \eqref{second_system} in the second parametrisation one can  search for a transformation in the form 
$$\hat{\xi}_0=\frac{f_1(t)\xi_0}{\kappa+(t-1)\xi_0},\;\;\hat{v}_0=f_2 (t)v_0(\kappa+(t-1)\xi_0)$$ with functions $f_1$ and $f_2$ to be determined. Notice that  choosing $f_1(t)=t$ causes the factor $( t \hat{\xi}_0-f_1)$  in the denominator of the rational expression for $\hat{v}_0'$ to factorise, and $f_2$ can be taken as a constant. 
Then the  system for $\hat{\xi}_0$ and $\hat{v}_0$ can be written in Hamiltonian form with respect to the  two-form as above. 
For the system \eqref{syst1} in the first parametrisation, taking a similar Ansatz $$\hat{x} = \frac{f_1(t) \tilde{x} }{\kappa+ (t - 1) \tilde{x} }, 
 \hat{y} =\frac{ f_2(t) \tilde{y} }{\kappa + (t - 1) \tilde{x}}$$ and similar considerations yield $f_1(t)=t$. 
 Then $f_2$ is chosen according to the compatibility condition of the pair of partial differential equations for the Hamiltonian function, which gives $(1 - 2 t) f_2 + (t-1) t f_2'=0$.
 \end{remark}

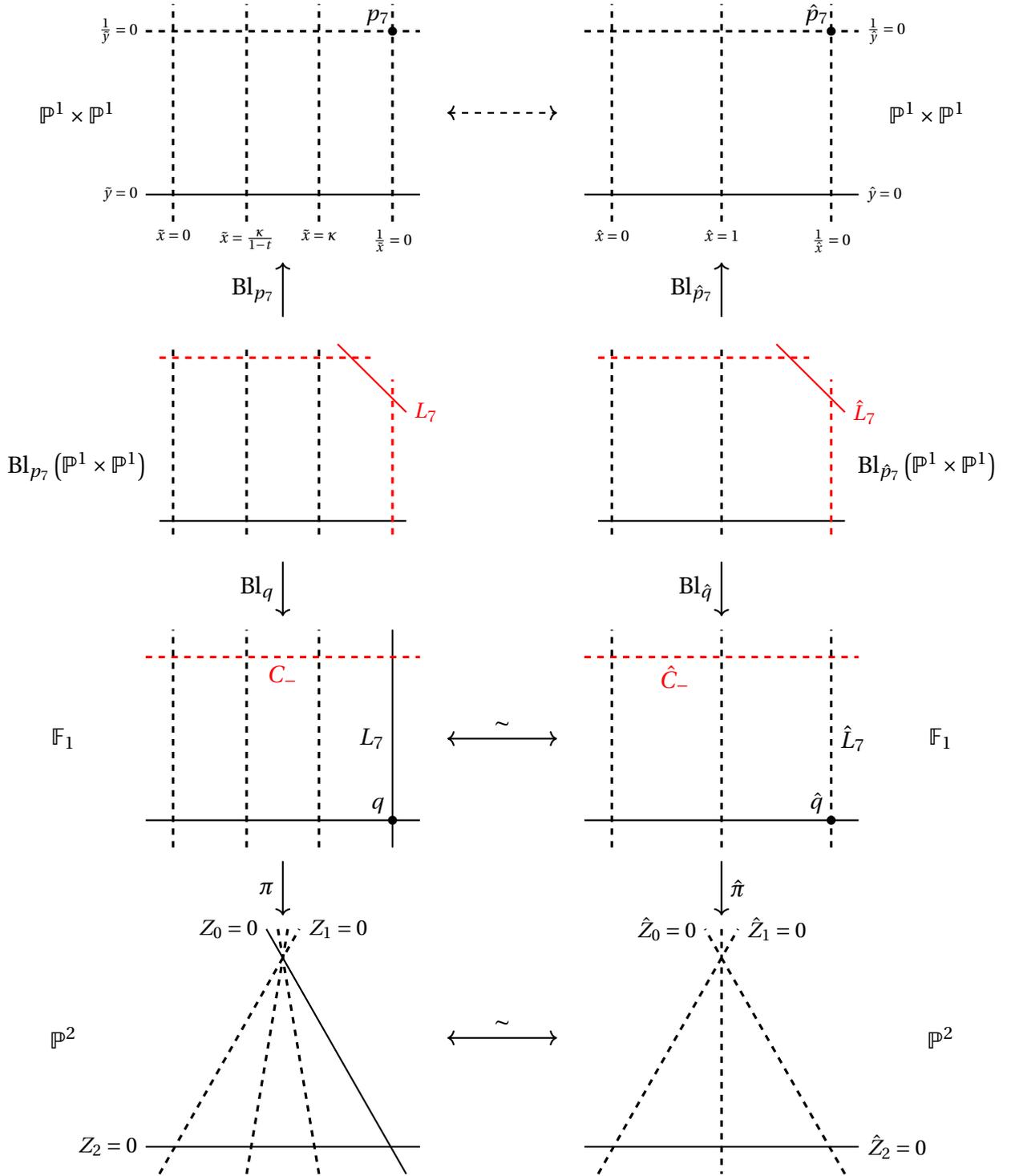
\begin{figure}[!h]
\centering
    \begin{tikzpicture}[scale=.9,basept/.style={circle, draw=black!100, fill=black!100, thick, inner sep=0pt,minimum size=1.2mm}]
    \begin{scope}[xshift = -4cm, yshift=+8cm]
    \draw[thick] 
     (-2.5,-1.5)--(2.5,-1.5) 
      ;     
     \draw[very thick, dashed]
           (-.66,-2)--(-.66,2)  
      (.66,-2)--(.66,2)  
          (-2,-2)--(-2,2)
          (-2.5,1.5)--(2.5,1.5)
          (2,-2)--(2,2) ;
    \node at (-3.75,0) {$\p^1\times\p^1$};
      \node[left] at (-2.5,-1.5) {\tiny $\tilde{y}=0$};
      \node[left] at (-2.5,1.5) {\tiny $\frac{1}{\tilde{y}}=0$};     
      \node[below] at (-2,-2) {\tiny $\tilde{x}=0$};
      \node[below] at (2,-2) {\tiny $\frac{1}{\tilde{x}}=0$};
      \node[below] at (-.66,-2) {\tiny $\tilde{x}=\frac{\kappa}{1-t}$};
      \node[below] at (.66,-2) {\tiny $\tilde{x}=\kappa$};
      
    \node (p7) at (2,1.5) [basept,label={[xshift=-7pt, yshift = -3 pt] \small $p_{7}$}] {};
    \end{scope}

    \draw[thick,black,<->, dashed] (-1,+8)--(+1,+8)
    ;

    \begin{scope}[xshift = +4cm, yshift=+8cm]
    \draw[thick] 
     (-2.5,-1.5)--(2.5,-1.5) 
      ;     
      \draw[very thick, dashed]
    (0,-2)--(0,2)  
          (-2,-2)--(-2,2)
          (-2.5,1.5)--(2.5,1.5)
          (2,-2)--(2,2) ;
     \node at (+3.75,0) {$\p^1\times\p^1$};
     \node[right] at (+2.5,-1.5) {\tiny $\hat{y}=0$};
      \node[right] at (+2.5,1.5) {\tiny $\frac{1}{\hat{y}}=0$};     
      \node[below] at (-2,-2) {\tiny $\hat{x}=0$};
      \node[below] at (2,-2) {\tiny $\frac{1}{\hat{x}}=0$};
      \node[below] at (0,-2) {\tiny $\hat{x}=1$};
    \node (p7) at (2,1.5) [basept,label={[xshift=-7pt, yshift = -3 pt] \small $\hat{p}_{7}$}] {};

    \end{scope}


\draw [thick,->] (+4,4.25)--(+4,5.25) node[pos=0.5, left] {$\text{Bl}_{\hat{p}_7}$};

\draw [thick,->] (-4,4.25)--(-4,5.25) node[pos=0.5, left] {$\text{Bl}_{p_7}$};
  \begin{scope}[xshift = -4cm, yshift=+2cm]
    \node at (-3.75,-.5) {$\operatorname{Bl}_{p_7}\left(\p^1\times\p^1\right)$};
     \draw[thick,black]    
    (-2.25,-1.5)--(2.25,-1.5);
     \draw[very thick, dashed] 
    (-2,-1.75)--(-2,1.75);
     \draw[very thick, dashed, red] 
    (-2.25,1.5)--(1.6,1.5);
    \draw[very thick, dashed, red] 
    (2,-1.75)--(2,1.1) ;
    \draw[very thick, dashed] 
      (-.66,-1.75)--(-.66,1.75);  
    \draw[very thick, dashed] 
      (.66,-1.75)--(.66,1.75)  ;     
    \draw[red, thick] (1.0,1.75) --(2.25,0.5); 
    \node[right,red] at (2.25,0.5) {\small $L_7$};      
    \end{scope}

  \begin{scope}[xshift = +4cm, yshift=+2cm]
    \node at (+3.75,-.5) {$\operatorname{Bl}_{\hat{p}_7}\left(\p^1\times\p^1\right)$}; 
     \draw[thick,black]    
    (-2.25,-1.5)--(2.25,-1.5);  
    \draw[very thick,dashed] 
    (-2,-1.75)--(-2,1.75);
    \draw[very thick,dashed, red] 
    (2,-1.75)--(2,1.1) ;
    \draw[very thick,dashed] 
      (0,-1.75)--(0,1.75);  
    \draw[very thick,dashed, red] 
    (-2.25,1.5)--(1.6,1.5);
    \draw[red, thick] (1.0,1.75) --(2.25,0.5); 
    \node[right,red] at (2.25,0.5) {\small $\hat{L}_7$};      
    \end{scope}
\draw [thick,->] (+4,-.25)--(+4,-1.25) node[pos=0.5, left] {$\text{Bl}_{\hat{q}}$};

\draw [thick,->] (-4,-.25)--(-4,-1.25) node[pos=0.5, left] {$\text{Bl}_{q}$};

   \begin{scope}[xshift=-4cm,yshift=-3cm]
    \node at (-4,-.5) {$\mathbb{F}_1$};
    \draw[thick,black] (-2.5,-2)  -- (+2.5,-2)  
    ;  
    \draw[thick,black] (+2,-2.5) -- (+2,+1.5)
    node [pos=0.5, left,black] {$L_7$}
    ; 
    \node (q7) at (2,-2) [basept,label={[xshift=-7pt, yshift = -3 pt] \small $q$}] {};
        \draw[very thick, dashed, red] (-2.5,+1) -- (+2.5,+1)
    node[pos=0.5,below,red] {$C_-$} ;
    \draw[very thick, dashed]
    (-.66,-2.5) -- (-.66,+1.5);
    \draw[very thick, dashed] (+.66,-2.5) -- (+.66,+1.5) ;
    \draw[very thick, dashed] (-2,-2.5) -- (-2,+1.5) ;

    \end{scope}
   
    \draw[thick,black,<->] (-1,-3.5)--(+1,-3.5) node[midway,above] {$\sim$};
   
    \begin{scope}[xshift=+4cm,yshift=-3cm]
    \node at (4,-.5) {$\mathbb{F}_1$};
    \draw[thick,black] (-2.5,-2)  -- (+2.5,-2) 
    ;  
    \node (q7) at (2,-2) [basept,label={[xshift=-7pt, yshift = -3 pt] \small $\hat{q}$}] {};
    \draw[very thick, dashed, red] (-2.5,+1) -- (+2.5,+1) 
    node[pos=0.33,below,red] {$\hat{C}_-$};
    \draw[very thick, dashed] (+2,-2.5) -- (+2,+1.5)
    node [pos=0.5, right,black] {$\hat{L}_7$};
    \draw[very thick, dashed] (0,-2.5) -- (0,+1.5) ;
    \draw[very thick, dashed] (-2,-2.5) -- (-2,+1.5);

    \end{scope}

    \draw[thick,black,->] (-4,-5.75)--(-4,-6.75) node[midway,left] {$\pi$};


    \begin{scope}[xshift=-4cm, yshift=-9cm]
    \node at (-4,0) {$\mathbb{P}^2$};
    \draw[thick,black] (-2.5,-2) -- (+2.5,-2)  node [pos=0,left,black] {\small$Z_2=0$} ;
    \draw[very thick,dashed] (-2.25,-2.5) -- (+.3,+2) node [pos=1,right,black] {\small$Z_1=0$} 
    ;
    \draw[thick,black] (+2.25,-2.5) -- (-.3,+2) node [pos=1,left] {\small$Z_0=0$}   
    ;
    
    \draw[very thick,dashed] (+.09,+2) -- (-.66,-2.5) 
    ;
    \draw[very thick,dashed] (-.09,+2) -- (+.66,-2.5) 
    ; 
    \end{scope}
    \draw[thick,black,->] (+4,-5.75)--(+4,-6.75) node[midway,right] {$\hat{\pi}$};

    \draw[thick,black,<->] (-1,-9)--(+1,-9) node[midway,above] {$\sim$};
    \begin{scope}[xshift=+4cm, yshift=-9cm]
    \node at (4,0) {$\mathbb{P}^2$};
    \draw[thick,black] (-2.5,-2) -- (+2.5,-2) 
    node [pos=1,right] {\small$\hat{Z}_2=0$} 
    ;
    \draw[very thick,dashed] (-2.25,-2.5) -- (+.3,+2) 
    node [pos=1,right,black] {\small$\hat{Z}_1=0$} 
    ;
    \draw[very thick,dashed] (+2.25,-2.5) -- (-.3,+2) 
    node [pos=1,left,black] {\small$\hat{Z}_0=0$} 
    ;

    \draw[very thick,dashed] (0,+2) -- (0,-2.5) 
    ; 
    \end{scope}
    
\end{tikzpicture}
\caption{Construction of the birational map $\p^1 \times \p^1 \dashrightarrow \p^1 \times \p^1$, $(\tilde{x},\tilde{y})\mapsto (\hat{x},\hat{y})$ in \Cref{prop:localhamsyst1}. $(-1)$-curves are coloured in red, and components of $\operatorname{div}(\omega_t)$ are indicated by dashed lines. 
}
\label{fig:syst1trick}
\end{figure} 
\subsection{Global Hamiltonian structure} \label{subsec:globalhamstructures}

As remarked above, the identification of the system on the hypersurface in either parametrisation with the Okamoto Hamiltonian form of $\pain{VI}$ means that it inherits a global Hamiltonian structure on the space of initial conditions.
Using the identifications in   \Cref{Th1} and \Cref{cor:secondsystemtosystfg}, the global Hamiltonian structure of $\pain{VI}$ in Appendix \eqref{appsubsec:globalhamstructure} is inherited by systems \eqref{syst1} and \eqref{second_system}. This is the content of \Cref{prop:hamsys} which is part of \Cref{thm:main} in the Introduction.

\begin{proposition} 
    The coordinates $(x_i,y_i)$, $i=0,1,\dots,5$,  of the symplectic atlas in   \Cref{prop:takanosymplecticatlas} for $E^{\operatorname{KNY}}$ from the KNY Hamiltonian form of $\pain{VI}$ can be pulled back under the identification in   \Cref{cor:secondsystemtosystfg} to provide a symplectic atlas for the space $E$ constructed in \Cref{sec:hypersurf} from the system in the second parametrisation \eqref{second_system}. 
    
    Similarly they can be pulled back under the identification in \Cref{Th1} to provide a symplectic atlas for the space constructed in \Cref{sec:firstpar} from the system in the first parametrisation \eqref{syst1}.

 Further, the system on the hypersurface in either parametrisation has a global Hamiltonian structure provided by the pullback of $\Omega^{\operatorname{KNY}}$ in \Cref{cor:HamiltonianstructureKNY} under the corresponding identification.

    \label{prop:hamsys}
\end{proposition}
\begin{proof}
    The proof consists of a direct check comparing the coordinates $(x_i,y_i;t)$ with those introduced according to the convention \eqref{eq:coordblowp}. This can be done using a computer algebra system, e.g.\ Maple~\cite{Maple} or 
    Mathematica~\cite{Mathematica}. 
      Concretely, note that the coordinates $(x_0,y_0)$ from the symplectic atlas for $E^{\operatorname{KNY}}$ in \Cref{prop:takanosymplecticatlas} are related to the variables $(\tilde{x},\tilde{y})$ from system \eqref{syst1} via the identification in   \Cref{Th1} by 
    \begin{equation} \label{symplecticatlasxytildetoxy0}
        \begin{aligned}
            x_0 &= 1 + \frac{(t-1)\tilde{x}}{\kappa}, \\ y_0&= \frac{\kappa \left[(t-1)(\alpha-\beta-\gamma-1)\tilde{x}^2 + \kappa(\beta \kappa + t(t-1)\tilde{y})- \kappa \left( 1 - \alpha+2\beta + \gamma + t(\alpha - \beta - 1) \right)\tilde{x}   \right]}{2(t-1)\tilde{x}\left(\tilde{x}-\kappa\right)\left(\kappa - (t-1)\tilde{x}\right)}, 
        \end{aligned}
    \end{equation}
    with the remaining coordinates $(x_i,y_i)$, $i=1,\dots,5$, related to $\tilde{x},\tilde{y}$ via the gluing in   \Cref{prop:takanosymplecticatlas} together with \eqref{symplecticatlasxytildetoxy0}. 
    It can be directly checked that these provide an atlas for $E$ as constructed in \Cref{sec:firstpar}.
    
Similarly the charts from the symplectic atlas are related to the variables 
$(\xi_0,\upsilon_0)$ in system \eqref{second_system} via \Cref{cor:secondsystemtosystfg} by 
\begin{equation} \label{symplecticatlasxiupsilontoxy0}
    \begin{aligned}
    x_0 &= 1 + \frac{(t-1)\xi_0}{\kappa}, \\
    y_0 &= \frac{\kappa \left[\kappa + \kappa (n+\alpha - t \alpha)\xi_0 \upsilon_0 + (t-1)(\alpha+n)\xi_0^2 \upsilon_0 \right]}{(t-1)\xi_0 (\xi_0 - \kappa)(\kappa + (t-1)\xi_0)\upsilon_0},
    \end{aligned}
\end{equation}
with the remaining coordinates $(x_i,y_i)$, $i=1,\dots,5$, related to $\xi_0,\upsilon_0$ via the gluing in   \Cref{prop:takanosymplecticatlas} together with \eqref{symplecticatlasxiupsilontoxy0}.
           It can be directly checked that these provide an atlas for $E$ as constructed in \Cref{sec:hypersurf}.

\end{proof}

\section{An autonomous limit of the system of three first-order differential equations}

\label{sec:autlim}

In this section we examine in details an autonomous limit of the three-dimensional system \eqref{syst3D} which reduces to   autonomous, i.e. time-independent, system of three ODEs and we prove part of \Cref{thm:main} in \Cref{prop:sys3d}. In particular, we show the existence of  an integrable autonomous limit whose orbits are elliptic curves. This suggests that the full three-dimensional system might possess the Painlev\'e property.

We recall that in the case of second-order Painlev\'e differential equations similar autonomous limits are integrable in terms of elliptic functions. The same is true in the case of discrete Painlev\'e equations whose autonomous limits, including QRT maps \cite{QRT1,QRT2}, have spaces of initial conditions consisting of rational elliptic surfaces \cite{tsuda,duistermaatbook,carsteatakenawaelliptic}. 

\begin{remark}
    We remark that also the opposite procedure is possible, i.e.\ the  
    deautonomisation of differential and difference equations. We recall 
    the Painlev\'e 
    $\alpha$-test~\cite{Weissetal1983,GrammaticosRamaniASIDE,painlevehandbook}
    for the continuous case, and the preservation of singularity patterns~\cite{Grammaticos1991,GrammaticosRamaniASIDE} in the discrete setting.
    We just mention  that those methods usually become unpractical for
    systems in dimension higher than two, due to the complexity of the calculations involved.
\end{remark}

We start this section by reviewing some key tools to build the associated Poisson tensor
and prove Liouville integrability for the autonomous system.

\subsection{Construction of a Poisson bracket from first integrals}

Recall that a system of (autonomous) first-order ODEs~\eqref{eq:firstord} is said 
to be \emph{volume preserving} if there exists a volume form $\Omega$, 
such that the Lie derivative of  $\Omega$ along the vector field defined by the
system vanishes. More explicitly, let $\vec{F}_t=\vec{F}=(F_1(\vec{x}),\ldots,F_{n}(\vec{x}))$ be the right-hand side of~\eqref{eq:firstord}. Then the associated
vector field is:
\begin{equation}
    \Gamma = \sum_{i=1}^n F_i(\vec{x})\partial_{x_i},
\end{equation}
and a volume form has the coordinate expression $\Omega = w(\vec{x}) d x_1 \wedge \ldots \wedge dx_n$. So, denoting by $\mathcal{L}_{\Gamma}$ the Lie derivative, the condition of being volume preserving is 
$\mathcal{L}_{\Gamma} \Omega = 0$. In coordinates, one has the expression:
\begin{equation}
    \mathcal{L}_{\Gamma} \Omega = 
    \sum_{i=1}^{n}\frac{\partial}{\partial x_i}(w(\vec{x})F_i(\vec{x}))= 0.
    \label{eq:LGammaOmega}
\end{equation}

\begin{remark}
    Throughout this section, when no confusion is possible, we omit the subscript $t$ on a function not depending explicitly  on time.
\end{remark}

In what follows, we will make extensive use of the following results regarding autonomous systems.

\begin{theorem}[{\cite[Proposition 3]{Byrnesetal1999}} ] 
    Consider a system of autonomous first-order ODEs. Assume that:
    \begin{itemize}
        \item the system preserves the volume form $\Omega$;
        \item there exist $I_1$, \ldots, $I_{n-2}$ functionally independent first integrals.
    \end{itemize}
    Then, choosing the $n$-multivector $\tau$ such that $\tau\lrcorner \Omega=1$,
    the 2-tensor:
    \begin{equation}
        J =
        \tau \lrcorner d I_1 \lrcorner \ldots \lrcorner d I_{n-2}, 
    \end{equation}
     is a Poisson tensor for the system of autonomous first-order ODEs. 
    \label{thm:byrnes}
\end{theorem}

In particular, in the case of integrable three-dimensional systems we have the following statement.

\begin{corollary}[{\cite[Corollary 16]{Byrnesetal1999}} ]
    Consider a three-dimensional system of autonomous first-order ODEs. Assume that the system is 
    integrable, in the sense that it admits two functionally independent first 
    integrals $I_1$, $I_2$, and  a preserved volume $\Omega$. Then, the system is multi-Poissonian with Poisson structures $J_1 = \tau \lrcorner d I_1$ and
    $J_2 = \tau \lrcorner d I_2$ where the $n$-multivector $\tau$ satisfies  $\tau\lrcorner \Omega=1$.
    \label{cor:byrnes}
\end{corollary}

The most delicate point in~\Cref{thm:byrnes} is to find a preserved volume,
a problem for which, in general, there is no finite algorithm. However,
there are cases where an algorithmic construction is possible, as highlighted
in the following result.

\begin{lemma}
    Consider a system of autonomous first-order ODEs. 
    Let $P$ be a Darboux function whose cofactor satisfies:
    \begin{equation}
        C = \sum_{i=1}^{n} \frac{\partial F_i}{\partial x_i}.
        \label{eq:Com}
    \end{equation}
    Then, the volume form:
    \begin{equation}
        \Omega_P = \frac{1}{P(\vec{x})} d x_1 \wedge \ldots \wedge dx_n
    \end{equation}
    is preserved by the ODE system.
    \label{lem:volume}
\end{lemma}

\begin{proof}
    By the condition of volume invariance written as in
    equation \eqref{eq:LGammaOmega} we have:
    \begin{equation}
        \mathcal{L}_{\Gamma} \Omega_P = 
        \sum_{i=1}^{n}\frac{\partial}{\partial x_i} 
        \left(\frac{F_i(\vec{x})}{P(\vec{x})}\right)
        =
        \frac{1}{P^2(\vec{x})}\sum_{i=1}^{n}
        \left(
            P(\vec{x})\frac{\partial F_i}{\partial x_i} -
            F_i(\vec{x})\frac{\partial P}{\partial x_i}
        \right).
    \end{equation}
    Using the defining property of a Darboux function whose cofactor is given by
    \eqref{eq:Com} we obtain the vanishing of the last sum. This ends the proof.
\end{proof}

\subsection{Autonomous limit of $\pain{VI}$}

As noted above, constructing  autonomous limits is in general a tricky task, 
as there is no algorithm for it. So, before proceeding to discuss   autonomous limits
of system \eqref{syst3D}, we review some known autonomous limits of $\pain{VI}$. 

One such limit is given for instance in {\cite[Eqs. (14), (15)]{alves}}. In our notation, let the function $f(t)$ solve equation 
\eqref{P6} with parameters $A,\,B,\,C,\,D$. Let us denote, with abuse of notation,  by $f(s) $ the composition $f(a+\varepsilon s)$, for $a\not=0,1$. Then, fixing parameters
\begin{equation}
  A=\frac{(a-1)^2 a^2 b}{\varepsilon^2},\,\,B=\frac{ac(a-1)^2}{\varepsilon^2},\,\,C=\frac{a^2 d (a-1)}{\varepsilon^2},\,\,D=\frac{a e (a-1)}{\varepsilon^2}  
\end{equation}
and taking the limit $\varepsilon \to 0$, the function $f(s)$ solves the so-called  $I_{49}$ equation,
see \cite{InceBook}, given by
\begin{equation}\label{I49}
f''=\frac{1}{2}\left(\frac{1}{f-a}+\frac{1}{f}+\frac{1}{f-1}\right)(f')^2+f(f-1)(f-a)\left(b+\frac{c}{f^2}+\frac{d}{(f-1)^2}+\frac{e}{(f-a)^2}\right).
\end{equation}
\Cref{I49} can be solved in terms of elliptic functions.
This equation is variational with the following Lagrangian:
\begin{equation}
    L_{49} = \frac{(f')^2}{2f(f-1)(f-a)}
    -bf+\frac{e}{f-a}+\frac{d}{f-1}+\frac{c}{f}+\frac{d \Delta}{d t},
    \label{eq:L49}
\end{equation}
where $\Delta=\Delta(t,f)$ is an arbitrary function (gauge function).
It is worth mentioning that in \cite{DambrNuc2009} this Lagrangian is derived through the
Jacobi Last Multiplier.

We can argue similarly with the two-dimensional differential system \eqref{systfg}.
Indeed, taking the same change of the independent variable $t=a+\varepsilon s $, scaling
the dependent variables as $(f(t),g(t))\to (f(s), g(s)/\varepsilon)$, and rescaling the parameters as 
\begin{equation}
    \label{par scaling}
    \alpha\to\frac{\tilde{\alpha}}{\varepsilon},\;\;\beta\to \frac{\tilde{\beta}}{\varepsilon},\;\;\gamma\to\frac{\tilde{\gamma}}{\varepsilon},\;\;n\to\frac{\tilde{n}}{\varepsilon}
\end{equation}
in the limit $\varepsilon\to 0$ yields the following autonomous analogue of 
the two-dimensional differential system \eqref{systfg}:
\begin{equation}
    \label{systfg:aut}
    \begin{aligned}
f' &=\frac{f(\tilde{\alpha}(a+1)-a\tilde{\beta}-\tilde{\gamma}+(\tilde{\gamma}+\tilde{\beta}-\tilde{\alpha})f)+2(f-1)(f-a)g-\tilde{\alpha}a}{a(a-1)}
,\\
   g'&=-\frac{a g(g-\tilde{\alpha})+f^2(\tilde{n}+\tilde{\alpha}-g)(\tilde{n}+\tilde{\beta}+\tilde{\gamma}+g)}{a(a-1)f},
\end{aligned}
\end{equation} 
where the differentiation is with respect to $s$. 
System \eqref{systfg:aut} gives equation \eqref{I49} with parameters:
\begin{equation}
b=\frac{\tilde{\kappa }^2}{2 a^2(a-1)^2 },\;c=-\frac{\tilde{\alpha }^2}{2 a(a-1)^2  },\;d=\frac{\tilde{\beta }^2}{2 a^2
   (a-1)},\;e=-\frac{\tilde{\gamma }^2}{2a (a-1) },    
\end{equation}
where 
\begin{equation}
\tilde{\kappa }=\tilde{\alpha }+\tilde{\beta }+\tilde{\gamma }+2 \tilde{n}.  
\end{equation} 
In fact the system \eqref{systfg:aut} is still Hamiltonian with the same symplectic
form as the original symplectic structure, see equation \eqref{systfgham}:
\begin{equation}
    H_{49} = \frac{(g-\tilde{n}-\tilde{\alpha})(\tilde{n}+\tilde{\beta}+\tilde{\gamma}+g)f + ((\tilde{\alpha}-\tilde{\beta})a+\tilde{\alpha}-\tilde{\gamma}-(a+1)g)gf}{a(a-1)}
    +\frac{g (g-\tilde{\alpha})}{(a-1)f},
    \label{eq:hamauto2d}
\end{equation}
see for instance  \eqref{systf H - intro}.
We observe that this Hamiltonian formalism does not  arise as the one 
obtained from the variational formalism~\eqref{eq:L49} (up to the additional
gauge function).

\begin{remark}
    The autonomous limit proposed in \cite{alves} is not the only possible
    one for the $\pain{VI}$ equation. For instance, in \cite{JoshiP6} the exponential-like limit $t\to 0$ was considered  by replacing $t\to e^{\tau}$ 
    and taking   $e^{\tau}\to 0$. A similar reasoning can be done near $t=1$.
   It is possible to take similar limits of the 3D system \eqref{syst3D}, but we will not give a full 
    treatment of these two cases, since  they are \emph{singular}, i.e. the system
    collapses to a two-dimensional one. For the sake of completeness, a short description of the result 
    obtained in such singular cases, is reported in \Cref{app:clim}.
\end{remark}

\subsection{Integrable autonomous limit of system  \eqref{syst3D}}
Following the construction in \cite{alves} for $\pain{VI}$  with the scaling on the
dependent variables:
\begin{equation}
    \label{xyz scaling}
    \vec{x}(t)\to \frac{1}{\varepsilon} \vec{x}(s), \quad
    t = a + \varepsilon s,
\end{equation}
and the scaling \eqref{par scaling} of the parameters,  the system \eqref{syst3D}
becomes:
\begin{equation}
    \begin{aligned}
        x'&=\frac{ \left(\tilde{\kappa }+  \tilde{\beta } +2
   a y -a (\tilde{\alpha }+ \tilde{\beta })-2 (a-1) z\right)x-\tilde{\kappa } \left(\tilde{\beta }+2 z\right)+(a-1) x^2}{(a-1) a},
        \\
         y' &=
   \frac{\left[\begin{gathered}
   a\left((a-1)
   \tilde{n} \left(\tilde{\gamma }+\tilde{n}\right)+ \left((a-1) \tilde{\gamma }+2 (a-1) \tilde{n}-\tilde{\alpha
   }\right)y+a y^2\right) x^2  
   \\ 
    - \left((a-1) x^2 \left(a \left(\tilde{\gamma }+2 \tilde{n}\right)+\tilde{\beta }+2 a
   y\right)+\tilde{\beta } \tilde{\kappa }^2\right)z- \left(\tilde{\kappa }^2-(a-1)^2 x^2\right)z^2\    
   \end{gathered}\right]}{a^2 \tilde{\kappa
   }x},
   \\
    z'&=\frac{a y'}{a-1}.   
    \end{aligned}
    \label{syst3Daut}
\end{equation}
This system is clearly autonomous.
The associated hypersurface $S_t$ degenerates to the following time-independent 
surface:
\begin{equation}
    S_a = \Set{h_a(x,y,z)=0},
    \label{Sa}
\end{equation}
where
\begin{equation}
    \begin{aligned}
h_a &= a x y \left( \left(a \tilde{\gamma }+2 a \tilde{n}-\tilde{\alpha }-\tilde{\gamma }-2 \tilde{n}+a y-2 (a-1)
   z\right)x+\tilde{\kappa } \left(\tilde{\alpha }+\tilde{\gamma }+2 \tilde{n}-2 z\right)\right) 
   \\ 
   &+\left(\tilde{\kappa
   }+(a-1) x\right) \left(x \left(a \tilde{n} \left(\tilde{\gamma }+\tilde{n}\right)- \left(a \tilde{\gamma }+2 a
   \tilde{n}+\tilde{\beta }\right)z+(a-1) z^2\right)+ \tilde{\kappa } \left(\tilde{\beta }+z\right)z\right)=0.        
    \end{aligned}
    \label{ha}
\end{equation}
Moreover, its defining polynomial  $h_a$ is still a Darboux polynomial, and it has cofactor:
\begin{equation}
    C_a =
    \frac{ \left(\tilde{\kappa }+ \tilde{\beta }+2
   a y -a (\tilde{\alpha }+ \tilde{\beta })-2 (a-1) z\right)x-\tilde{\kappa } \left(\tilde{\beta }+2 z\right)+(a-1) x^2}{(a-1) a x}.
   \label{eq:Ca}
\end{equation} 
\begin{remark}
    We remark that in the singular limit $a\to0$, the Darboux surface $S_a$ becomes
    reducible:
    \begin{equation}\label{S0 limit}
        S_0 
        =
        \Set{z \left(\tilde{\beta }+z\right) \left(\tilde{\kappa }-x\right)^2=0}.
    \end{equation}
    This already highlights why one has to carefully deal with the
    value of the parameters in the autonomous limits. A similar occurrence
    will be displayed in \Cref{app:clim} while analysing other possible
    autonomous limits.
\end{remark}

The existence of the Darboux polynomial is crucial in  the proof of the  following proposition which constitutes part of \Cref{thm:main} in the Introduction.

\begin{proposition}
    The system~\eqref{syst3Daut} has the following properties:
    \begin{enumerate}
        \item it admits two functionally independent first integrals given by:
            \begin{equation}
                \mathcal{I}_1 = a y -(a-1)z,
                \quad
                \mathcal{I}_2 = \frac{h_a(x,y,z)}{x};
            \end{equation}
            \label{list:inv}
        \item it preserves the following volume forms:
            \begin{equation}
                \Omega_{a}^{(1)} =
                \frac{dx \wedge dy \wedge dz}{x},
                \quad
                \Omega_{a}^{(2)} =
                \frac{dx \wedge dy \wedge dz}{h_a(x,y,z)};
                \label{eq:sys3dvolume}
            \end{equation}
            \label{list:vol}
        \item the level curves of the invariants form an elliptic fibration;
            \label{list:ell}
        \item the system admits two pairs of Poisson structures, which 
            we denote by $J_{a,j}^{(i)}$, for $i,j= 1,2$.
            \label{list:poisson}
    \end{enumerate}
    \label{prop:sys3d}
\end{proposition}

\begin{proof}
    The two invariants in point \eqref{list:inv} can be obtained as follows.
    
    For $\mathcal{I}_1$ integrate
    the third equation in \eqref{syst3Daut} with respect to $s$.  While, for $\mathcal{I}_2$
    note that the ratio of $h_a$ and $x$ is a first integral as both are Darboux polynomials with the same cofactor $C_a$, see equation \eqref{eq:Ca}. 

    Point \eqref{list:vol} follows from \Cref{lem:volume}. Indeed, the volume forms 
    $\Omega_{a}^{(i)}$, for $i=1,2$,  are obtained by noting that the cofactor $C_a$ satisfies  
    the condition of \Cref{lem:volume}. As a consequence, again  by \Cref{lem:volume}, the
    system preserves both $\Omega_{a}^{(1)}$ and $\Omega_{a}^{(2)}$.

    Let us now prove point \eqref{list:ell}. First, fix an admissible initial condition 
    $(x_0,y_0,z_0)$. To it is assigned a value of the two first
    integrals $\mathcal{I}_1=\iota_1$ and $\mathcal{I}_2=\iota_2$ where 
    $\iota_i=\iota_i(x_0,y_0,z_0)$. Since $\mathcal{I}_1$ is linear, we can solve it with 
    respect to $z$ and plug the solution into $\mathcal{I}_2$. After clearing the denominators
    we see that  the variables $(x,y)$ satisfy the following
    relation:
    \begin{equation}\label{eq:SIGMAcurve}
        \begin{aligned}[t]
        & \left(\tilde{n} (\tilde{n}+\tilde{\gamma})a^2-a\tilde{\kappa}y+\left(( \tilde{\kappa}+\iota_1 -\tilde{\alpha}\right)\iota_1
        -\tilde{n}(\tilde{n}+\tilde{\gamma})) a  \right)(a-1)^2 x^2
        \\
        &+2
        \begin{aligned}[t]
        &\left[\left((\tilde{\beta}+\tilde{\alpha}) a-2\iota_1- \tilde{\kappa}-\tilde{\beta} \right)  \frac{a\tilde{\kappa}}{2} y
        +( \tilde{n}+\tilde{\gamma})  \frac{\tilde{n}\tilde{\kappa}}{2} a^2
        \right.
        \\
        &+\left(
        (  \tilde{\kappa}- \tilde{\beta}) 
        \left( \tilde{n}+\frac{\tilde{\gamma}}{2}\right) \iota_1
        \right.
        -\frac{\iota_1 (\tilde{\beta}+\tilde{\alpha})+\tilde{n} ( \tilde{n}+\tilde{\gamma})}{2} \tilde{\beta}
        - \tilde{n}^3        
        \\
        &\left. \left.-\frac{(\tilde{\alpha}+3 \tilde{\gamma})\tilde{n}^2+ \tilde{\gamma} (\tilde{\alpha}+\tilde{\gamma})  \tilde{n} + \iota_2}{2}  
        \right) a 
        +\iota_1    \tilde{\kappa}(   \tilde{\beta}
        +     \iota_1)+\frac{\iota_2}{2}\right] (a-1) x 
        \end{aligned}
        \\
        &+(ay-\iota_1) (a \tilde{\beta}+a y-\tilde{\beta}-\iota_1) 
        \tilde{\kappa}^2=0,
        \end{aligned}
    \end{equation} 
depending on $\iota_i$, for $i=1,2$.    Using the command \texttt{genus} from the Maple package \texttt{algcurves}~\cite{Maple} we check that the curve \eqref{eq:SIGMAcurve} has genus one, i.e. it is an elliptic curve. As a consequence, generically the level curves of $\mathcal{I}_2$ are elliptic curves possibly degenerating to singular ones. Indeed, the orbits of the system~\eqref{syst3Daut} are exactly given as the intersections of the planes $\Set{a y -(a-1)z=\iota_1}$ with the hypersurface $S_a$.  The fibration is given by varying the values of $\iota_1,\iota_2$.

    Finally, the proof of point \eqref{list:poisson} follows from \Cref{thm:byrnes}
    and \Cref{cor:byrnes} applied to the two volume forms $\Omega_{a}^{(i)}$ and the 
    first integrals $\mathcal{I}_i$, for $i=1,2$. For instance, the Poisson structure
    associated to $\Omega_{a}^{(1)}$ and $\mathcal{I}_1$, has the following matrix form:
    \begin{equation}
        J_{a,1}^{(1)} =
        -\frac{x}{6}
        \begin{pmatrix}
            0 & a-1 & a
            \\
            1-a & 0 & 0
            \\
            -a & 0 & 0
        \end{pmatrix}.
    \end{equation}
    The other Poisson structures can be derived analogously and their explicit
    expressions are rather cumbersome. For this reason, and as their computation 
    consists of a standard computer routine, we omit them. This ends the proof of the proposition.
\end{proof}

\begin{remark}
    We give some final remarks on \Cref{prop:sys3d}.
    \begin{itemize}
        \item It is possible to prove points \eqref{list:inv} and \eqref{list:ell}
            also with a different approach. Indeed, after noting that the invariant
            $\mathcal{I}_1$ is trivial, one can use it to reduce immediately to a 
            two-dimensional system. It is then  straightforward  to show that the system is variational with a Lagrangian of the form:
            \begin{equation}
                L=\frac{\mu(x)}{2}(x')^2-V(x,\mathcal{I}_1).
            \end{equation}
            Through E.\ Noether's theorem~\cite{Noether1918} one immediately gets the other 
            invariant, and the statements follow.
        \item Using the construction in point \eqref{list:poisson}, it is possible
            to find multiple Hamiltonian functions for the system~\eqref{syst3Daut}. 
            That is, we have:
            \begin{equation}
                \vec{{x}}' = J_{a,j}^{(i)}\grad H_{j}^{(i)},
            \end{equation}
            where:
            \begin{equation}
                \begin{gathered}
                H_{1}^{(1)} = \frac{6}{a^2(a-1)\tilde{\kappa}} \mathcal{I}_2,
                \quad
                H_{1}^{(2)} = \frac{6}{a^2(a-1)\tilde{\kappa}} \log \mathcal{I}_2,
                \\    
                H_{2}^{(1)} = -\frac{6}{a^2(a-1)\tilde{\kappa}} \mathcal{I}_1,
                \quad
                H_{2}^{(2)} =
                -\frac{6}{a^2(a-1)\tilde{\kappa}} \frac{\mathcal{I}_1}{\mathcal{I}_2}.
                \end{gathered}
            \end{equation}
    \end{itemize}
\end{remark}

\subsection{Projection of the autonomous limit in 2D} 

We conclude this section showing that we can restrict the three-dimensional autonomous 
system~\eqref{syst3Daut} to an autonomous system in 2D. Indeed, the autonomous 
versions of the transformations \eqref{injectionfg} and \eqref{invinj}  are 
\begin{equation}
\begin{aligned}
    x&=\frac{\tilde{\kappa }(f-1) }{a-1},\\
    y&=\frac{
    \left[\begin{gathered}
        f^2(\tilde{n}+\tilde{\alpha}-g)^2
        +a(\tilde{\alpha}-g)(\tilde{\kappa}-g)
        -f(\tilde{n}+\tilde{\alpha}-g)(\tilde{n}(a+1)
        \\
        +\tilde{\alpha}(a+1)+\tilde{\beta}+a \tilde{\gamma}-(a+1)g)          
    \end{gathered}\right]}{a \tilde{\kappa}},
   \\
    z&=\frac{(f-1)
    \left[\begin{gathered}
    f^2(\tilde{n}+\tilde{\alpha}-g)^2
    -f(\tilde{n}+\tilde{\alpha}-g)((a-1)\tilde{n} 
    \\
    +a g(g-\tilde{\alpha})+a\tilde{\alpha}+(a-1)\tilde{\gamma}-(a+1)g)    
    \end{gathered}\right]
    }{\tilde{\kappa }(a-1) f }
\end{aligned}
\label{eq:redaut}
\end{equation}
and
\begin{equation}
    f=\frac{(a-1) x}{\tilde{\kappa }}+1, 
    \quad
    g=\tilde{n}+\tilde{\alpha}+\frac{\tilde{\kappa }((\tilde{\kappa}+(a-1)x )z-ax(\tilde{n}+y))}{(a-1) x \left(\tilde{\kappa }-x\right)},
    \label{eq:sysredaout}
\end{equation}
respectively. We stress the fact that they enstablish a birational equivalence between $S_a$ and the affine plane.  

The associated system is as follows:
\begin{equation}
    \begin{aligned}
        f' &=\frac{( \tilde{\beta}- \tilde{\alpha}+ \tilde{\gamma}+2 g) f^2
        +(( \tilde{\alpha}- \tilde{\beta}) a-2(a+1) g+ \tilde{\alpha}- \tilde{\gamma}) f-a ( \tilde{\alpha}-2 g)}{(a-1) a},
        \\
        g'&=\frac{(\tilde{n}+ \tilde{\alpha}-g) (\tilde{n}+ \tilde{\beta}+ \tilde{\gamma}+g) f^2-a g ( \tilde{\alpha}-g)}{(a-1)a f }.
    \end{aligned}
\end{equation}
The system keeps integrability, because the invariant $\mathcal{I}_1$ upon
substitution of~\eqref{eq:redaut} becomes (up to an inessential multiplicative
factor and an additive constant):
\begin{equation}
    \begin{aligned}
    (a-1)a\tilde{\mathcal{I}}_1
    &=
    (g-\tilde{n}- \tilde{\alpha})(\tilde{n}+ \tilde{\beta}+ \tilde{\gamma}+g)f
    -(a+1)g^2
    \\
    &+(a(\tilde{\alpha}-\tilde{\beta})+ \tilde{\alpha}- \tilde{\gamma})g+\frac{ag (g - \tilde{\alpha})}{f},    
    \end{aligned}
    \label{eq:redint}
\end{equation}
which is a first integral for the system~\eqref{eq:sysredaout}. Moreover,
the first integral \eqref{eq:redint} is a Hamiltonian function for the 
system~\eqref{eq:sysredaout} with respect to the same symplectic
form as the original symplectic structure, see equation \eqref{systfgham}.

\section{Conclusions}
\label{sec:concl}

 In this paper, we studied the 3D system of first-order differential
equations \eqref{syst3D}. This system is constrained on an invariant hypersurface
\eqref{eq:St} obtained from the results of \cite{MinChen} on the
asymptotic expansion as $z\to\infty$ of the coefficients of ladder
operators for degenerate Jacobi unitary polynomials.

We showed in \Cref{sec:firstpar,sec:hypersurf} that this system can be restricted to a 2D system of first-order differential
equations in two different ways. In both cases, we were able to reduce
these systems to the standard Hamiltonian form of the sixth Painlev\'e
equation. Then, we turned to the study of an autonomous version of the 3D system
and found out that it admits many properties akin to the equations of
Painlev\'e type. This justifies \Cref{conj:3d} about
the non-autonomous 3D system.

From our perspective, the main ideas behind our work are the following two: 
\begin{enumerate}
    \item in some cases, it can be easier to derive higher-dimensional
        differential equations from the theory of orthogonal polynomials,
        together with appropriate invariant surfaces, rather than looking
        directly for 2D systems;
    \item the relation between the obtained higher-dimensional systems and
        the Painle\-v\'e equations can be obtained through    parametrisation of the invariant surface followed by the application of the geometric theory of Painlev\'e equations.
\end{enumerate}

However, both points raise several further open and very interesting general questions,
which we will outline in the following paragraphs.

First of all, for our result it is crucial that we are able to parametrise
the hypersurface. As far we are aware, this is a very complicated problem in
algebraic geometry named the \textit{rationality problem}, see \cite{RationalityProbl}. For smooth surfaces in $\C^3$, there is a classical criterion,
known as the \emph{Castelnuovo criterion}~\cite{Castelnuovo1939}, which characterises
the rationality of the surface. Practical algorithms to construct the parametrisation
were introduced only recently in \cite{Schicho1999}. In the case of singular 
surfaces we are not aware neither of the existence of a general rationality criteria, nor
of algorithms capable of producing a parametrisation. These two are clearly challenging
problems that could be addressed in the context of algebraic geometry and computational
algebra. 

Another open problem, is whether the construction of a space of
initial conditions for the general 3D system \eqref{syst3D} is possible, see  \Cref{conj:3d}. 
Indeed, despite the success obtained with
autonomous limits, the problem in three dimensions is much harder. 
The difficulty
of this problem can be understood also by considering the step that we made
to obtain a parameterisation in \Cref{sec:hypersurf}. The choice of the 
``good'' compactification therein made was not evident at all. The same consideration
applies to the choice of compactification to make for the full system. From some
preliminary computations, we assess that choosing different compactifications
presents a sort of trade-off between the number of charts and the number of singularities.
To be more specific, compactifying the system~\eqref{syst3D} to
$\left(\Pj^1\right)^{\times3}$ gives rise
to a ``less singular'' system in comparison to the choice of compactifying it to $\Pj^{3}$. The picture is also complicated by the presence of indeterminacies of different
natures, i.e.\ points and potentially intersecting lines, rather than just points as in the two-dimensional
case. In \Cref{fig:sinrD} are depicted the configurations of the indeterminacies of the system with respect two different compactifications. Another complication is a consequence of the fact that the counterparts of minimal surfaces for threefolds are in general singular \cite{Reid}.
A complete geometric study of these singularities, their resolution, and the
precise relationship with the ones of the two-dimensional system will be the
object of further research. 

\begin{figure}[htb]
    \centering
    \begin{tikzpicture}[scale=1.7]
            
    \coordinate (p1) at (-1,-1,-1);
    \coordinate (p2) at (-1,1,1);
    \coordinate (p3) at (1,-1,1);
    \coordinate (p4) at (1,1,-1);
    
    \coordinate (q1) at (-1,-1,1);
    \coordinate (q2) at (-1,1,-1);
    \coordinate (q3) at (1,-1,-1);
    \coordinate (q4) at (1,1,1);
            
    \draw[thick,dashed] (p1)--(q3) -- (p3) -- (q1) -- (p2) -- (q2);
    \draw[thick,dashed] (p2) -- (q4) -- (p3);

    \draw[magenta,ultra thick] (q1) -- (p1);
    \draw[magenta,ultra thick] (q2)-- (p4) -- (q3);
    \draw[magenta,ultra thick] (p1) -- (-1,0.45,-1);
    \draw[magenta,ultra thick] (-1,0.55,-1) -- (q2);
    \draw[magenta,ultra thick] (q4) -- (p4);
    \draw[magenta,ultra thick] (-1/2,1,1) -- (-1/2,1,-1);
    \draw[magenta,ultra thick] (-1/2+0.3,1,1) -- (-1/2+0.3,1,-1);
    \draw[magenta,ultra thick] (-1,0,-1) -- (-1,0,1);

    \draw [teal,ultra thick] plot [smooth, tension=0.7] coordinates { (1,-1,1/2) (1,0,1/2) (p4)};
    \draw [teal,ultra thick] plot [smooth, tension=0.5] coordinates { (1,2/3,1) (1,1/2,1/2) (p4)};

    \node[anchor=south west] at (p1) {\small$(0,\infty,0)$};
    \draw[fill = black,draw=black] (p1) circle (1pt);
    \node[black,anchor=south east] at (p2) {\small$(0,0,\infty)$};
    \draw[fill = black,draw=black] (p2) circle (1pt);
    \node[black,anchor=north west] at (p3) {\small$(\infty,0,0)$};
    \draw[fill = black,draw=black] (p3) circle (1pt);
    \node[black,anchor=south west] at (p4) {\small$(\infty,\infty,\infty)$};
    \draw[fill = black,draw=black] (p4) circle (1pt); 
    \node[black,anchor=north east] at (q1) {\small$(0,0,0)$};
    \draw[fill = black,draw=black] (q1) circle (1pt);
    \node[black,anchor=south west] at (q2) {\small$(0,\infty,\infty)$};
    \draw[fill = black,draw=black] (q2) circle (1pt);
    \node[black,anchor=north west] at (q3) {\small$(\infty,\infty,0)$};
    \draw[fill = black,draw=black] (q3) circle (1pt);
    \node[black,anchor=south east] at (q4) {\small$(\infty,0,\infty)$};
    \draw[fill = black,draw=black] (q4) circle (1pt);
\end{tikzpicture}  
\quad
\tdplotsetmaincoords{170}{150}
\begin{tikzpicture}[scale=1.75,tdplot_main_coords]    
    \coordinate (e0) at (-1,1,-1);
    \coordinate (e1) at (1,1,1);
    \coordinate (e2) at (1,-1,-1);
    \coordinate (e3) at (-1,-1,1);

    \draw[dashed] (e1) -- (e2);
    \draw[dashed] (e0) --(e3) -- (e1);
    \draw[dashed] (e1) -- (e0);

    \newcommand{\td}{0.6}
    \coordinate (p0) at (1-2*\td, -1, -1+2*\td);
    \draw[cyan, ultra thick] (e3) -- (e2);
    \draw[cyan, ultra thick] (e2) -- ($0.65*(e2)+0.35*(e0)$);
    \draw[cyan, ultra thick] ($0.6*(e2)+0.4*(e0)$) -- (e0);
    \draw[cyan, ultra thick] (e0) -- (p0);
    \draw[cyan, ultra thick] (p0) --(e1);

    \coordinate (p1) at (-1, {1-2*\td}, {-1+2*\td});
    \draw[cyan, ultra thick] (e2) -- ($ 0.54*(e2) + 0.46*(p1) $);
    \draw[cyan, ultra thick] ($ 0.49*(e2) + 0.51*(p1) $) -- (p1);

    \node[black,anchor=south west] at (e0) {\small$e_0$};
    \draw[fill = black,draw=black] (e0) circle (1pt);    
    \node[black,anchor=south east] at (e1) {\small$e_1$};
    \draw[fill = black,draw=black] (e1) circle (1pt);    
    \node[black,anchor=north east] at (e2) {\small$e_2$};
    \draw[fill = black,draw=black] (e2) circle (1pt);    
    \node[black,anchor=north west] at (e3) {\small$e_3$};
    \draw[fill = black,draw=black] (e3) circle (1pt);
\end{tikzpicture}
    \caption{Configurations of the indeterminacies of the system with respect the compactifications $\left(\Pj^1\right)^{\times3}$ and $\Pj^3$.}
    \label{fig:sinrD}
\end{figure}
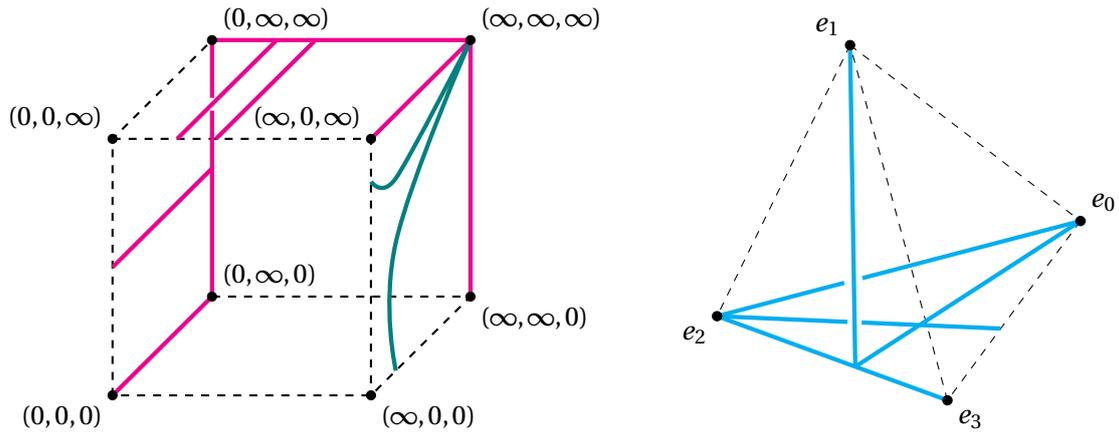

In general, we hope that this kind of research on the geometry of the 
three-di\-men\-sion\-al systems could pave the way to a theory similar to the one  of 
Okamoto-Sakai for two-dimensional systems. Indeed, a classification of Painlevé-like
equations in three or more components via spaces of initial conditions is currently lacking, 
though
sporadic examples are known in the literature, see for 
instance~\cite{noumiyamada1998,SuzukiFuji2012,FujiSuzuki2008,Cosgrove2000,Cosgrove2006Stud,Cosgrove2006JPhysA}.
The reason for 
this absence is that systems with three or more components exhibit a more
complicated and varied behavior than two-dimensional ones, as highlighted by
Chazy's equation, which possesses a movable singularity that serves as a
natural boundary for its solutions \cite{Chazy1909,Chazy1911}. In
particular, we expect that a geometric description of Painlevé-like
equations in higher dimensions may be used to classify differential 
equations reducible to Painlev\'e equations. For instance, treating equations 
similar to the $\sigma$-forms of the standard Painlevé equations. Looking at 
appropriate autonomous limits, as we did in \Cref{sec:autlim}, might help in the 
classification problem. 

We conclude by mentioning that another possible research direction is to 
determine if phenomena similar to the ones we unveiled in this paper can 
arise also in the case of discrete systems. For instance, given that the
notion of Darboux polynomials, and hence of Darboux surfaces, is available
also in the discrete setting~\cite{Celledoni_etal2019arXiv}, it would be 
interesting to see whether or not there are 3D non-autonomous discrete systems 
that can be put in relation with discrete Painlevé equations through restriction 
on Darboux surfaces.

\SkipTocEntry\section*{Acknowledgements}

 GF would like to thank Gleb Pogudin (Institute Polytechnique de Paris) for illuminating discussions.

 MG is member of the group GNSAGA of the Istituto
Nazionale di Alta Matematica   (INdAM)  from 2024.
 
\SkipTocEntry\section*{Statements and Declarations}

\subsection*{Competing interests}

On behalf of all authors, the corresponding author states that there is no conflict of interest.  

\subsection*{Funding}
 
GF acknowledges  support of the grant entitled ``Geometric approach to  ordinary differential equations'' (01.03.2023-29.02.2024) funded under New Ideas 3B  competition within Priority Research Area III    implemented under the “Excellence Initiative – Research University” (IDUB) Programme  (University of Warsaw) (nr  01/IDUB/2019/94). 
The work of GF is also partially supported by the project PID2021-124472NB-I00  funded by MCIN/AEI/10.13039/501100011033 and by  "ERDF A way of making Europe".

GG's research was partially supported by GNFM of the Istituto
Nazionale di Alta Matematica (INdAM), and the research project Mathematical Methods
in Non-Linear Physics (MMNLP) by the Commissione Scientifica Nazionale – Gruppo
4 – Fisica Teorica of the Istituto Nazionale di Fisica Nucleare (INFN). 

AS acknowledges the support of Japan Society for the Promotion of Science (JSPS) KAKENHI grant number 24K22843.

\subsection*{Author's contributions}
All authors equally contributed to the study conception, performed relevant computations  and analysed the results.  The first draft of the manuscript was written jointly and all authors commented on previous versions of the manuscript and edited it. All authors read and approved the final manuscript.

\subsection*{Data availability statement}
The manuscript has no associated or generated data.

\subsection*{Ethical conduct.}
The authors and the manuscript comply with all ethical standards and guidelines. The authors  provided correct citations as needed and have no conflict of interest.

\subsection*{Compliance with ethical standards}
The manuscript is not  submitted to more than one journal for simultaneous consideration. The submitted work is original and has not been published elsewhere in any form or language.  

\appendix

\section{Standard realisation of the space of initial conditions for the sixth Painlev\'e equation}
\label{appendix:standardP6}

In this appendix we recall the construction of the minimal space of initial conditions for $\pain{VI}$ coming from the standard model of Sakai surfaces of type $\D_4^{(1)}$, see \cite{KNY}.
Most of the contents were also provided in \cite[Appendix]{hypergeometric}, but we include this in order to make the present paper self-contained.
 
\subsection{Surfaces}

The usual construction of the space of initial conditions for system \eqref{systfg} proceeds as follows.
For the compactification, take $\p^1_{[f_0:f_1]} \times \p^1_{[g_0:g_1]}$ with $f,g$ providing an affine chart via $[1:f]=[f_0:f_1]$, $[1:g]=[g_0:g_1]$.
The standard atlas for $\p^1_{[f_0:f_1]} \times \p^1_{[g_0:g_1]}$ is then provided by the four charts $(f,g)$, $(F,g)$, $(f,G)$, $(F,G)$ defined by
\begin{equation}
    [f_0:f_1] = [1:f]=[F:1], \quad [g_0:g_1] = [1:g]=[G:1],  
\end{equation}
so $F = 1/f$ and $G=1/g$ on the overlaps of coordinate patches.
The eight indeterminacy points of the system are given in \eqref{eq:indetapp}, where we use arrows to indicate infinitely near points. 
In this case the points $b_4$ and $b_8$ lie on the exceptional divisors of the blow ups of $b_3$ and $b_7$ respectively. We express their locations according to the convention established in \Cref{subsec:spaceofinitialconditionssys3.2,subsec:soichyper} and given explicitly in equation \eqref{eq:coordblowp}.
\begin{equation}\label{eq:indetapp}
\begin{aligned}
    &b_1 : (F,g) = (0, -a_2),   & \quad   & \quad     &&b_2 : (F,g) = (0,-a_1 -a_2), \\
    &b_3 : (f,G) = (t, 0)      &\leftarrow& \quad    &&b_4 : (u_3,v_3) 
    = (t a_0, 0), \\  
    &b_5 : (f,g) = (0, 0),      & \quad   &  \quad    &&b_6 : (f,g) = (0,a_4), \\
    &b_7 : (f,G) = (1,0)       &\leftarrow& \quad     &&b_8  : (u_7,v_7) 
    = (a_3, 0).
\end{aligned}
\end{equation}

Blowing up these points, and denoting by $L_i$ the exceptional divisors over the point $b_i$, for $i=1,\ldots,8$, we get for fixed $\boldsymbol{a}=(a_0,a_1,a_2,a_3,a_4)$ the surface $X^{\operatorname{KNY}}_{t}$, as shown in the right hand side of  \Cref{fig:surface:systfg}. 
\begin{figure}[htb]
\centering
    \begin{tikzpicture}[basept/.style={circle, draw=black!100, fill=black!100, thick, inner sep=0pt,minimum size=1.2mm},scale=0.9]
    \begin{scope}[xshift = -4cm]
    \draw[thick] 
    (-2,-2)--(-2,2)
    (-2.5,-1.5)--(2.5,-1.5) 
    (2,-2)--(2,2) 
    (-2.5,1.5)--(2.5,1.5)
      ;     
      \node[left] at (-2.5,-1.5) {\tiny $g=0$};
      \node[left] at (-2.5,1.5) {\tiny $G=0$};     
      \node[below] at (-2,-2) {\tiny $f=0$};
      \node[below] at (2,-2) {\tiny $F=0$};

	 \node (b1) at (2,-1) [basept,label={[xshift=-7pt, yshift = -3 pt] \small $b_{1}$}] {};
	 \node (b2) at (2,+.5) [basept,label={[xshift=-7pt, yshift = -3 pt] \small $b_{2}$}] {};
	 \node (b3) at (-.8,+1.5) [basept,label={[xshift=-3pt, yshift = -18 pt] \small $b_{3}$}] {};
	 \node (b4) at (-.8,+2.5) [basept,label={[xshift=-7pt, yshift = -3 pt] \small $b_{4}$}] {};
	 \node (b5) at (-2,-1.5) [basept,label={[xshift=-7pt, yshift = -3 pt] \small $b_{5}$}] {};
	 \node (b6) at (-2,+0) [basept,label={[xshift=-7pt, yshift = -3 pt] \small $b_{6}$}] {};
	 \node (b7) at (+.8,+1.5) [basept,label={[xshift=-3pt, yshift = -18 pt] \small $b_{7}$}] {};
	 \node (b8) at (+.8,+2.5) [basept,label={[xshift=-7pt, yshift = -3 pt] \small $b_{8}$}] {};

    \draw[thick, ->] (b4) -- (b3);
    \draw[thick, ->] (b8) -- (b7);
    \node at (0,-2.5) {$\mathbb{P}^1\times \mathbb{P}^1$};

    \end{scope}

 \draw [->] (.8,0)--(-.8,0) node[pos=0.5, below] {$\text{Bl}_{b_1\dots b_{8}}$};

     \begin{scope}[xshift = +4.5cm]
    \draw[thick, blue] 
    (-2,-2)--(-2,2)
    (2,-2)--(2,2) 
    (-2.5,1.5)--(2.5,1.5)
      ;     
      \node[left, blue] at (-2.5,1.5) {\tiny $D_2$};     
      \node[below, blue] at (-2,-2) {\tiny $D_4$};
      \node[below, blue] at (2,-2) {\tiny $D_1$};

      \draw[red, thick] (1.25,-1) -- (2.75,-1) node[pos=1,xshift=+7pt,yshift=0pt] {\small $L_1$};
      \draw[red, thick] (1.25,+.5) -- (2.75,+.5) node[pos=1,xshift=+7pt,yshift=0pt] {\small $L_2$};
      \draw[blue, thick] (-.8,+1) -- (-.8,3) node[pos=0,xshift=-7pt,yshift=0pt] {\tiny $D_0$};
      \draw[red, thick] (-0.25,+2.5) -- (-1.5,2.5) node[pos=1,xshift=-10pt,yshift=0pt] {\small $L_4$};
      \draw[red, thick] (-2.75,-1.5) -- (-1.25,-1.5) node[pos=0,xshift=-7pt,yshift=0pt] {\small $L_5$};
      \draw[red, thick] (-2.75,0) -- (-1.25,0) node[pos=0,xshift=-7pt,yshift=0pt] {\small $L_6$};
      \draw[blue, thick] (+.8,+1) -- (+.8,3) node[pos=0,xshift=-7pt,yshift=0pt] {\tiny $D_3$};
      \draw[red, thick] (0.25,+2.5) -- (1.5,2.5) node[pos=1,xshift=+7pt,yshift=0pt] {\small $L_8$};
    \node at (0,-2.5) {$X^{\operatorname{KNY}}_{t}$};

    \end{scope}

\end{tikzpicture}
\caption{Graphic depiction of the construction of $X^{\operatorname{KNY}}$. In blue $(-2)$-curves, in red $(-1)$ curves.
}
\label{fig:surface:systfg}
\end{figure}
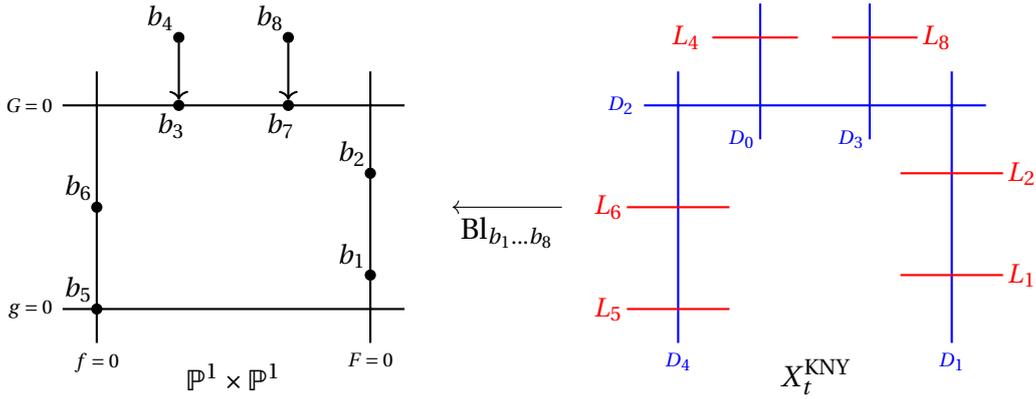

The unique effective anti-canonical divisor of the surface $X^{\operatorname{KNY}}_t$ has irreducible components $D_{0,t},D_{1,t}, D_{2,t} ,D_{3,t},D_{4,t}$,
where $D_{0,t}$ is the proper transform of $L_3$, $D_{1,t}$ is the proper transform of $\Set{f=\infty}$, $D_{2,t}$ is the proper transform of $\Set{g=\infty}$, $D_{3,t}$ is the proper transform of $L_7$ and $D_{4,t}$ is the proper transform of $\Set{f=0}$.  
This is the pole divisor of the rational two-form $\omega_t$ defined in the $(f,g)$-chart by $\omega_t = \frac{d_tf \wedge d_tg}{f}$.
We then get a bundle $E^{\operatorname{KNY}}$ over $B$ with fibre, over the point $t\in B$, the open surface $E^{\operatorname{KNY}}_t = X^{\operatorname{KNY}}_{t}\setminus D^{\operatorname{KNY}}_t$, where $D^{\operatorname{KNY}}_t = \bigcup_{i=0}^4D_{t,i}$ on which the system of differential equation \eqref{systfg} defines a uniform foliation.
Summarising all of this, we have the following.
\begin{lemma}
    We have a space of initial conditions for \eqref{systfg} in the sense of Definition \ref{def:spaceofinitialconditions} with
    \begin{itemize}
        \item $B=\C \setminus \Set{0,1} $ being the independent variable space of \eqref{systfg};
        \item $X=X^{\operatorname{KNY}}$ being the family of surfaces parametrised by $  B$, where $\pi_X^{-1}(t)=X_t^{\operatorname{KNY}}$, for $t\in B$;
        \item $D=D^{\operatorname{KNY}}$ being the locus in $X$ cut out by $D^{\operatorname{KNY}}_t$ in $X_t^{\operatorname{KNY}}$;
        \item $\varphi_X$ being the composition $\varphi_X|_{X_t^{\operatorname{KNY}}} : X^{\operatorname{KNY}}_t \rightarrow \p^1 \times \p^1 \dashrightarrow \C^2$, where the first map is the blow up of $b_1,\dots,b_8$ and the second corresponds to taking the affine coordinate chart $(f,g)$;
        \item $E=E^{\operatorname{KNY}}$ being the family of open surfaces $E^{\operatorname{KNY}}_t = X_t^{\operatorname{KNY}}\setminus D^{\operatorname{KNY}}_t$ over $B$, for $t\in B$.
    \end{itemize}
\end{lemma}

\subsection{Symplectic atlas and global Hamiltonian structure} \label{appsubsec:globalhamstructure}
If we introduce coordinates as usual on the blow up $\Bl_{b_1,\ldots,b_8}(\Pj^1\times\Pj^1)$, the resulting atlas for $E$ is neither Hamiltonian nor symplectic, which can be verified by direct calculation, see \Cref{subsec:HAMILTONIANSTRUC}.
We use instead the symplectic atlas constructed by Takano et al in \cite{takano1}. 
Rather than starting from the system \eqref{systfg} in coordinates $(f,g)$ and compactifying to $\p^1 \times \p^1$, this requires first changing variables to $(q,p) = (f,g/f)$.
The resulting system is, up to a relabelling of parameters, the Hamiltonian form of $\pain{VI}$ due to Okamoto \cite{OkamotoHams1}, explicitly 
\begin{equation}
    \label{systqpapp}
q'= \frac{\partial H^{\operatorname{Ok}}}{\partial p} , \quad p' = -\frac{\partial H^{\operatorname{Ok}}}{\partial q},
\end{equation}
where:
\begin{equation}
    \begin{aligned}
    H^{\operatorname{Ok}} &= \frac{q(q-1)(q-t)p^2  + \left( (a_1+2a_2) q (q-1) + a_3 (t-1)q + a_4 t (q-1) \right)p}{t(t-1)}
    \\
    &
     + \frac{a_2(a_1+1)q}{t(t-1)}.
    \end{aligned}
\end{equation}
We then have the following symplectic atlas for $E^{\operatorname{KNY}}$, which up to relabeling of the parameters is the same as that constructed by Takano et al \cite{takano1}.
\begin{proposition} \label{prop:takanosymplecticatlas}
The bundle    $E^{\operatorname{KNY}}$ admits the symplectic atlas $\mathcal{U} = \Set{ U_0,\ldots, U_5}$,
    where $U_i$ are coordinate patches with coordinates $(x_i,y_i;t)$, for $i=0,\ldots,5$, with gluing and relation to the original variables $f,g$ given as follows:
    \begin{equation}
        \begin{aligned}
            x_0 &= q = f, &\quad &y_0 = p = \frac{g}{f}, \\
            x_1 &= f ( g+a_2) = a_2 x_0 + y_0 , &\quad &y_1 = \frac{1}{f} = \frac{1}{x_0}, \\
            x_2 &= f ( g+a_1 + a_2) = (a_1+a_2) x_0 + y_0 , &\quad &y_2 = \frac{1}{f} = \frac{1}{x_0}, \\
            x_3 &= \frac{g\left( a_0 f + t g - f g \right)}{f^2} = \frac{y_0 \left(a_0 x_0^2 - x_0 y_0+ t y_0 \right)}{x_0^4}, &\quad &y_3 = \frac{f}{g} = \frac{x_0^2}{y_0}, \\
            x_4 &= \frac{g(a_4-g)}{f} = \frac{y_0 \left(a_4 x_0 -  y_0 \right)}{x_0^3}, &\quad &y_4 = \frac{f}{g} = \frac{x_0^2}{y_0}, \\
            x_5 &= \frac{g\left(g + a_3 f- f g \right)}{f^2} = \frac{y_0 \left(a_3 x_0^2 + y_0 - x_0 y_0  \right)}{x_0^4}, &\quad &y_5 = \frac{f}{g} = \frac{x_0^2}{y_0}.
        \end{aligned}
    \end{equation}
\end{proposition}

\begin{corollary} \label{cor:HamiltonianstructureKNY}
    The system on $E^{\operatorname{KNY}}$ extended from \eqref{systfg} has a global Hamiltonian structure with respect to the two-form $\omega^{\operatorname{KNY}}_t$ on the fibre given in the original coordinates $(f,g)$ by  
    \begin{equation}
        \omega^{\operatorname{KNY}}_t = \frac{d_t f\wedge d_t g}{f}. 
    \end{equation}
     This is defined by the two-form $\Omega^{\operatorname{KNY}}$ on $E^{\operatorname{KNY}}$ written in each chart $U_i \in \mathcal{U}$ of the symplectic atlas from Proposition \ref{prop:takanosymplecticatlas} as
     \begin{equation}
         \Omega^{\operatorname{KNY}} = dx_i \wedge dy_i + dH_i \wedge dt,
     \end{equation}
     where 
     \begin{equation}
        \begin{aligned}
         H_0(x_0,y_0;t) &= \frac{x_0 (x_0-1)(x_0-t)y_0^2 +a_2(a_1+a_2)x_0}{t(t-1)}
         \\
         & + \frac{\left( x_0^2(a_1+2a_2)+ x_0 t (a_3+a_4) -x_0(a_1+2a_2-a_3)- t a_4\right) y_0}{t(t-1)},
        \end{aligned}
     \end{equation}
     and the remaining $H_i$ are determined according to Lemma \ref{symplecticlemma}.
\end{corollary}

\section{Relation of the 3D system with the \texorpdfstring{$\sigma$}{}-form of the \texorpdfstring{$\pain{VI}$}{} equation}
\label{app:sigma}

In \cite[Prop. 3.2, Th. 3.3, Th. 3.4]{MinChen} a connection between the solution of the system \eqref{syst3D} coming from the degenerate Jacobi weight and the $\sigma$-form of the sixth Painlev\'e equation was established. 
The authors present only how $r_n(t)$ and $y_n(t)$ (that is, $z(t)$ and $y(t)$ in our notation) are related to an auxiliary function $H(t)$ and its derivatives (see formulas in \cite[Prop. 3.2]{MinChen}). 
This function satisfies a non-linear second-order differential equation of degree two \cite[Th. 3.3]{MinChen} which  is then  related to the $\sigma$-form of the sixth Painlev\'e equation \cite[Th. 3.4]{MinChen} via an affine transformation. We give more details below following \cite{MinChen}.

The $\sigma$-form of the sixth Painlev\'e equation is given by \cite[p. 346]{OkamotoP6}
\begin{equation} \label{sigma}
\begin{aligned}
&
\sigma'(t(t-1)\sigma'')^2+(2\sigma'(t\sigma'-\sigma)-(\sigma')^2-\nu_1\nu_2\nu_3\nu_4)^2\\
&\quad \quad -(\sigma'+\nu_1^2)(\sigma'+\nu_2^2)(\sigma'+\nu_3^2)(\sigma'+\nu_4^2)=0,
\end{aligned}
\end{equation}
where $\sigma=\sigma(t)$, $\sigma'=\tfrac{d}{dt}\sigma$, and $\nu_1,\dots,\nu_4$ are parameters.
In our notation, the parameters of the $\sigma$-form and the parameters in the weight are related by \cite{MinChen}:
\begin{equation}
    \nu_1=\frac{\alpha+\beta}{2},\;\;\nu_2=\frac{\beta-\alpha}{2},\;\;\nu_3=\frac{2n+\alpha+\beta}{2},\;\;\nu_4=\frac{2n+\alpha+\beta+2\gamma}{2}.
\end{equation}
The function $H=H(t)$ defined by
\begin{equation}\label{sigmatoH}
    \sigma=H - \left(n(n+\alpha+\beta+\gamma) + \frac{(\alpha+\beta)^2}{4}\right)t + \frac{2n(n+\alpha+\beta+\gamma)+(\alpha+\beta)\beta-(\alpha-\beta)\gamma}{4}
\end{equation}
satisfies a second-order second-degree differential equation \cite[Th. 3.3]{MinChen} of the form
\begin{equation} \label{seconddegreeH}
    \mathfrak{H}(H,H',H'';t) = 0,
\end{equation}
where $\mathfrak{H}$ is a polynomial functions of its arguments, of degree two in $H''$, which we do not
reproduce here for the sake of brevity. 
This is equivalent to the $\sigma$-form \eqref{sigma} via the relation \eqref{sigmatoH}.

If one differentiates  equation \eqref{seconddegreeH} with respect to $t$, the result is a non-autonomous third-order differential equation for $H$ of  degree one. 
We rewrite this as a first-order system by taking $(H_0,H_1,H_2)=(H,H',
H'')$, so $H_0'=H_1,\;H_1'=H_2$, and $H_2'=H'''$ is known from the third-order equation. 
We shall refer to this as the 3D system for $H$, and similarly refer to the differentiated $\sigma$-form, recast as a first-order system, as the 3D system for $\sigma$. 
For the sake of readability we do not present either of these here.

By construction the 3D system for $H$ has the first integral $\mathfrak{H}(H_0,H_1,H_2;t)$, with $\mathfrak{H}$ as in \eqref{seconddegreeH}, and the level sets 
\begin{equation}
\Gamma_c \defeq \Set{\mathfrak{H}(H_0,H_1,H_2;t)=c } \subset \mathbb{C}_t^{3},  \quad c \in \mathbb{C},
\end{equation} 
are invariant hypersurfaces for the 3D system for $H$.
The second-degree equation \eqref{seconddegreeH} corresponds to the restriction of the 3D system for $H$ to the particular level set $\Gamma_0$.

We will show that this restriction of the 3D system for $H$ to $\Gamma_0$, or equivalently the 3D system for $\sigma$ subject to \eqref{sigma}, can be identified with the system \eqref{syst3D} subject to the constraint \eqref{hypersurfaceconstraint}.

As remarked above, \cite{MinChen} provides expressions for $y$ and $z$ in terms of $H$ and its derivatives, so we just require the expression for $x$, which was not given in \cite{MinChen}. 
This can be obtained from the 3D system for $H$ in addition to the second-order second-degree equation \eqref{seconddegreeH}. 
However, the computations and the final expression is extremely cumbersome, so we only outline the steps involved so that the results can be reproduced in any computer algebra system. 
One needs to first substitute the expressions 
\begin{equation} \label{zyfromMinChen}
    z = \frac{H - t H' + n(n+\alpha+\gamma)}{2n+\alpha+\beta+\gamma}, \quad 
    y = \frac{H - (t-1) H' - n(n+\beta+\gamma)}{2n+\alpha+\beta+\gamma},
\end{equation}
from \cite[Prop. 3.2]{MinChen} into the system \eqref{syst3D}.
From the second of the resulting equations, one can obtain $x^2$  in terms of $x$, $H$, $H'$ and $H''$. 
Then one differentiates this expression with respect to $t$ and substitutes $H'''$ as a rational function of $H,H',H''$ according to the third-order equation for $H$.
Then restriction to the level set $\Gamma_0$ is imposed by replacing $(H'')^2$ by a polynomial in $H,H'$ with coefficients rational in $t$ according to \eqref{seconddegreeH}. 
Finally, using the first equation of the system \eqref{syst3D} for $x'$, with $y$ and $z$ replaced according to \eqref{zyfromMinChen}, and  replacing the powers of $x$ using the known expression for $x^2$ one can finally find the cumbersome rational expression for $x$ in terms of $H,H',H''$. 
It can be checked that this, together with the expressions \eqref{zyfromMinChen} for $y$ and $z$, provides a birational transformation between the copies of $\C^3$ with coordinates $x,\,y,\,z$ and 
$H_0, \,H_1,\,H_2$ respectively. 
This restricts to a birational transformation between the hypersurface \eqref{eq:St} and the particular level set $\Gamma_0$. 
Under this restriction the 3D system \eqref{syst3D} on the hypersurface \eqref{eq:St} is transformed to the 3D system for $H$ on $\Gamma_0$, or similarly via \eqref{sigmatoH} to the 3D system for $\sigma$ subject to \eqref{sigma}.
This provides an additional description of the connection to the $\sigma$-form of $\pain{VI}$ identified in \cite{MinChen}.

\section{An apparent singularity of the system in the second parametrisation with \texorpdfstring{$\p^1\times\p^1$}{} compactification} \label{app:apparentsingularity}
In this appendix we exhibit how the choice of the compactification might lead to  complications in the identification procedure.

Consider the system \eqref{second_system}, and choose as compactification $\p^1 \times \p^1$ rather than $\BF_1$ as was done in \Cref{subsec:soichyper}.
With this choice, we find the following eight indeterminacy points given in coordinates as follows:  
\begin{equation}
    \begin{gathered}
c_1=([0:1],[1:\kappa\beta]),\quad 
c_2=([\kappa:1],[1:-n\kappa t]),\;\\
c_3 =([\kappa:1],[1:-(n+\gamma)\kappa t]),\quad  
c_4 =([\kappa:1-t],[1-t:\alpha  \kappa t]), \;\\
c_5 =([0:1],[1:0]),\quad
c_6 =([\kappa:1-t],[1:0]),\;\\ 
c_7=([1:0],[0:1]),\quad
c_8 =([1:0],[1:-n(n+\gamma)t]),
\end{gathered}
\end{equation}
where, the standard affine chart centred at $([0:1],[0:1])$ has coordinates $(x_0,v)$.
 All of them but two, namely $c_7,\,c_8$, require only one blow up to resolve, i.e. that the system is free of indeterminacy points on the exceptional divisors of these blow ups introduced according to \eqref{eq:coordblowp}.
 The indeterminacy point $c_7$ is resolved after three blow ups in total, first of $c_7$ then of two further indeterminacy points given in coordinates by 
 \[
 c_9:\ (U_7,\,V_7)=(0,\,0), \quad c_{10}:\ (u_9,\,v_9)=\left(\frac{1}{1-t},\,0\right).
 \]
After blowing   up $c_8$,  if one proceeds to then repeatedly look for further indeterminacy points on the exceptional divisor and blow these up, one observes that this procedure does not terminate after a reasonable number of iterations. 
In fact, the indeterminacy point $c_8$ constitutes an \emph{apparent singularity} and should not be blown up at all, as we shall explain next.
We give a schematic representation of the surface obtained by the blowing up $c_1,\ldots,c_7,c_9,c_{10}$ and indicate the indeterminacy point $c_8$ in \Cref{fig:surface:secondsystem}.

Consider a non-autonomous 2D system of the form \eqref{systqpdef} as a rational vector field on $\mathbb{C}^2\times B$, and assume that this has the Painlev\'e property in the sense that it has a space of initial conditions.
Then, choose some compactification $\Sigma$ of the fibre $\mathbb{C}^2$ and consider the system as a rational vector field on $\Sigma \times B$.
The blow ups that must be performed in order to achieve a uniform foliation are centred at the points in the fibre $\Sigma$ through which infinitely many solutions pass.
Such points are always indeterminacy points of the system but the converse is not true.

It might be possible that an indeterminacy point does not require blowing up. Two possible reasons for this follow. The first is that there are actually no solutions passing through the point and it should be removed as part of the inaccessible divisors, see \cite{Stud} for such a case. 
The second is that there is only one solution passing through the point.  
The latter is what is happening in the case of $c_8$, which can be seen as follows.

Consider the proper transform on $\operatorname{Bl}_{c_1\dots c_{10}} \mathbb{P}^1 \times\mathbb{P}^1$ of the line $x_0=\infty$, on which the problematic indeterminacy point $\operatorname{Bl}_{c_1\dots c_{10}}^{-1}(c_8)$ lies. 
This is a $(-1)$-curve and we can contract it to a point, say $p$, at the origin of the chart with coordinates $(x,y)$ given by $x= \frac{1}{x_0 v}$, $y=\frac{1}{x_0}$. 
Pushing the rational vector field forward under the contraction to one in local coordinates $(x,y)$, we see that it becomes regular at $p$.
Then existence and uniqueness theorems for ordinary differential equations give a unique solution in $(x,y)$ coordinates passing through this point for a fixed $t_*$, which when lifted back up to $\operatorname{Bl}_{c_1\dots c_{10}} \mathbb{P}^1\times\mathbb{P}^1$ must be the single solution passing through the point $c_8$.
Indeed it can be verified by direct calculation in this case that the unique solution of the regular initial value problem for the system \eqref{second_system} transformed to coordinates $(x,y)=\left(\frac{1}{x_0 v},\frac{1}{x_0}\right)$ with initial condition $(x(t_*),y(t_*))=(0,0)$ corresponds to one passing through $c_8$ at $t=t_*$.

We must remark that care must be taken of apparent singularities such as $c_8$ only when performing blow ups based purely on consideration of indeterminacy points.

In his original work, Okamoto have chosen the centres of blow up by studying infinite families of solutions passing trough points. In this way there would be no risk of confusion regarding what needs to be blown up, but this requires a preliminary analysis.     

\begin{figure}[htb]
\centering
    \begin{tikzpicture}[basept/.style={circle, draw=black!100, fill=black!100, thick, inner sep=0pt,minimum size=1.2mm}]
    \begin{scope}[xshift = -4cm]
    \node at (0,2.5) {$\mathbb{P}^1\times \mathbb{P}^1$};
    \draw[thick] 
    (-2,-2)--(-2,2)
    (-2.5,-1.5)--(2.5,-1.5) 
    (2,-2)--(2,2) 
    (-2.5,1.5)--(2.5,1.5)
      (-.66,-2)--(-.66,2)  
      (.66,-2)--(.66,2)  
      ;     
      \node[left] at (-2.5,-1.5) {\tiny $v=0$};
      \node[left] at (-2.5,1.5) {\tiny $v=\infty$};     
      \node[below] at (-2,-2) {\tiny $x_0=0$};
      \node[below] at (2,-2) {\tiny $x_0=\infty$};
      \node[below] at (-.66,-2) {\tiny $x_0=\kappa$};
      \node[below] at (.66,-2) {\tiny $x_0=\frac{\kappa}{1-t}$};
      
	 \node (A1) at (-2,-1) [basept,label={[xshift=-7pt, yshift = -3 pt] \small $c_{1}$}] {};
	 \node (A5) at (-2,+1.5) [basept,label={[xshift=-7pt, yshift = -3 pt] \small $c_{5}$}] {};
	 \node (A2) at (-.66,-.5) [basept,label={[xshift=-7pt, yshift = -3 pt] \small $c_{2}$}] {};
	 \node (A3) at (-.66,+.75) [basept,label={[xshift=-7pt, yshift = -3 pt] \small $c_{3}$}] {};	 
      \node (A4) at (+.66,.25) [basept,label={[xshift=-7pt, yshift = -3 pt] \small $c_{4}$}] {};
	 \node (A6) at (+.66,+1.5) [basept,label={[xshift=-7pt, yshift = -3 pt] \small $c_{6}$}] {};
    \node (A7) at (2,-1.5) [basept,label={[xshift=-7pt, yshift = -3 pt] \small $c_{7}$}] {};
    \node (A8) at (2,-.25) [basept,magenta,label={[xshift=-7pt, yshift = -3 pt] \small $c_{8}$}] {};

	 \node (A9) at (2.7,-2.25) [basept,label={[xshift=+5pt, yshift = 0 pt] \small $c_{9}$}] {};

    \node (A10) at (3.2,-2.25) [basept,label={[xshift=+7pt, yshift = 0 pt] \small $c_{10}$}] {};
    \draw[thick, ->] (A9) --(2.3,-1.5)-- (A7);
    \draw[thick, ->] (A10) -- (A9);

    \end{scope}

 \draw [->] (.8,0)--(-.8,0) node[pos=0.5, below] {$\text{Bl}_{c_1\dots c_{7}}$};
 
    \begin{scope}[xshift = +4cm]
    \draw[thick,blue] 
    (-2,-1.75)--(-2,1);
     \draw[thick,red]    
    (-2.25,-1.5)--(1.625,-1.5);
        \draw[thick,red] 
    (2,-1.125)--(2,1.75) ;
    \draw[thick, blue] 
    (-1.5,1.5)--(2.25,1.5);
    \draw[thick,blue] 
      (-.66,-1.75)--(-.66,1.75);  
    \draw[thick,blue] 
      (.66,-1.75)--(.66,1)  ;
    \draw[red, thick] (-2.5,-1) -- (-1.5,-1) node[pos=0,xshift=-5pt,yshift=-5pt] {\small $C_1$}; 
    \draw[red, thick] (-1.0,1.75) --(-2.25,0.5) node[pos=1,xshift=-5pt,yshift=-5pt] {\small $C_5$};

    \draw[red, thick] (-1.16,-.5) -- (-.16,-.5) node[pos=0,xshift=-5pt,yshift=-5pt] {\small $C_2$};
    \draw[red, thick] (-1.16,+.75) -- (-.16,+.75) node[pos=0,xshift=-5pt,yshift=-5pt] {\small $C_3$};
    \draw[red, thick] (+.16,+.25) -- (+1.16,+.25) node[pos=0,xshift=-5pt,yshift=-5pt] {\small $C_4$};
    \draw[red, thick] (0.41,0.5) -- (+1.66,1.75)  node[pos=1,xshift=-7pt,yshift=+5pt] {\small $C_6$};      
    
    
    \draw[red, thick] (1.0,-1.75) --(2.25,-0.5);      
    \node (A8) at (2,-.25) [basept,magenta,label={[xshift=-7pt, yshift = -3 pt] \small $c_{8}$}] {};

      \node (A9) at (1.25,-1.5) [basept,label={[xshift=-7pt, yshift = +0 pt] \small $c_{9}$}] {};
      \node (A10) at (1.75,-2.25) [basept,label={[xshift=+7pt, yshift = -3 pt] \small $c_{10}$}] {}; 
    \draw[thick, ->] (A10) -- (A9);
    \end{scope}

    	\draw [->] (4,-3.5)--(4,-2.7) node[pos=0.5, left] {$\text{Bl}_{c_9}$};

    \begin{scope}[xshift = +4cm, yshift=-6cm]
    \draw[thick,blue] 
    (-2,-1.75)--(-2,1);
     \draw[thick,blue]    
    (-2.25,-1.5)--(1.4,-1.5);
        \draw[thick,red] 
    (2,0)--(2,1.75) ;
    \draw[thick, blue] 
    (-1.5,1.5)--(2.25,1.5);
    \draw[thick,blue] 
      (-.66,-1.75)--(-.66,1.75);  
    \draw[thick,blue] 
      (.66,-1.75)--(.66,1)   ;
    \draw[red, thick] (-2.5,-1) -- (-1.5,-1) node[pos=0,xshift=-5pt,yshift=-5pt] {\small $C_1$};
    \draw[red, thick] (-1.0,1.75) --(-2.25,0.5) node[pos=1,xshift=-5pt,yshift=-5pt] {\small $C_5$};

    \draw[red, thick] (-1.16,-.5) -- (-.16,-.5) node[pos=0,xshift=-5pt,yshift=-5pt] {\small $C_2$};
    \draw[red, thick] (-1.16,+.75) -- (-.16,+.75) node[pos=0,xshift=-5pt,yshift=-5pt] {\small $C_3$};
    \draw[red, thick] (+.16,+.25) -- (+1.16,+.25) node[pos=0,xshift=-5pt,yshift=-5pt] {\small $C_4$};
    \draw[red, thick] (0.41,0.5) -- (+1.66,1.75)  node[pos=1,xshift=-7pt,yshift=+5pt] {\small $C_6$};      
    
    
    \draw[blue, thick] (1.0,-0.75) --(2.25,+0.5);      
    \node (A8) at (2,.75) [basept,magenta,label={[xshift=-7pt, yshift = -3 pt] \small $c_{8}$}] {};
    \draw[red, thick] (+1.15,-1.75) -- (+1.15,-.4);
      \node (A10) at (1.15,-1.1) [basept,label={[xshift=+10pt, yshift = -5 pt] \small $c_{10}$}] {};
    \end{scope}

 \draw [->] (.8,-6)--(-.8,-6) node[pos=0.5, below] {$\text{Bl}_{c_{10}}$};

\draw [->] (-4,-3.5)--(-4,-2.7) node[pos=0.5, left] {$\text{Bl}_{c_1\dots c_7 c_{9} c_{10}}$};

    \begin{scope}[xshift = -4cm, yshift=-6cm]
    \draw[thick,blue] 
    (-2,-1.75)--(-2,1);
     \draw[thick,blue]    
    (-2.25,-1.5)--(1.4,-1.5);
        \draw[thick,red] 
    (2, 0)--(2,1.75) ;
    \draw[thick, blue] 
    (-1.5,1.5)--(2.25,1.5);
    \draw[thick,blue] 
      (-.66,-1.75)--(-.66,1.75);  
    \draw[thick,blue] 
      (.66,-1.75)--(.66,1)  ;
    \draw[red, thick] (-2.5,-1) -- (-1.5,-1) node[pos=0,xshift=-5pt,yshift=-5pt] {\small $C_1$};
    \draw[red, thick] (-1.0,1.75) --(-2.25,0.5) node[pos=1,xshift=-5pt,yshift=-5pt] {\small $C_5$};

    \draw[red, thick] (-1.16,-.5) -- (-.16,-.5) node[pos=0,xshift=-5pt,yshift=-5pt] {\small $C_2$};
    \draw[red, thick] (-1.16,+.75) -- (-.16,+.75) node[pos=0,xshift=-5pt,yshift=-5pt] {\small $C_3$};
    \draw[red, thick] (+.16,+.25) -- (+1.16,+.25) node[pos=0,xshift=-5pt,yshift=-5pt] {\small $C_4$};
    \draw[red, thick] (0.41,0.5) -- (+1.66,1.75)  node[pos=1,xshift=-7pt,yshift=+5pt] {\small $C_6$};      
    
    
    \draw[blue, thick] (1.0,-0.75) --(2.25,+0.5);      
    \node (A8) at (2,.75) [basept,magenta,label={[xshift=-7pt, yshift = -3 pt] \small $c_{8}$}] {};
    \draw[blue, thick] (+1.15,-1.75) -- (+1.15,-.4);
    \draw[red, thick] (1,-1.1) -- (2, -1.1)    node[pos=1,xshift=+7pt,yshift=+5pt] {\small $C_{10}$};
    \end{scope}
\end{tikzpicture}
\caption{Sequence of blow ups and surface for system \eqref{second_system} with $\p^1 \times \p^1$ compactification.
In red $(-1)$-curves and in blue $(-2)$-curves.
}
\label{fig:surface:secondsystem}
\end{figure}
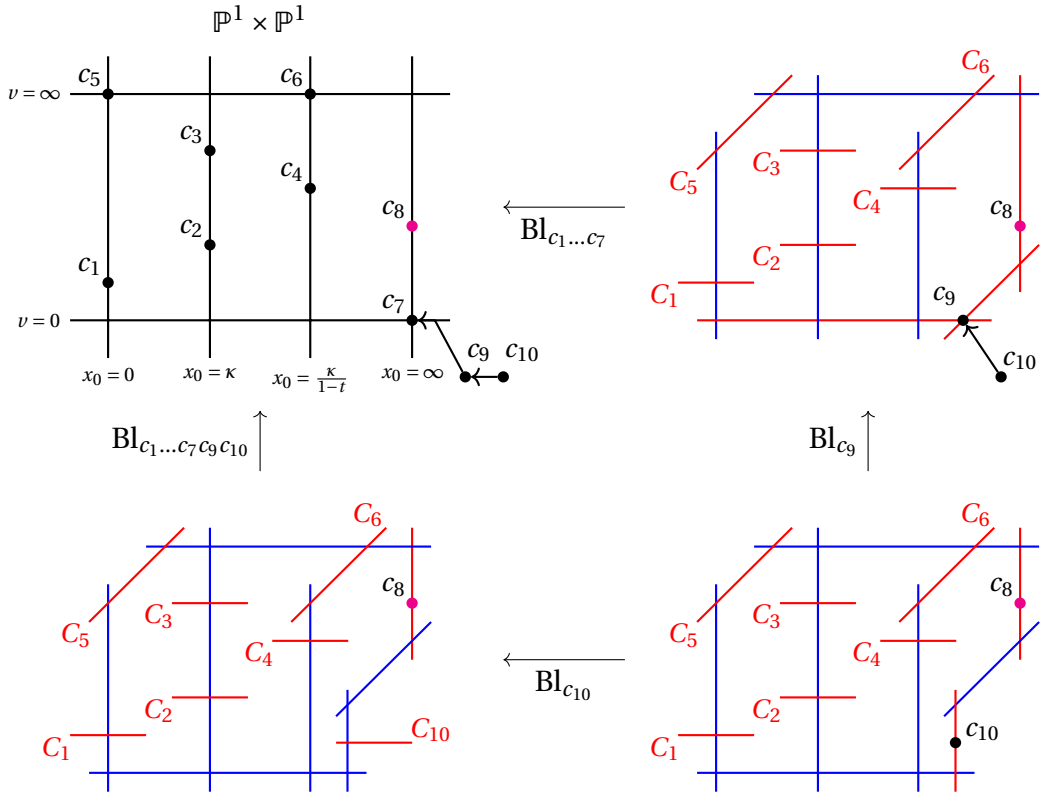 

\section{Additional autonomous limits for the 3D system}

\label{app:clim}

In this appendix we discuss the exponential-type autonomous
limits for the system~\eqref{syst3D}. We recall that, following \cite{JoshiP6},
this autonomous limits are obtained through the change of independent variable:
\begin{equation}
    t \to e^\tau + t_0, \quad t_0\in\Set{0,1},
    \label{eq:cvlim}
\end{equation}
then taking $e^\tau\to 0$.

\subsection{Exponential type autonomous limit as $t\to0$}

Consider the case $t_0=0$ in \eqref{eq:cvlim}. Then, in the limit 
above, system \eqref{syst3D} collapses to the following system of 
algebro-differential equations:
\begin{equation}\label{syst limit 0}
    \begin{aligned}
        x'=&(\kappa -x) (\beta -x+2 z),  \\
   z' =&\frac{z (\kappa -x) (\beta +z) (\kappa+x-2)}{x (\kappa-1)}=0,
    \end{aligned}
\end{equation} 
where differentiation is with respect to $\tau$. 
Note that we had to multiply  both sides  of the second equation of the non-autonomous system by $e^{\tau}$ and then let $e^{\tau}\to 0$. We see that the variable $y$ completely disappeared.

Then $z$ must be constant $z=z_*$ and the hypersurface
$S_t$ degenerates to a double plane:
\begin{equation}\label{S0}
S_0: \; z_*(z_*+\beta)(\kappa -x )^2=0,
\end{equation}
which is consistent with $S_a$ from the first subsection when $a=0$  and $\varepsilon\to 0$ by taking \eqref{xyz scaling} and \eqref{par scaling}, that is, with  \eqref{S0 limit}.

In the end, we see that this autonomous system is trivial, in the sense that it admits 
only the trivial solutions:
\begin{equation}
    \begin{gathered}
        x = \kappa,
        \quad z = z_*,
        \\
        x = -\alpha-\beta-\gamma-2n+1,
        \quad
        z = -\frac{1}{2}\alpha-\beta-\frac{1}{2}\gamma-n+\frac{1}{2},
        \\
        x = \frac{\kappa e^{(\beta-\kappa)(\lambda+\tau)}-\beta}{e^{(\beta-\kappa)(\lambda+\tau)}-1},
        \quad z_* = 0,
        \\
        x= -\frac{\beta e^{(\beta+\kappa)(\lambda+\tau)}+\kappa}{e^{(\beta+\kappa)(\lambda+\tau)}-1},
        \quad 
        z_* = - \beta,
    \end{gathered}
\end{equation}
where $\lambda$ is an arbitrary constant.

\subsection{Exponential type autonomous limit as $t\to1$}

Consider now $t_0=1$ in \eqref{eq:cvlim}. Then, in the limit described 
above the system \eqref{syst3D} becomes the following:
\begin{equation}
    \label{eq:sys3dcollaps}
\begin{aligned}
x'&=  x (\kappa -\alpha-1 +2 y)-\kappa  (\beta +2 z),\;\quad y'=0,  \\
(\kappa-1)xz'&= x^2 y (y-\alpha )
   + x\left(n  
   (\gamma +2 y-2 z+n) + y (\alpha +\gamma -2 z)-z (\beta
   +\gamma )\right)
   \\& \quad -z (\beta +z)  \left((  \kappa-1)^2-1\right),
\end{aligned}
\end{equation} 
where the differentiation is with respect to $\tau$. From the second equation of \eqref{eq:sys3dcollaps}, we see that $y$ is constant. This system has the following invariant algebraic hypersurface:
\begin{equation}
    S_1=\Set{h_1(x,y,z)=0},
\end{equation}
where
\begin{equation}
    h_1 = \kappa  x \left(n^2+n (\gamma +2 y-2 z)+y (\alpha +\gamma -2 z)-z (\beta +\gamma )\right)+x^2 y (y-\alpha )+\kappa ^2
   z (\beta +z).
\end{equation}

Note that the   hypersurface $S_1$ is obtained from the hypersurface $S_a$ in
\eqref{Sa} as $a\to1$, and $\varepsilon\to 0$ by taking \eqref{xyz scaling} and \eqref{par scaling}. Division by $x$ of the defining polynomial of the hypersurface does not give the first integral as before.

Let us now discuss the properties of the solutions of the system 
\eqref{eq:sys3dcollaps}. We observe, that using the constancy of $y$ we 
obtain a single second-order equation for $x$. This equation has the following
form:
\begin{equation}
    x'' +\frac{ (\kappa-2)(x')^2}{2  (\kappa-1)x}+V_3'(x)=0,
    \label{eq:redsys2}
\end{equation}
where $V_3$ is an assigned function, of which we omit the explicit expression 
for the sake of brevity.
From its shape, it is evident  that equation~\eqref{eq:redsys2} has the
first integral $\mathcal{I}$ defined by
\begin{equation}
\begin{gathered}
 x^{1-\frac{1}{\alpha +\beta +\gamma +2 n}} \mathcal{I}=  \beta ^2 \kappa^2+x^2 \left((\beta +\gamma )^2+4 n(n+\beta +\gamma )+4 y
   (\kappa-1)\right)\\-2 x \kappa \left(\beta  (\beta +\gamma ) +2 n (n+\beta
   +\gamma )+2 y (\kappa-1)\right)  -(x')^2.
\end{gathered}
\end{equation}
For generic values of the parameters the level curves of the function
$\mathcal{I}$ are not algebraic curves. For this reason, in general, the 
system~\eqref{eq:redsys2} will not possess the Painlev\'e property.

\end{document}